\documentclass[final,authoryear,12pt]{elsarticle}%
\usepackage[utf8]{inputenc}
\usepackage{natbib}
\usepackage{booktabs}
\usepackage{amsmath,amsfonts,amssymb,amsthm,enumitem}
\usepackage{graphicx}
\usepackage{rotating}
\usepackage{pdflscape}
\usepackage{caption}
\usepackage{subcaption}
\usepackage{multirow}
\usepackage{amsthm}
\usepackage{pgf}
\usepackage{setspace}
\usepackage[margin=1in]{geometry}
\usepackage{hyperref}
\usepackage{bbm}
\usepackage{float}
\usepackage{comment}
\usepackage{float}
\usepackage{tikz,calc}
\usepackage{accents}
\usepackage{amsmath}
\usepackage{amsfonts}
\usepackage{xcolor}
\usepackage{soul}
\usepackage{amssymb}
\usepackage{pgfplots}%
\usepackage{algorithm}
\usepackage{algpseudocode}
\providecommand{\U}[1]{\protect\rule{.1in}{.1in}}
\pgfplotsset{compat=1.18}
\makeatletter
\def\ps@pprintTitle{ \let\@oddhead\@empty
 \let\@evenhead\@empty
 \def\@oddfoot{} \let\@evenfoot\@oddfoot}
\makeatother
\usetikzlibrary{chains}
\usetikzlibrary{positioning}
\newtheorem{assumption}{Assumption}
\newtheorem{theorem}{Theorem}[section]

\newtheorem{lemma}{Lemma}[section]

\newtheorem{proposition}{Proposition}[section]
\newtheorem{remark}{Remark}[section]

\newcommand{\E}{\operatorname{E}}
\providecommand{\Op}{O_p}
\providecommand{\op}{o_p}

\begin{document}

\begin{frontmatter}
\title{Estimation and Inference for Latent Dual Networks Using High-Dimensional IV Screening\tnoteref{acknowledgements}}
\tnotetext[acknowledgements]{Acknowledgements: We would like to thank \'Aureo De Paula, Frank Kleibergen, Hashem Pesaran, Martin Weidner and seminar participants and attendees at the Tinbergen Institute, UC Davis, CFE 2025, the 2025 Conference on Sustainable Finance, EcoSta 2025, the 2025 Bari Panel and High Dimensional Data Conference, and the 2026 EFiC Conference for helpful comments and suggestions. Any remaining errors are our own.}
\author[1]{Art\={u}ras Juodis}
\ead{a.juodis@uva.nl}
\author[2]{George Kapetanios\corref{cor1}}
\ead{george.kapetanios@kcl.ac.uk}
\author[3]{Vasilis Sarafidis\corref{cor1}}
\ead{vasilis.sarafidis@brunel.ac.uk}
\cortext[cor1]{Corresponding authors.}
\affiliation[1]{University of Amsterdam}
\affiliation[2]{King's College London}
\affiliation[3]{Brunel University London}
\begin{abstract}
\singlespacing
\noindent
We develop a novel methodology for estimation and inference in high-dimensional panel network models with latent dual structures. The framework allows outcomes to be affected simultaneously by positive and negative interaction channels, accommodating settings in which some interactions reinforce outcomes while others generate competition and displacement effects.  The proposed method identifies and estimates the network directly from the structural model using observed data without the need to pre-specify the network. Network recovery is achieved through a sequential instrumental-variable screening procedure. We establish exact support recovery and oracle-equivalent post-selection inference. An application to U.S. corporate leverage data reveals the coexistence of reinforcing and displacement interactions in firms' financial decisions.
\end{abstract}
\end{frontmatter}
\clearpage
\section{Introduction}


Networks and strategic interactions play an important role in shaping economic outcomes across a wide range of settings, including innovation, finance, trade, and regional development. A fundamental feature of many economic environments is that interactions may reflect multiple and potentially opposing forces. The same network may simultaneously generate reinforcing effects through imitation, learning, and information transmission, as well as displacement effects arising from competition, market-share reallocation, and strategic differentiation. Consequently, spillovers may reflect qualitatively different forms of dependence. 

Motivated by this observation, we consider an econometric framework in which interactions are generated by two \emph{latent} networks, one capturing the reinforcing effects and the other capturing the displacement effects
\begin{equation}
y_{i,t}
=
\eta_i
+
\boldsymbol{\beta}_i'\boldsymbol{x}_{i,t}
+
\rho_1 \sum_{j\neq i} w_{1,i,j} y_{j,t}
+
\rho_2 \sum_{j\neq i} w_{2,i,j} y_{j,t}
+
\varepsilon_{i,t},
\quad i=1,\ldots,N,\quad t=1,\ldots,T.
\label{eq:dgp1}
\end{equation}
Here, $y_{i,t}$ is the outcome of unit $i$, $\boldsymbol{x}_{i,t}$ is a vector of observed covariates, $\eta_i$ is a unit-specific effect, and $\varepsilon_{i,t}$ is an idiosyncratic disturbance. A key departure from conventional network and spatial models is that Eq. \eqref{eq:dgp1} allows interactions to arise through two latent channels. These channels are represented by non-negative interaction/weighting matrices $\boldsymbol{W}_1=(w_{1,i,j})$ and $\boldsymbol{W}_2=(w_{2,i,j})$, with associated intensities $\rho_1$ and $\rho_2$. Together, they give rise to a dual-network framework.

Given that the two networks are latent, for the sake of identification and interpretation we assume $\rho_1>0$ and $\rho_2<0$,\footnote{Or, the observationally equivalent setting $\rho_1<0$ and $\rho_2>0$.} so that the two channels operate in opposite directions. The first channel captures positive spillovers, strategic complementarities, or clustering effects, whereas the second captures strategic substitution, competition, or spatial dispersion effects. 
Every nonzero bilateral interaction is assumed to belong to at most one of the reinforcing and displacement networks; a dyad may belong to neither. This restriction has a natural economic interpretation and is important for identification.\footnote{Without this restriction, the data would generally identify only the combined effect $\rho_1 w_{1,i,j}+\rho_2 w_{2,i,j}$ rather than the separate contributions of the two interaction channels.}


Models based on a single interaction network arise as special cases of Eq. \eqref{eq:dgp1}. In particular, when one interaction component is absent, or when $\rho_1$ and $\rho_2$ have the same sign, the model reduces to a conventional interaction structure in which all links generate effects of the same qualitative type. Such models are related to the literature on social interactions and the reflection problem, beginning with \citet{Manski1993}, to identification through network structure in \citet{BramoulleEtAl2009}, and to micro-founded linear social interaction models such as \citet{BlumeEtAl2015} and \citet{LewbelQuTang2023}. They are also closely connected to spatial econometric models, see e.g. \citet{Anselin1988}, \citet{Elhorst2003}, \citet{ArbiaBaltagi2009}, \citet{Baltagi2021Spatial}, and \citet{Elhorst2014Book}.

The majority of the existing literature assumes that the interaction matrix is known a priori and focuses on estimation and inference for the structural parameters conditional on the network structure, see e.g. \citet{KelejianPrucha2010,BaltagiEtAl2013,ShiLee2017,Yang2018,CuiEtAl2023,Elhorst2024TwoLaws}.\footnote{Related contributions include \citet{HuangEtAl2020}, who develop a two-mode network autoregressive model that allows interactions across different classes of nodes through observed network matrices; and \citet{ChenEtAl2025}, who allow for heterogeneous interaction effects while maintaining a known interaction structure.} More recently, a growing literature has considered estimating sparse network links directly from the data; see, among others, \citet{AhrensBhattacharjee2015}, \citet{Manresa2016}, \citet{Rose2017}, \citet{LamSouza2020JBES}, \citet{dePaula2020}, \citet{dePaulaRasulSouza2024}, and \citet{KrisztinPiribauer2023}. This literature builds upon available results on high-dimensional sparse estimation, where techniques such as the Lasso, the adaptive Lasso, and the elastic net regularisation are used to recover sparse dependence structures in large systems. These methods typically impose a common sign structure on network effects, effectively treating all links as manifestations of a single interaction mechanism.\footnote{Although penalised network-recovery procedures may permit negative first-stage interaction estimates, these are typically treated within a single-network framework rather than as a distinct interaction mechanism. For example, in the adaptive elastic-net procedure of \citet{dePaulaRasulSouza2024}, the first-stage step imposes only $|w_{ij}|\le 1$, while the adaptive step sets $\widetilde w_{i,j}=0.05$ whenever $\widetilde w_{i,j}<0.05$ when constructing the adaptive penalty weights. Thus negative first-stage estimates are treated as small positive values for penalisation purposes, rather than as evidence of a separate displacement network.} Identification is then achieved through the reduced-form representation of the model. 

\emph{The present paper} is related to this literature, but differs in two important respects.
First, we allow for two latent interaction channels with opposite signs, while treating both \(\boldsymbol W_1\) and \(\boldsymbol W_2\) as unknown. Therefore, the first object of estimation is the composite structural interaction matrix
\[
\boldsymbol A
=
\rho_1\boldsymbol W_1+\rho_2\boldsymbol W_2
=
(\alpha_{i,j}),
\qquad
\alpha_{i,j}
=
\rho_1w_{1,i,j}+\rho_2w_{2,i,j}.
\]

\noindent
Second, the proposed approach differs from existing data-driven network-recovery methods in the way $\boldsymbol A$ is recovered. Reduced-form approaches identify the network
by first recovering the equilibrium mapping from covariates or shocks to outcomes and then using this mapping to infer the underlying interaction structure. Such approaches are
necessarily global: the reduced form is an all-or-nothing object that depends on the joint behaviour of the entire system. As a result, conditions ensuring the existence, uniqueness,
and stability of the global reduced-form representation play a central role. By contrast, the estimation approach developed here is constructive and equation-specific. For each equation $i$, it asks which candidate outcomes $y_{j,t}$ enter the structural equation for $y_{i,t}$; that is, which coefficients $\alpha_{i,j}$ are nonzero. Estimation is therefore based on local incremental contributions within each equation, using instrumental-variable screening statistics and multiple-testing control.


A central feature of this row-wise recovery strategy is that the nonzero entries of
$\boldsymbol A$ may be positive or negative. The procedure does not impose a common
sign restriction on the interaction matrix. This allows the recovered structural interaction matrix to contain both reinforcing and displacement forces.  To this end, we develop a step-wise instrumental-variables (IV) screening procedure with multiple-testing control for recovering sparse latent interaction structures in panel
network models. The procedure builds on the least-squares boosting framework of
\citet{Kapetanios2026BMT}, but adapts it to structural network recovery with endogenous
contemporaneous outcomes. For each equation, every other unit is initially treated as a candidate link. The candidates
are evaluated using IV screening statistics. At each step, the procedure selects at most one
link: the candidate with the largest statistic, provided that it exceeds the multiple-testing
threshold. The screening statistics are then recomputed conditional on the links selected in previous steps. The algorithm stops when no remaining candidate exceeds the threshold. We refer to the resulting procedure as  \textbf{Boosting One-Link-at-a-Time with Multiple Testing} (BOLMT).

After $\boldsymbol A$ has
been estimated, positive and negative entries are assigned to distinct latent channels.
Together with a disjoint-support restriction and row normalisations, this yields a
decomposition of the recovered composite matrix into reinforcing and displacement
networks. Hence, the dual-network structure is not imposed through a pre-specified
interaction matrix; it is obtained as a structured decomposition of the estimated
structural interaction matrix.

 The structural network recovery problem we consider raises several challenges that do not arise in standard model selection problems, e.g. as in \citet{Kapetanios2026BMT}. First, the candidate regressors are endogenous, such that the algorithm must rely on fitted values obtained from the corresponding first-stage projections given some instruments. Second, network recovery requires solving \(N\) high-dimensional selection problems, one for each equation of the system. Finally, while in standard high-dimensional regression, the inclusion of a small number of irrelevant variables need not be particularly harmful (provided that the selected model delivers a good approximation to the underlying signal), in the present context the support itself has a structural interpretation. Consequently, distinguishing direct links from indirect or spurious links that arise through common sources of dependence becomes important for recovery of the underlying interaction structure and its economic interpretation. In this respect, step-wise conditioning is particularly valuable in the present framework because the explanatory power of proxy links is expected to diminish once the relevant direct links have been selected.


To the best of our knowledge, this is the first framework to recover latent reinforcing and displacement interaction networks and to establish exact support recovery together with post-selection IV inference in a high-dimensional panel setting.

The remainder of the paper is organised as follows. Section~\ref{sec:BOLMT} introduces the dual-network framework, presents the BOLMT procedure, and discusses post-selection estimation and network decomposition. 
Section~\ref{sec:theory} establishes exact network recovery, recovery of the dual-network structure, and oracle-equivalent post-selection inference. Section~\ref{sec:discussion} discusses extensions, including block screening statistics and latent factor structures, as well as implementation issues. Section~\ref{sec:mc} reports Monte Carlo evidence and  Section~\ref{sec:empirical} provides an empirical illustration based on corporate financial policies. A final section concludes.

\section{Model and estimation}\label{sec:BOLMT}

\subsection{Network representation}

We rewrite the dual-network model in Eq.~\eqref{eq:dgp1} in the equivalent form
\begin{equation}
y_{i,t}
=
\eta_i
+
\boldsymbol{\beta}_i'\boldsymbol{x}_{i,t}
+
\sum_{j\neq i}\alpha_{i,j}y_{j,t}
+
\varepsilon_{i,t},
\qquad
i=1,\ldots,N,\quad t=1,\ldots,T,
\label{eq:structural_model}
\end{equation}
where we define
\[
\alpha_{i,j}
:=
\rho_1 w_{1,i,j}
+
\rho_2 w_{2,i,j}.
\]
The matrix of the corresponding interaction coefficients, $\boldsymbol{A}=(\alpha_{i,j})$, may be interpreted as a directed weighted network. There is an edge $j\rightarrow i$ whenever $\alpha_{i,j}\neq0$, since the outcome of unit $j$ enters the equation for unit $i$. The network need not be symmetric: the presence of $j\rightarrow i$ does not imply the presence of $i\rightarrow j$.

The representation in Eq.~\eqref{eq:structural_model} highlights two important features
of the recovery problem with latent $\alpha_{i,j}$. First, the interaction intensities and
the corresponding network weights are not separately identified from the composite
coefficients $\alpha_{i,j}$. The decomposition into reinforcing
$(\rho_1,\boldsymbol W_1)$ and displacement $(\rho_2,\boldsymbol W_2)$ networks is
therefore a second-stage operation performed only after recovery and estimation of the
structural interaction matrix $\boldsymbol A=(\alpha_{i,j})$. Second, the interaction
structure is inherently heterogeneous across equations. The supports need not be symmetric,
may differ in size, and may involve coefficients of different magnitudes. As a result, both the support and the coefficients of each row of $\boldsymbol A$ are individual-specific. The support-recovery problem therefore consists of $N$ separate high-dimensional selection
problems, one for each row of $\boldsymbol A$.

For each unit $i$, denote by $\mathcal S_i^n$ the set of units whose outcomes enter the
structural equation of unit $i$ through nonzero interaction coefficients
\begin{equation}
\mathcal S_i^{n}
=
\left\{
j\neq i:\alpha_{i,j}\neq 0
\right\}.
\label{eq:true_support}
\end{equation}
This support set is unobserved. The researcher observes only the set of potential candidate
links, $\{1,\ldots,N\}\setminus\{i\}$, and seeks to determine which of these candidates
belong to $\mathcal S_i^n$.
The framework developed in this paper builds upon the decomposition
\begin{equation}
\label{eq:set_decomposition}
\{1,\ldots,N\}\setminus\{i\}
=
\mathcal S_i^{n}
\cup
\mathcal S_i^{p}
\cup
\mathcal S_i^{d},
\end{equation}
where the three (non-overlapping) sets are defined as:
\begin{itemize}
\item $\mathcal S_i^{n}$ is the set of true links;

\item $\mathcal S_i^{p}$ denotes the set of \textit{proxy} links; namely, units for which $\alpha_{i,j}=0$ but whose outcomes $y_{j,t}$ remain correlated with the structural component of $y_{i,t}$ through indirect network propagation (or common factors);

\item $\mathcal S_i^{d}$ is the set of irrelevant units -- namely, units for which $\alpha_{i,j}=0$ and whose population screening signal is asymptotically negligible..
\end{itemize}
\noindent
The distinction between the three types of units is central for network recovery. In standard high-dimensional regression, the inclusion of a small number of proxy
variables may still yield a useful approximating model. In the present context, however, the support itself has a structural interpretation. Selecting proxy links or irrelevant units may
distort the estimated interaction structure and, consequently, the implied network propagation mechanism. Moreover, when the reduced form is well defined, reduced-form dependence may still be misleading about the sign and strength of the underlying structural interaction coefficient.\footnote{For example, suppose \(N=4\), and the reinforcing network contains links \(2\rightarrow 1\) and \(3\rightarrow 2\), while the displacement network contains links \(3\rightarrow 1\) and \(4\rightarrow 2\). Let \(\rho_1=0.5\) and \(\rho_2=-0.2\). Then unit 3 has a negative immediate structural interaction coefficient in the equation for unit 1 through the displacement link \(3\rightarrow 1\). At the same time, unit 3 also affects unit 1 indirectly through the path \(3\rightarrow 2\rightarrow 1\), generating a positive
second-order effect. The latter dominates the former, so that reduced-form dependence between units 1 and 3 has the opposite sign from the immediate structural interaction coefficient.
}

Figure~\ref{fig:true_pseudo_distant} provides a simple illustration of these concepts for equation $i$, where unit $j_1 \in \mathcal S_i^n$, unit $j_2 \in \mathcal S_i^p$, and unit $j_3 \in \mathcal S_i^d$.

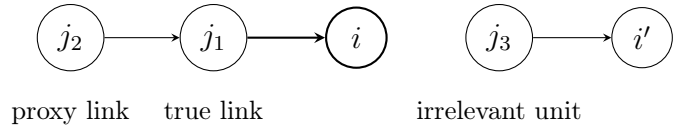
\begin{figure}[ht]
\centering
\begin{tikzpicture}[>=stealth,scale=1.05]

\node[circle,draw,minimum size=0.8cm,thick] (i) at (4.0,0) {$i$};
\node[circle,draw,minimum size=0.8cm] (j1) at (2.2,0) {$j_1$};
\node[circle,draw,minimum size=0.8cm] (j2) at (0.4,0) {$j_2$};

\node[circle,draw,minimum size=0.8cm] (j3) at (5.8,0) {$j_3$};
\node[circle,draw,minimum size=0.8cm] (ip) at (7.6,0) {$i'$};

\draw[->,thick] (j1) -- (i);
\draw[->] (j2) -- (j1);
\draw[->] (j3) -- (ip);

\node[below=0.25cm of j1] {\footnotesize true link};
\node[below=0.25cm of j2] {\footnotesize proxy link};
\node[below=0.25cm of j3] {\footnotesize irrelevant unit};

\end{tikzpicture}
\caption{True, proxy, and irrelevant candidates for equation \(i\).}
\label{fig:true_pseudo_distant}
\end{figure}

\subsection{The BOLMT algorithm}\label{BOLMT}
In what follows, we formally introduce the BOLMT algorithm. 
The contemporaneous nature of the model implies that the network-related regressors are endogenous. In particular, if unit $j$ responds contemporaneously to shocks elsewhere in the system, then $y_{j,t}$ can be correlated with $\varepsilon_{i,t}$. Consequently, the screening statistics are constructed using instrumental variables and fitted regressors obtained from the corresponding first-stage projections. The instruments may consist of exogenous covariates, predetermined variables, or other application-specific instruments satisfying the usual relevance and exogeneity conditions.

At any stage of the algorithm, the objective is to test whether a remaining candidate unit
\(j\) contributes additional information about the outcome of unit \(i\), conditional on the controls and the links selected in previous stages. Since the previously selected contemporaneous neighbour outcomes are endogenous in the network system, the stagewise statistic is computed
as a partial IV statistic. Thus the previously selected link outcomes and the current candidate are instrumented jointly, and the candidate is tested after partialling out the fitted selected
controls.

Let \(H_{i,m-1}\subseteq \{1,\ldots,N\}\setminus\{i\}\) denote the set of links already selected before a
generic stage of the algorithm, $m$. Let \(\mathbf C_{i,H_{i,m-1}}\) denote the matrix containing all always-included regressors together with the outcomes of the units in \(H_{i,m-1}\).\footnote{Depending on the application, the always-included regressors may contain observed covariates, time lags of
the dependent variable, deterministic trends, seasonal effects, fixed effects, factor proxies, or other controls that are not subject to selection.}

For a remaining candidate \(j\notin H_{i,m-1}\), let \(\mathbf v_j:= \mathbf y_j\) denote the candidate
outcome vector. Let \(\boldsymbol{\mathcal Z}_{i,j,H_{i,m-1}}\) denote the full stagewise instrument matrix used for
both \(\mathbf C_{i,H_{i,m-1}}\) and \(\mathbf v_j\). Thus \(\boldsymbol{\mathcal Z}_{i,j,H_{i,m-1}}\) contains instruments for the
previously selected neighbour outcomes, instruments for the current candidate, and any
exogenous controls used as their own instruments.

Define the fitted selected controls and fitted candidate by
\[
\widehat{\mathbf C}_{i,j,H_{i,m-1}}
=
\boldsymbol{\mathcal Z}_{i,j,H_{i,m-1}}
(\boldsymbol{\mathcal Z}_{i,j,H_{i,m-1}}'\boldsymbol{\mathcal Z}_{i,j,H_{i,m-1}})^{-1}
\boldsymbol{\mathcal Z}_{i,j,H_{i,m-1}}'\mathbf C_{i,H_{i,m-1}},
\]
and
\[
\overline{\mathbf v}_{i,j,H_{i,m-1}}
=
\boldsymbol{\mathcal Z}_{i,j,H_{i,m-1}}
(\boldsymbol{\mathcal Z}_{i,j,H_{i,m-1}}'\boldsymbol{\mathcal Z}_{i,j,H_{i,m-1}})^{-1}
\boldsymbol{\mathcal Z}_{i,j,H_{i,m-1}}'\mathbf v_j.
\]
Let
\begin{equation}
\widehat{\mathbf M}^{C}_{i,j,H_{i,m-1}}
=
\mathbf I_T
-
\widehat{\mathbf C}_{i,j,H_{i,m-1}}
(\widehat{\mathbf C}_{i,j,H_{i,m-1}}'\widehat{\mathbf C}_{i,j,H_{i,m-1}})^{-1}
\widehat{\mathbf C}_{i,j,H_{i,m-1}}' .
\label{eq:ivfwl_residualmaker}
\end{equation}
The partial IV/FWL residuals entering the scalar screening statistic are
\begin{equation}
\widetilde{\mathbf y}_{i,j,H_{i,m-1}}
=
\widehat{\mathbf M}^{C}_{i,j,H_{i,m-1}}\mathbf y_i,
\qquad
\widehat{\mathbf v}_{i,j,H_{i,m-1}}
=
\widehat{\mathbf M}^{C}_{i,j,H_{i,m-1}}\overline{\mathbf v}_{i,j,H_{i,m-1}}.
\label{eq:ivfwl_scalar_objects}
\end{equation}
Thus \(\widehat{\mathbf v}_{i,j,H_{i,m-1}}\) is the fitted candidate residual after projecting the fitted
candidate on the fitted selected controls. 

The corresponding IV estimator is given by define
\begin{equation}
\widehat{\theta}_{i,j,H_{i,m-1}}
=
\frac{
\widehat{\mathbf v}_{i,j,H_{i,m-1}}'
\widetilde{\mathbf y}_{i,j,H_{i,m-1}}
}
{
\widehat{\mathbf v}_{i,j,H_{i,m-1}}'
\widehat{\mathbf v}_{i,j,H_{i,m-1}}
},
\end{equation}
with the corresponding screening statistic
\begin{equation}
t_{i,j,H_{i,m-1}}
=
\frac{
T^{-1/2}
\widehat{\mathbf v}_{i,j,H_{i,m-1}}'
\widetilde{\mathbf y}_{i,j,H_{i,m-1}}
}
{
\widehat{\sigma}_{i,j,H_{i,m-1}}
\left(
T^{-1}
\widehat{\mathbf v}_{i,j,H_{i,m-1}}'
\widehat{\mathbf v}_{i,j,H_{i,m-1}}
\right)^{1/2}
}.
\label{eq:tstat}
\end{equation}
Here \(\widehat{\sigma}_{i,j,H_{i,m-1}}^{2}\) is the residual variance estimator from the same full
stagewise IV regression. Only the values of $j$ such that:
\begin{equation}
|t_{i,j,H_{i,m-1}}|> c_{N},
\end{equation}
are included in the next stage of the algorithm. The full algorithm is summarized in Algorithm \ref{alg:bolmt}, where procedure is repeated for every unit \(i=1,\ldots,N\). 

\begin{algorithm}[H]
\caption{BOLMT support selection algorithm}
\label{alg:bolmt}
\begin{algorithmic}[1]
\State Initialize $H_{i,0} \gets \varnothing$.
\For{$m=1,2,\ldots$}
    \State Compute $t_{i,j,H_{i,m-1}}$ for every remaining candidate $j\notin H_{i,m-1}$.
    \State Let
    \[
    j^\ast \in \arg\max_{j\notin H_{i,m-1}} |t_{i,j,H_{i,m-1}}|.
    \]
    \If{$|t_{i,j^\ast,H_{i,m-1}}|>c_N$}
        \State Set $H_{i,m} \gets H_{i,m-1}\cup\{j^\ast\}$.
    \Else
        \State Set $\widehat{\mathcal S}_i\gets H_{i,m-1}$.
        \State \textbf{stop}.
    \EndIf
\EndFor
\end{algorithmic}
\end{algorithm}

In practice, we implement the threshold using
\begin{equation}
c_{N}
=
\Phi^{-1}
\left(
1-\frac{p}{2f(N,\delta)}
\right),
\qquad
f(N,\delta)
=
cN^\delta,
\label{eq:cvf}
\end{equation}
where \(p\in(0,1)\) is a tuning parameter and \(c,\delta>0\) are user-specified
constants, and \(\Phi\) is the standard normal CDF. Under standard tail
approximations, \(c_N\asymp\sqrt{\log N}\), so that \(c_p(N,\delta)\) provides a
feasible tuning rule of the same order as the theoretical threshold sequence \(c_N\).
The recovery theory covers this rule only when its leading constant is sufficiently large.
Interpreting Eq. \eqref{eq:cvf} as a nominal Gaussian multiple-testing critical value would
additionally require asymptotic standard normality, or appropriate long-run-variance
studentisation, neither of which is imposed in the selection theory below.

\noindent
Since the selected supports need not be symmetric, the resulting network is directed.

\noindent
If all stagewise regressors are exogenous and the instrument matrix coincides with the regressor matrix, this construction reduces to ordinary least-squares partialling out, as in standard BMT of \citet{Kapetanios2026BMT}. 

\subsection{Post-selection dual network decomposition}\label{sec:psel_estimation}

Once BOLMT has been applied to all cross-sectional units, the estimated support matrix is constructed from the selected supports $\{\widehat{\mathcal{S}}_i\}_{i=1}^{N}$, with elements
\begin{equation}
\widehat a_{i,j}^{\,s}
=
\mathbf 1
\left(
j\in \widehat{\mathcal{S}}_i
\right).
\end{equation}
Conditional on the selected support, we estimate model
\begin{equation}
y_{i,t}
=
\eta_i
+
\boldsymbol{\beta}_i'
\boldsymbol{x}_{i,t}
+
\sum_{j\in\widehat{\mathcal{S}}_i}
\alpha_{i,j}y_{j,t}
+
u_{i,t}(\widehat{\mathcal S}_i),
\label{eq:postselection}
\end{equation}
for every $i$ using the corresponding IV regression, where
$u_{i,t}(\widehat{\mathcal S}_i)$ includes any omitted interaction terms and equals
$\varepsilon_{i,t}$ on the exact-recovery event.
The post-selection instrument matrix is taken to be a fixed function of the selected support;
the oracle estimator uses the same rule evaluated at $\mathcal S_i^n$.

Let $\widehat{\boldsymbol{\beta}}_i$ and $\widehat{\alpha}_{i,j}$ denote the
post-selection IV estimates obtained from Eq. \eqref{eq:postselection}, with
$\widehat\alpha_{i,j}=0$ for $j\notin\widehat{\mathcal S}_i$. Define
\begin{equation}
\widehat{\mathcal{S}}_{1,i}
=
\{j\in\widehat{\mathcal{S}}_i:\widehat\alpha_{i,j}>0\},
\qquad
\widehat{\mathcal{S}}_{2,i}
=
\{j\in\widehat{\mathcal{S}}_i:\widehat\alpha_{i,j}<0\}.
\end{equation}
Using this definition, the estimated interaction coefficients $\widehat{\alpha}_{i,j}$ are decomposed according to their sign. Define
\begin{equation}
\widehat\rho_{1,i}
:=
\sum_{j\in\widehat{\mathcal S}_{1,i}}
\widehat\alpha_{i,j},
\qquad
\widehat\rho_{2,i}
:=
\sum_{j\in\widehat{\mathcal S}_{2,i}}
\widehat\alpha_{i,j},
\end{equation}
and define the row-normalised weights, for $\ell=1,2$, by
\[
\widehat w_{\ell,i,j}
=
\begin{cases}
\widehat\alpha_{i,j}/\widehat\rho_{\ell,i},
&j\in\widehat{\mathcal S}_{\ell,i}\ \text{and}\
 \widehat{\mathcal S}_{\ell,i}\neq\varnothing,\\
0,&\text{otherwise}.
\end{cases}
\]
Thus an empty signed row of $\widehat{\mathbf W}_{\ell}$ is set equal to zero and no division by zero is made.

As a final step, we aggregate all unit $i$ specific estimates $\widehat{\boldsymbol{\theta}}_i
=
\left(
\widehat{\boldsymbol{\beta}}_i',
\widehat\rho_{1,i},
\widehat\rho_{2,i}
\right)'$ using the Mean Group (MG) aggregation, as \citet{Pesaran199579}:
\begin{equation}
\widehat{\boldsymbol{\theta}}_{MG}
=
\frac{1}{N}
\sum_{i=1}^{N}
\widehat{\boldsymbol{\theta}}_i.
\label{eq:mg_estimator}
\end{equation}

The implementation is fundamentally local and proceeds equation by equation. It does not
require global graph restrictions such as symmetry, connectedness, or particular topological
features of the interaction network. Such restrictions may be useful in specific applications,
but they are not needed for the row-wise selection arguments developed here.

The primary object recovered by BOLMT is the structural interaction matrix
$\boldsymbol A=(\alpha_{i,j})$. The coefficient $\alpha_{i,j}$ indicates whether the
outcome of unit $j$ enters the structural equation for unit $i$, and with what sign and
magnitude. It is therefore an immediate interaction coefficient, not a direct effect in the
usual spatial-econometric sense. Direct, indirect, and total effects are defined from the
relevant multiplier matrix. For example, when the full system is stable, the effect matrix
for covariate $\ell$ is based on
\[
(\boldsymbol I_N-\boldsymbol A)^{-1}
\operatorname{diag}(\beta_{1\ell},\ldots,\beta_{N\ell}),
\]
with direct effects obtained from the diagonal elements and indirect effects from the
off-diagonal elements. If only a retained subsystem is stable, the same logic applies to the
corresponding subsystem multiplier, conditional on the excluded outcomes being treated as
observed external inputs.

The dual-network representation is therefore a structured decomposition of the recovered
matrix $\boldsymbol A$. Without additional restrictions, the data identify the composite
coefficients $\alpha_{i,j}$, not the separate objects $\rho_1w_{1,i,j}$ and
$\rho_2w_{2,i,j}$. The sign decomposition, disjoint-support restriction, and row
normalisations deliver the reinforcing and displacement weights and the corresponding
row-specific interaction intensities.

\section{Theoretical Results}
\label{sec:theory}

This section collects the assumptions used for recovery of the structural interaction matrix $\boldsymbol{A}=(\alpha_{i,j})$ and for post-selection estimation. Afterwards, based on these assumptions consistency of the proposed screening procedure, exact recovery of the latent interaction network, and oracle properties of the resulting post-selection estimators is established.
Throughout, this section all asymptotic statements are understood as jointly $N,T\to\infty$ unless otherwise stated.

\subsection{Assumptions}\label{sec:assumptions}

The assumptions in this section are stated as population restrictions on a deterministic class
of low-dimensional conditioning sets.  They are not restrictions on a realised random path of
the algorithm.  The connection with the realised BOLMT path is made in the appendix by an
induction argument: with probability tending to one, the selected set remains in the class of
conditioning sets defined below.

Recall that \(\mathcal S_i^n\) denotes the set of true links. For the theoretical
analysis, the remaining candidates are partitioned into two groups.
The set \(\mathcal S_i^p\) contains \emph{proxy links}, namely candidates for which
\(\alpha_{i,j}=0\) but whose outcomes remain correlated with the structural
component of \(y_{i,t}\) through indirect network propagation or common exposure.
The remaining candidates are termed \emph{irrelevant links}. Accordingly, define
\[
\mathcal K_i
:=
\mathcal S_i^n
\cup
\mathcal S_i^p
\qquad
k_i^K:=|\mathcal K_i|,
\qquad
\bar K:=\sup_i k_i^K,
\]
and
\[
\mathcal S_i^d
=
\{1,\ldots,N\}\setminus(\{i\}\cup\mathcal K_i).
\]

Let
\[
\mathfrak U_i=\{U:U\subseteq\mathcal K_i\}
\]
be the deterministic class of low-dimensional conditioning sets.  In the theorem statements
and appendix we write \(H\) for a generic element of this class and set
\[
\mathcal H_i:=\mathfrak U_i.
\]
Also define
\[
\mathcal H_i^0=\{H\in\mathcal H_i:\mathcal S_i^n\not\subseteq H\},
\qquad
\mathcal H_i^1=\{H\in\mathcal H_i:\mathcal S_i^n\subseteq H\}.
\]
We call \(H\in\mathcal H_i^0\) a pre-completion conditioning set and
\(H\in\mathcal H_i^1\) a post-completion conditioning set. Completion means that
all true links have been included in \(H\); it does not necessarily mean that the
screening algorithm has stopped.
Thus \(\mathcal H_i^0\) contains low-dimensional conditioning sets that do not yet contain all
true links, while \(\mathcal H_i^1\) contains conditioning sets that contain the true support.
Since \(\bar K<\infty\),
\[
|\mathcal H_i|\le 2^{\bar K}
\]
uniformly in \(i\).  Hence the high-dimensional probability bounds below are required to be
uniform over equations and candidates, but not over all subsets of the \(N-1\) candidate set.

For any \(H\in\mathcal H_i\), let \(\mathbf C_{i,H,t}\) collect the always-included controls and
the outcomes of units in \(H\). For a candidate \(j\notin H\), the scalar candidate outcome is
\(y_{j,t}\). Let $\boldsymbol{\mathcal Z}_{i,j,H,t}$ denote the full stagewise instrument vector used jointly for \(\mathbf C_{i,H,t}\) and \(y_{j,t}\).
It contains instruments for the conditioning variables, instruments for the current candidate,
and any exogenous regressors used as their own instruments.

Define
\[
\mathbf Q^{\boldsymbol{\mathcal Z}}_{i,j,H}
:=
\E(\boldsymbol{\mathcal Z}_{i,j,H,t}\boldsymbol{\mathcal Z}_{i,j,H,t}').
\]
The population fitted controls and fitted candidate are
\[
\mathbf C^*_{i,j,H,t}
:=
\boldsymbol\Pi_{C,i,j,H}'\boldsymbol{\mathcal Z}_{i,j,H,t},
\qquad
\bar v^*_{i,j,H,t}
:=
\boldsymbol\pi_{v,i,j,H}'\boldsymbol{\mathcal Z}_{i,j,H,t},
\]
where
\[
\boldsymbol\Pi_{C,i,j,H}
:=
(\mathbf Q^{\boldsymbol{\mathcal Z}}_{i,j,H})^{-1}
\E(\boldsymbol{\mathcal Z}_{i,j,H,t}\mathbf C_{i,H,t}'),
\qquad
\boldsymbol\pi_{v,i,j,H}
:=
(\mathbf Q^{\boldsymbol{\mathcal Z}}_{i,j,H})^{-1}
\E(\boldsymbol{\mathcal Z}_{i,j,H,t}y_{j,t}).
\]
Let
\[
\boldsymbol\Gamma_{vC,i,j,H}
:=
\{\E(\mathbf C^*_{i,j,H,t}\mathbf C_{i,j,H,t}^{*\prime})\}^{-1}
\E(\mathbf C^*_{i,j,H,t}\bar v^*_{i,j,H,t}),
\]
and define the population partial fitted candidate
\begin{equation}
v^*_{i,j,H,t}
:=
\bar v^*_{i,j,H,t}
-
\boldsymbol\Gamma_{vC,i,j,H}'\mathbf C^*_{i,j,H,t}.
\label{eq:v_star_def}
\end{equation}
Here, \(\bar v^*_{i,j,H,t}\) is the population first-stage fitted value of
\(y_{j,t}\), while \(v^*_{i,j,H,t}\) is the residual obtained after partialling
this fitted value with respect to the population fitted controls.

For $H\in\mathcal H_i$, define the omitted structural component
\[
\mu_{i,H,t}
=
\sum_{k\in\mathcal S_i^n\setminus H}\alpha_{i,k}y_{k,t},
\qquad
\mu_{i,j,H,t}:=\mu_{i,H,t},
\]
where the second notation is retained for compatibility with candidate-specific score arrays.
Also let
\[
\sigma_{i,j,H}^2
:=
\operatorname*{plim}_{T\to\infty}\widehat\sigma_{i,j,H}^2
\]
denote the population limit of the residual-variance estimator from the corresponding
stagewise IV regression.

The tail and dependence conditions below are imposed uniformly on the centred
versions of the following scalar processes, over \(i\), \(j\), and
\(H\in\mathcal H_i\):
\[
v^*_{i,j,H,t}\varepsilon_{i,t},
\qquad
v^*_{i,j,H,t}\mu_{i,j,H,t},
\qquad
(v^*_{i,j,H,t})^2,
\]
\[
\boldsymbol{\mathcal Z}_{i,j,H,t,a}
\boldsymbol{\mathcal Z}_{i,j,H,t,b},
\qquad
\boldsymbol{\mathcal Z}_{i,j,H,t,a}
y_{j,t},
\qquad
\boldsymbol{\mathcal Z}_{i,j,H,t,a}
\mathbf C_{i,H,t,b},
\qquad
\varepsilon_{i,t}^{2}.
\]
These are the scalar score, signal, Gram-matrix, and first-stage arrays
controlled by the probability bounds.

For a scalar array \(\xi_t\), write \(\xi_t\in E(s)\), \(s>0\), if
\(\sup_t\E \exp(a|\xi_t|^s)<\infty\) for some \(a>0\).  Write \(\xi_t\in H(\vartheta)\),
\(\vartheta>2\), if \(\sup_t\E|\xi_t|^\vartheta<\infty\).  For \(\zeta>0\), define
\begin{align}
f_T(\gamma_1,\gamma_2,c,\zeta)
&=
c_0
\left\{
\exp(-c_1\zeta^{\gamma_1})
+
\exp
\left[
-c_2
\left(
\frac{\zeta\sqrt T}
{\log^2T}
\right)^{\gamma_2}
\right]
\right\},
\label{eq:fT}
\\
g_T(\gamma,\vartheta,c,\zeta)
&=
c_0
\left\{
\exp(-c_1\zeta^\gamma)
+
\zeta^{-\vartheta}
T^{-(\vartheta/2-1)}
\right\}.
\label{eq:gT}
\end{align}
The functions \(f_T(\cdot)\) and \(g_T(\cdot)\) are the thin-tail and heavy-tail concentration
envelopes of \citet{DendramisEtal2022}. All population moments used below are assumed to be time invariant, unless stated otherwise.

\begin{assumption}
\label{ass:sampling}
For every array \(\{\xi_{i,j,H,t}\}\) generated by the scalar processes listed
above, the centred process
\[
\widetilde{\xi}_{i,j,H,t}:=\xi_{i,j,H,t}-\E(\xi_{i,j,H,t})
\]
is exponentially \(\alpha\)-mixing in \(t\), uniformly over \(i\), \(j\), and \(H\in\mathcal H_i\):
\[
\alpha(k)\le C\varphi^k,
\qquad
0<\varphi<1.
\]
\end{assumption}

\begin{assumption}
\label{ass:network}
The network is row sparse.  Specifically,
\[
0\le |\mathcal S_i^n|\le \bar k<\infty,
\qquad
|\mathcal S_i^p|\le \bar k_p<\infty,
\]
uniformly in \(i\).  Hence \(\bar K=\sup_i|\mathcal K_i|<\infty\).
\end{assumption}

\begin{assumption}\label{ass:iv}
The dimensions of the stagewise instrument and control vectors are uniformly bounded.
For every \(i\), \(H\in\mathcal H_i\), and \(j\notin H\), the full stagewise instrument vector
satisfies
\[
\E(\boldsymbol{\mathcal Z}_{i,j,H,t}\varepsilon_{i,t})=\mathbf 0.
\]
Moreover,
\[
0<c\le
\lambda_{\min}\left\{\E(\boldsymbol{\mathcal Z}_{i,j,H,t}\boldsymbol{\mathcal Z}_{i,j,H,t}')\right\}
\le
\lambda_{\max}\left\{\E(\boldsymbol{\mathcal Z}_{i,j,H,t}\boldsymbol{\mathcal Z}_{i,j,H,t}')\right\}
\le C<\infty
\]
uniformly over \(i,j,H\). The fitted-control Gram matrix also satisfies
\[
0<c\le
\lambda_{\min}\!\left\{\E(\mathbf C^*_{i,j,H,t}\mathbf C_{i,j,H,t}^{*\prime})\right\}
\le
\lambda_{\max}\!\left\{\E(\mathbf C^*_{i,j,H,t}\mathbf C_{i,j,H,t}^{*\prime})\right\}
\le C<\infty
\]
uniformly over $i,j,H$. The scalar partial fitted candidate is nondegenerate:
\[
\inf_{i,j,H}\E\left[(v^*_{i,j,H,t})^2\right]>0.
\]
Finally,
\[
0<c\le\inf_{i,j,H}\sigma_{i,j,H}^2
\le\sup_{i,j,H}\sigma_{i,j,H}^2\le C<\infty.
\]
\end{assumption}

\begin{assumption}
\label{ass:first_stage}
Uniformly over \(i\), \(j\), and \(H\in\mathcal H_i\), the sample first-stage moments converge to their
population counterparts:
\[
\left\|
T^{-1}\sum_{t=1}^T
\left(\boldsymbol{\mathcal Z}_{i,j,H,t}\boldsymbol{\mathcal Z}_{i,j,H,t}'
-
\E(\boldsymbol{\mathcal Z}_{i,j,H,t}\boldsymbol{\mathcal Z}_{i,j,H,t}')\right)
\right\|=o_p(1),
\]
\[
\left\|
T^{-1}\sum_{t=1}^T
\boldsymbol{\mathcal Z}_{i,j,H,t}\mathbf C_{i,H,t}'
-
\E(\boldsymbol{\mathcal Z}_{i,j,H,t}\mathbf C_{i,H,t}')
\right\|=o_p(1),
\]
\[
\left\|
T^{-1}\sum_{t=1}^T
\boldsymbol{\mathcal Z}_{i,j,H,t}y_{j,t}
-
\E(\boldsymbol{\mathcal Z}_{i,j,H,t}y_{j,t})
\right\|=o_p(1).
\]
In addition, the induced feasible-to-population partial fitted replacement used in the
proofs satisfies:
\[
\sup_{i,j,H}
T^{-1}\sum_{t=1}^T(\widehat v_{i,j,H,t}-v^*_{i,j,H,t})^2=o_p(1),
\]
\[
\sup_{i,j,H}
\left|
T^{-1/2}\sum_{t=1}^T
(\widehat v_{i,j,H,t}-v^*_{i,j,H,t})\varepsilon_{i,t}
\right|=o_p(\sqrt{\log N}),
\]
\[
\sup_{i,j,H}
\left|
T^{-1/2}\sum_{t=1}^T
(\widehat v_{i,j,H,t}-v^*_{i,j,H,t})\mu_{i,j,H,t}
\right|=o_p(\sqrt{\log N}),
\]
and
\[
\sup_{i,j,H}
|\widehat\sigma_{i,j,H}^2-\sigma_{i,j,H}^2|
=o_p(1).
\]
\end{assumption}

\begin{assumption}
\label{ass:tails}
One of the following regimes holds.

\medskip
\noindent Regime E. Every centred scalar process listed above belongs uniformly to an \(E(s)\) class.
Let \(\gamma_E\in(0,1]\) denote the smallest DGK exponent induced by these processes.  For all
sufficiently large constants \(C\),
\begin{equation}
N^2
f_T(2,\gamma_E,c,C\sqrt{\log N})
\to0.
\label{eq:growth_E}
\end{equation}

\medskip
\noindent Regime H. Every centred scalar process listed above belongs uniformly to an \(H(\vartheta)\)
class, where \(\vartheta\) is the moment exponent of the relevant score and Gram arrays.  For
some \(2<\vartheta_0<\vartheta\),
\begin{equation}
N^2
(\log N)^{-\vartheta_0/2}
T^{-(\vartheta_0/2-1)}
\to0.
\label{eq:growth_H}
\end{equation}
\end{assumption}

\noindent
For each equation \(i\), let
\[
U_i(H)=\mathcal S_i^n\setminus H,
\qquad
P_i(H)=\mathcal S_i^p\setminus H,
\]
denote the sets of true and proxy links not yet selected under conditioning set \(H\).
We use the convention that a maximum over an empty set is zero.
For any remaining candidate \(j\notin H\), set
\(\widetilde y_{j,H,t}:=v^*_{i,j,H,t}\), with the equation index $i$ suppressed on the
left-hand side. Define
\[
d_{i,j,H}
=
\sigma_{i,j,H}
\left\{
\E[(\widetilde y_{j,H,t})^2]
\right\}^{1/2}.
\]
Define the normalised IV alignment
\[
\bar\gamma_{i,j,k}(H)
=
\frac{
\E(\widetilde y_{j,H,t}y_{k,t})
}{
d_{i,j,H}
},
\qquad k\in U_i(H),
\]
and the normalised population IV screening signal
\[
m_{i,j}(H)
=
\frac{
\E(\widetilde y_{j,H,t}\mu_{i,j,H,t})
}{
d_{i,j,H}
}.
\]
Because $v^*_{i,j,H,t}$ is a linear combination of
$\boldsymbol{\mathcal Z}_{i,j,H,t}$,
\[
\E(v^*_{i,j,H,t}\mathbf C_{i,H,t}')
=\E(v^*_{i,j,H,t}\mathbf C_{i,j,H,t}^{*\prime})=\mathbf0,
\qquad
\E(v^*_{i,j,H,t}\varepsilon_{i,t})=0.
\]
Consequently,
$\E(\widetilde y_{j,H,t}y_{i,t})=\E(\widetilde y_{j,H,t}\mu_{i,j,H,t})$.
The normalised population screening signal admits the decomposition
\begin{equation}
m_{i,j}(H)
=
\sum_{k\in U_i(H)}
\alpha_{i,k}\bar\gamma_{i,j,k}(H),
\label{eq:scalar_signal_decomp}
\end{equation}
which follows directly from the structural equation.
This decomposition links the primitive alignment quantities $\bar\gamma_{i,j,k}(H)$ to the population screening signal $m_{i,j}(H)$, and underlies the subsequent detectability and dominance assumptions.
\begin{assumption}[True-link detectability]
\label{ass:t_signal}
For scalar screening, the following population restrictions hold.

\begin{enumerate}[label=(\roman*)]

\item There exists a sequence \(b_T>0\) such that, over all nonzero structural links,
\[
\inf_{\{(i,j):\,j\in\mathcal S_i^n\}}
|\alpha_{i,j}|
\ge b_T,
\qquad
\frac{\sqrt T\,b_T}{\sqrt{\log N}}
\to\infty.
\]
If \(\mathcal S_i^n=\varnothing\) for some equation, this condition is void for that equation.

\item At every pre-completion conditioning set, at least one remaining true link has an own
population IV alignment that dominates the contribution induced by the other remaining true links:
there exists \(c_t>0\) such that
\[
\inf_{\{(i,H):\,H\in\mathcal H_i^0\}}
\max_{j\in U_i(H)}
\left\{
|\alpha_{i,j}|\,|\bar\gamma_{i,j,j}(H)|
-
\sum_{\substack{k\in U_i(H)\\ k\neq j}}
|\alpha_{i,k}|\,|\bar\gamma_{i,j,k}(H)|
\right\}
\ge c_t b_T.
\]
If \(\mathcal H_i^0=\varnothing\), the condition is void for that equation.

\item Irrelevant candidates have negligible scalar population IV signal:
\[
\max_i\max_{H\in\mathcal H_i}\max_{j\in\mathcal S_i^d}
\sum_{k\in U_i(H)}
|\alpha_{i,k}|\,|\bar\gamma_{i,j,k}(H)|
=
 o\!\left(\sqrt{\frac{\log N}{T}}\right).
\]
\end{enumerate}
\end{assumption}
Assumption~\ref{ass:t_signal} formalises the population conditions under which true links remain distinguishable from irrelevant candidates throughout the screening procedure. Together with Assumption~\ref{ass:t_dominance}, it implies the screening properties required for the subsequent selection proofs, as established in Lemma~\ref{lem:app_primitive_t_signal}.

\begin{assumption}[Primitive proxy-alignment condition]
\label{ass:t_dominance}
For exact recovery, the following conditions hold.

First, before completion, proxy links are uniformly weaker than the strongest
remaining true-link signal. There exists a constant \(\lambda\in[0,1)\) such
that, uniformly over \(i\) and \(H\in\mathcal H_i^0\),
\[
\max_{p\in P_i(H)}
\left|
\sum_{k\in U_i(H)}
\alpha_{i,k}\bar\gamma_{i,p,k}(H)
\right|
\le
\lambda
\max_{s\in U_i(H)}
\left|
\sum_{k\in U_i(H)}
\alpha_{i,k}\bar\gamma_{i,s,k}(H)
\right|.
\]

Second, after completion, remaining proxy links have negligible normalised
population screening signal:
\[
\max_{p\in P_i(H)}
|m_{i,p}(H)|
=
o\!\left(\sqrt{\frac{\log N}{T}}\right)
\]
uniformly over \(i\) and \(H\in\mathcal H_i^1\).
\end{assumption}

\noindent
The first part of Assumption~\ref{ass:t_dominance} places a primitive restriction on the population IV alignment of proxy links relative to true links. A true link contributes directly to the structural equation for \(y_i\), whereas a proxy link can generate a screening signal only through its association with omitted true links. The assumption therefore requires proxy links to be uniformly less aligned with the remaining structural component than the strongest remaining true link. The appendix shows that, together with Assumption~\ref{ass:t_signal}, this primitive restriction implies that, before completion, the strongest remaining true-link screening signal dominates the strongest remaining proxy-link screening signal. The second part imposes the corresponding post-completion requirement: once all true links have been selected, proxy links should have no remaining structural source of population screening signal and therefore should not cross the stopping threshold.

To see the intuition, suppose \(N=5\), \(i=1\),
\(\mathcal S_1^n=\{2,3\}\), \(\mathcal S_1^p=\{4\}\), and
\(\mathcal S_1^d=\{5\}\). Then
\[
y_{1t}
=
\boldsymbol{\beta}_1'\mathbf x_{1t}
+
\alpha_{12}y_{2t}
+
\alpha_{13}y_{3t}
+
\varepsilon_{1t}.
\]
For convenience, define the normalised population IV screening signal
\begin{equation}
m_{i,j}(H)
:=
\frac{
E\!\left(
\widetilde y_{j,H,t}
y_{i,t}
\right)
}{
\sigma_{i,j,H}
\left\{
E\!\left[
(\widetilde y_{j,H,t})^2
\right]
\right\}^{1/2}
}.
\label{eq:m_def}
\end{equation}
Also define the corresponding normalised IV alignment
\[
\bar\gamma_{i,j,k}(H)
=
\frac{
E(\widetilde y_{j,H,t}y_{k,t})
}{
\sigma_{i,j,H}
\left\{
E\!\left[
(\widetilde y_{j,H,t})^2
\right]
\right\}^{1/2}
}.
\]

At the first stage, the normalised population IV screening signals are
\[
m_{12}
=
\alpha_{12}\bar\gamma_{1,2,2}
+
\alpha_{13}\bar\gamma_{1,2,3},
\qquad
m_{13}
=
\alpha_{12}\bar\gamma_{1,3,2}
+
\alpha_{13}\bar\gamma_{1,3,3},
\]
whereas the proxy signal is
\[
m_{14}
=
\alpha_{12}\bar\gamma_{1,4,2}
+
\alpha_{13}\bar\gamma_{1,4,3}.
\]
Thus the required dominance condition is
\[
\max\{|m_{12}|,|m_{13}|\}>|m_{14}|.
\]
A proxy link may be correlated with several omitted true links simultaneously. When this occurs, its population screening signal combines information from multiple true neighbours and can, in principle, exceed the marginal signal of an individual omitted true link. If there is only one omitted true link, dominance reduces to the explicit normalised-alignment comparison below; it is not automatic with candidate-specific instruments and normalisations.\footnote{
To see this, suppose that the omitted true-link set is
\(\mathcal S_i^n\setminus H=\{s\}\) and let
\(p\in\mathcal S_i^p\setminus H\) be a proxy link. The structural equation is
\[
y_{i,t}
=
\boldsymbol{\beta}_i'\mathbf x_{i,t}
+
\alpha_{i,s}y_{s,t}
+
\varepsilon_{i,t}.
\]
Hence the normalised population screening signals are
\[
m_{i,s}
=
\frac{
\alpha_{i,s}E\!\left(\widetilde y_{s,H,t}y_{s,t}\right)
}{
\sigma_{i,s,H}
\{E[(\widetilde y_{s,H,t})^2]\}^{1/2}
},
\qquad
m_{i,p}
=
\frac{
\alpha_{i,s}E\!\left(\widetilde y_{p,H,t}y_{s,t}\right)
}{
\sigma_{i,p,H}
\{E[(\widetilde y_{p,H,t})^2]\}^{1/2}
}.
\]
Therefore,
\[
\frac{|m_{i,p}|}{|m_{i,s}|}
=
\frac{
\left|
E\!\left(\widetilde y_{p,H,t}y_{s,t}\right)
\right|
}{
\left|
E\!\left(\widetilde y_{s,H,t}y_{s,t}\right)
\right|
}
\cdot
\frac{
\sigma_{i,s,H}
\left\{
E[(\widetilde y_{s,H,t})^2]
\right\}^{1/2}
}{
\sigma_{i,p,H}
\left\{
E[(\widetilde y_{p,H,t})^2]
\right\}^{1/2}
}.
\]
Thus, in the single-link case, signal-to-proxy dominance reduces to comparing
the normalised population IV alignment of the proxy with that of the true link.
Dominance holds precisely when the displayed ratio is uniformly below one.
}

\begin{remark}
The appendix first establishes uniform probability bounds over the deterministic class of low-dimensional conditioning sets \(\mathcal H_i\), and then shows that
the realised BOLMT path remains in this class with probability tending to one. The tail assumption is imposed on score and Gram arrays, rather than only on
primitive variables, because the screening statistics involve products such as
\(v^*_{i,j,H,t}\varepsilon_{i,t}\). Finite moments of primitive variables do not generally imply the same moment order for such products.
\end{remark}

\begin{assumption}\label{ass:panel}

For each equation $i$, let
$\mathbf c^o_{i,t}$ denote the $d_i\times1$ oracle regressor vector constructed using the true link set $\mathcal S_i^n$, and let $\mathbf z^o_{i,t}$ denote the corresponding $m_i\times1$ oracle instrument vector, with $m_i\ge d_i$. The dimensions $d_i$ and $m_i$ are uniformly bounded in $i$. The oracle equation is
\(
y_{i,t}
=
\mathbf c_{i,t}^{o\prime}
\boldsymbol{\gamma}_i
+
\varepsilon_{i,t}.
\)
There exists a fixed matrix $\mathbf R_i$, with uniformly bounded norm, such that the
equation-specific parameter is
\[
\boldsymbol\theta_i
=
\mathbf R_i\boldsymbol\gamma_i
=
(\boldsymbol\beta_i',\rho_{1,i},\rho_{2,i})'.
\]
The sequence $\{\boldsymbol\theta_i\}_{i=1}^N$ is treated as fixed, and the finite-population
mean-group target is
\[
\boldsymbol\theta_N
=
N^{-1}\sum_{i=1}^N\boldsymbol\theta_i.
\]
The following conditions hold.
\begin{enumerate}[label=(\roman*)]
\item
For every $i$,
\(
\E(
\mathbf z^o_{i,t}\varepsilon_{i,t}
)
=
\mathbf 0.
\)
\item
The process
$\{(\mathbf c^o_{i,t},\mathbf
z^o_{i,t},\varepsilon_{i,t})\}_{t\ge1}$
is strictly stationary and exponentially $\alpha$-mixing uniformly in $i$. There exists $\delta>0$ such that
\(
\sup_i
\E
\|
\mathbf z^o_{i,t}\varepsilon_{i,t}
\|^{2+\delta}
<\infty,
\)
\(
\sup_i
\E
\|
\mathbf z^o_{i,t}\mathbf c_{i,t}^{o\prime}
\|^{1+\delta}
<\infty,
\qquad
\sup_i
\E
\|
\mathbf z^o_{i,t}\mathbf z_{i,t}^{o\prime}
\|^{1+\delta}
<\infty.
\)
\item
Define
\(
\boldsymbol{A}_i
=
\E(
\mathbf z^o_{i,t}\mathbf c_{i,t}^{o\prime}
),
\qquad
\mathbf B_i
=
\E(
\mathbf z^o_{i,t}\mathbf z_{i,t}^{o\prime}
),
\)
\(
\mathbf Q_i
=
\boldsymbol{A}_i'
\mathbf B_i^{-1}
\boldsymbol{A}_i.
\)
There exists $c>0$ such that uniformly in $i$,
\(
\lambda_{\min}(\mathbf B_i)\ge c,
\qquad
\lambda_{\min}(\mathbf Q_i)\ge c,
\)
and
\(
\lambda_{\max}(\mathbf B_i)\le c^{-1},
\qquad
\lambda_{\max}(\mathbf Q_i)\le c^{-1}.
\)
\item
The sample matrices satisfy, uniformly in $i$,
\(
\left\|
T^{-1}
\sum_{t=1}^T
\mathbf z^o_{i,t}\mathbf c_{i,t}^{o\prime}
-
\boldsymbol{A}_i
\right\|
=
o_p(1),
\)
and
\(
\left\|
T^{-1}
\sum_{t=1}^T
\mathbf z^o_{i,t}\mathbf z_{i,t}^{o\prime}
-
\mathbf B_i
\right\|
=
o_p(1).
\)
Writing the two sample matrices as $\widehat{\boldsymbol A}_i$ and
$\widehat{\mathbf B}_i$, respectively, define
\[
\widehat{\mathbf L}_i
=
(\widehat{\boldsymbol A}_i'\widehat{\mathbf B}_i^{-1}
\widehat{\boldsymbol A}_i)^{-1}
\widehat{\boldsymbol A}_i'\widehat{\mathbf B}_i^{-1},
\qquad
\mathbf L_i
=
\mathbf Q_i^{-1}\boldsymbol A_i'\mathbf B_i^{-1}.
\]
The aggregate linearisation remainder satisfies
\[
N^{-1/2}\sum_{i=1}^N
\mathbf R_i(\widehat{\mathbf L}_i-\mathbf L_i)
T^{-1/2}\sum_{t=1}^T\mathbf z^o_{i,t}\varepsilon_{i,t}
=o_p(1).
\]
\item
Let
\(
\mathbf \Omega_i
=
\sum_{\ell=-\infty}^{\infty}
\E
\left(
\mathbf z^o_{i,t}
\varepsilon_{i,t}
\varepsilon_{i,t-\ell}
\mathbf z_{i,t-\ell}^{o\prime}
\right).
\)
The series is absolutely summable uniformly in $i$, and
\(
0<c
\le
\lambda_{\min}(\mathbf\Omega_i)
\le
\lambda_{\max}(\mathbf\Omega_i)
\le
c^{-1}
<
\infty
\)
uniformly in $i$.
\item
Let
\(
\mathbf \Xi_i
=
\mathbf R_i
\mathbf Q_i^{-1}
\boldsymbol{A}_i'
\mathbf B_i^{-1}
\mathbf \Omega_i
\mathbf B_i^{-1}
\boldsymbol{A}_i
\mathbf Q_i^{-1}
\mathbf R_i'.
\)
The average variance satisfies
\(
\mathbf V_N
=
N^{-1}
\sum_{i=1}^N
\mathbf \Xi_i
\to
\mathbf V,
\)
where $\mathbf V$ is finite and positive definite.

Moreover, either the oracle influence functions are independent across $i$, or they are weakly cross-sectionally dependent and satisfy
\[
N^{-1/2}
\sum_{i=1}^N
\mathbf R_i
\mathbf Q_i^{-1}
\boldsymbol{A}_i'
\mathbf B_i^{-1}
T^{-1/2}
\sum_{t=1}^T
\mathbf z^o_{i,t}
\varepsilon_{i,t}
\Rightarrow
\mathcal N(\mathbf 0,\mathbf V).
\]

\end{enumerate}
\end{assumption}

\begin{remark}
Assumption~\ref{ass:panel} concerns only the oracle finite-dimensional IV estimator, that is, the estimator that would be computed if the true link set $\mathcal S_i^n$ were known. Conditions (i)--(iv) are standard IV exogeneity, relevance, and sample-moment conditions. Condition (v) accommodates serial dependence through the long-run covariance matrix of the IV score, while condition (vi) provides the cross-sectional averaging condition required for mean-group inference.
Under Assumption~\ref{ass:panel}, the oracle equation-level estimator admits an asymptotic linear representation, and the corresponding oracle mean-group estimator is asymptotically normal around the finite-population target $\boldsymbol\theta_N$. The role of the network-selection theory is to establish that the feasible post-selection estimator coincides with its oracle counterpart with probability tending to one. Consequently, selection does not alter the oracle asymptotic distribution; feasible inference also requires a consistent estimator of $\mathbf V$.
\end{remark}


\subsection{Exact Network Recovery}

Let $\widehat{\mathcal S}_i^t$ denote the set selected by the BOLMT procedure based on the scalar screening statistic $t_{i,j,H}$ and threshold \(c_N^t=C_t\sqrt{\log N}\), where $C_t$ is a sufficiently large fixed constant. Define the exact-recovery event
\begin{equation}
\mathcal A_1^t
=
\left\{
\widehat{\mathcal S}_i^t
=
\mathcal S_i^n
\text{ for all } i
\right\}.
\label{eq:A1t}
\end{equation}

\begin{theorem}
\label{thm:t_exact}
Suppose Assumptions~\ref{ass:sampling}--\ref{ass:tails},
\ref{ass:t_signal}, and \ref{ass:t_dominance} hold. Then
\[
\Pr(\mathcal A_1^t)
\to
1.
\]
\end{theorem}
Theorem~\ref{thm:t_exact} establishes exact recovery of the structural support of every unit. Assumption~\ref{ass:t_signal} prevents premature stopping by requiring that, before
the true support has been fully recovered, at least one omitted true link is sufficiently
detectable to cross the high-dimensional screening threshold. Assumption~\ref{ass:t_dominance} ensures correct selection and stopping: true links dominate
proxy candidates before completion, while remaining proxy links do not cross the threshold
after the true support has been selected.

\begin{remark}\label{rem:approximation}
Exact recovery requires both components of Assumption~\ref{ass:t_dominance}. If the
pre-completion dominance condition is maintained but the post-completion proxy-control
condition is relaxed, the procedure may continue to select proxy links after all true links
have entered. In this case, the selected support remains contained in
$\mathcal K_i=\mathcal S_i^n\cup\mathcal S_i^p$ with probability approaching one and
contains the true support asymptotically, although some proxy links may also be included.
This support-containment property remains useful for estimation. Since proxy links do not
enter the population equation directly, post-selection IV estimation continues to identify
the nonzero interaction coefficients associated with the true links, while the coefficients
on selected proxy links converge to zero. Proposition~\ref{prop:containment} formalises
this result. The Monte Carlo results reported in Section~\ref{sec:mc} illustrate that
estimation can remain accurate even when exact recovery is not achieved in every replication.
\end{remark}

For Proposition~\ref{prop:containment}, assume that every augmented selected model satisfying
\[
\mathcal S_i^n\subseteq\widehat{\mathcal S}_i\subseteq \mathcal S_i^n\cup\mathcal S_i^p
\]
satisfies, uniformly over this finite class, the same IV exogeneity, rank, sample-moment, and long-run variance conditions as the
oracle model, with dimensions uniformly bounded.

\begin{proposition}
\label{prop:containment}
Suppose that
\[
\Pr\!\left(
\mathcal S_i^n \subseteq \widehat{\mathcal S}_i^t
\subseteq\mathcal S_i^n\cup\mathcal S_i^p
\text{ for all }i
\right)\to1.
\]
Then the post-selection IV estimator is consistent for the coefficients
associated with the true links. Moreover, for any selected proxy link
$j\in\widehat{\mathcal S}_i^t\setminus\mathcal S_i^n$,
\[
\widehat{\alpha}_{i,j}
\stackrel{p}{\to}
0.
\]
\end{proposition}

\subsection{Recovery of the Dual-Network Structure}

The preceding results establish recovery of the structural interaction matrix
\(\boldsymbol{A}\). As discussed previously, recovery of $\boldsymbol{A}$
is the primary objective of BOLMT and is sufficient in conventional single-network settings. To recover the reinforcing and displacement networks separately, additional structure is required.

\begin{assumption}
\label{ass:stability}
The spectral radius of $\boldsymbol{A}$ satisfies
\[
r(\boldsymbol{A})<1.
\]
\end{assumption}
Assumption~\ref{ass:stability} ensures that the higher-order feedback effects generated by
$(\mathbf I_N-\boldsymbol A)^{-1}
= \mathbf I_N+\boldsymbol A+\boldsymbol A^2+\cdots$
are represented by a convergent Neumann series for each $N$. This assumption is not used
for support recovery or for the sign decomposition below.


\begin{assumption}
\label{ass:dual}
The interaction matrix admits the decomposition
\[
\boldsymbol{A}
=
\rho_1\mathbf W_1
+
\rho_2\mathbf W_2,
\]
where

\begin{enumerate}[label=(\roman*)]
\item $\rho_1>0$ and $\rho_2<0$;

\item $\mathbf W_1$ and $\mathbf W_2$ are non-negative matrices;

\item the supports of $\mathbf W_1$ and $\mathbf W_2$ are disjoint, namely
\[
w_{1,i,j}w_{2,i,j}=0,
\qquad
i\neq j;
\]

\item each nonzero row of $\mathbf W_1$ and $\mathbf W_2$ sums to one.
\end{enumerate}
\end{assumption}

For $\ell=1,2$, define the population signed support and row intensity by
\[
\mathcal S_{1,i}=\{j:\alpha_{i,j}>0\},
\qquad
\mathcal S_{2,i}=\{j:\alpha_{i,j}<0\},
\qquad
\rho_{\ell,i}=\sum_{j\in\mathcal S_{\ell,i}}\alpha_{i,j},
\]
and set
\[
w_{\ell,i,j}
=
\begin{cases}
\alpha_{i,j}/\rho_{\ell,i},
&j\in\mathcal S_{\ell,i}\ \text{and}\ \mathcal S_{\ell,i}\neq\varnothing,\\
0,&\text{otherwise}.
\end{cases}
\]
Under Assumption~\ref{ass:dual}, $\rho_{\ell,i}=\rho_\ell$ on every active row of
$\mathbf W_\ell$, while inactive rows are zero.

\begin{theorem}
\label{thm:dualrecover}
Suppose the conditions of Theorem~\ref{thm:t_exact} hold together with Assumption
\ref{ass:dual}. Suppose further that the post-selection coefficient estimates satisfy
\[
\max_i\max_{j\in\mathcal S_i^n}
|\widehat\alpha_{i,j}-\alpha_{i,j}|=o_p(b_T\wedge 1).
\tag{PS}
\label{eq:post_sign_rate}
\]
Then
\[
\Pr
\left(
\widehat{\mathcal S}_{\ell,i}=\mathcal S_{\ell,i}
\text{ for all }i\text{ and }\ell=1,2
\right)
\to
1.
\]
Moreover,
\[
\max_i|\widehat\rho_{\ell,i}-\rho_{\ell,i}|\to_p0,
\qquad
\max_i\sum_{j\ne i}|\widehat w_{\ell,i,j}-w_{\ell,i,j}|\to_p0,
\qquad \ell=1,2.
\]
\end{theorem}
Theorem~\ref{thm:dualrecover} shows that exact support recovery, together with uniform coefficient accuracy, implies exact recovery of the signed supports and consistent estimation of the reinforcing and displacement weights. The sign restrictions and disjoint-support conditions ensure that each recovered link can be assigned uniquely to one of the two interaction channels.
The decomposition theorem is stated under common interaction intensities $(\rho_1,\rho_2)$ for expositional simplicity. In empirical applications, however, the interaction intensities
and slope coefficients may vary across equations. The post-selection estimation step therefore permits equation-specific parameters and aggregates them using Mean Group estimation. Since support recovery depends only on the interaction matrix $\boldsymbol{A}$, this extension does not affect
the preceding recovery theory.

The results above are asymptotic in their nature and are stated under exact sparsity of the relevant local interaction structure. Approximate sparsity can be accommodated by allowing sufficiently small interaction effects to enter the proxy set or approximation error. In that case, however, the target becomes an approximating interaction structure rather than exact support recovery. The present formulation keeps the structural target explicit.

\subsection{Post-Selection Estimation}

Let $\widehat{\boldsymbol\theta}_i$ be the equation-level post-selection IV estimator
and let $\widehat{\boldsymbol\theta}_{MG}$ denote the corresponding mean-group estimator, both defined in Eq. \eqref{eq:mg_estimator}.

\begin{theorem}
\label{thm:eqoracle}
Suppose the assumptions of Theorem~\ref{thm:t_exact} and condition Eq. \eqref{eq:post_sign_rate} hold.
Then, uniformly in $i$,
\[
\widehat{\boldsymbol\theta}_i
-
\widehat{\boldsymbol\theta}_i^o
=
o_p(T^{-1/2}),
\]
where $\widehat{\boldsymbol\theta}_i^o$ denotes the oracle estimator based on the true support set $\mathcal S_i^n$.
\end{theorem}

\begin{theorem}\label{thm:t_panel}
Suppose the assumptions of Theorem~\ref{thm:eqoracle} and Assumption~\ref{ass:panel} hold. Then
\[
\sqrt{NT}(\widehat{\boldsymbol{\theta}}_{MG}-\boldsymbol{\theta}_N)
\Rightarrow
\mathcal N(0,\mathbf{V}),
\]
with the same limiting variance as the oracle mean-group estimator.
\end{theorem}

Theorems~\ref{thm:eqoracle} and~\ref{thm:t_panel} show that network selection does not affect first-order inference for the fixed finite-population target $\boldsymbol\theta_N$. Since exact recovery holds with probability approaching one, the feasible post-selection IV estimator coincides asymptotically with the oracle estimator, and the mean-group estimator has the same limiting distribution as if the true support sets were known. Feasible inference additionally requires a consistent estimator of $\mathbf V$ under the maintained cross-sectional dependence conditions.

\section{Extensions and discussions}
\label{sec:discussion}

\subsection{Block Screening for General Network Models}
\label{sec:block_screening}

The framework developed above extends naturally to settings in which each
candidate link contributes multiple regressors. Such specifications arise in a
variety of network and spatial models. For example, a single interaction link
\((i,j)\) may generate both contemporaneous and lagged regressors,
\[
\alpha_{i,j}y_{j,t}
+
\gamma_{i,j}y_{j,t-1},
\]
while spatial Durbin specifications may augment interaction effects with
neighbours' covariates,
\[
\alpha_{i,j}y_{j,t}
+
\boldsymbol{\delta}_{i,j}'\mathbf x_{j,t}.
\]
More generally, a candidate link may contribute a fixed-dimensional block of
regressors rather than a single interaction term. In such settings, the object
of interest is not an individual coefficient but the presence or absence of a
link. Consequently, screening is performed at the block level by jointly
testing all regressors associated with a candidate link.

For candidate link \(j\), let
\[
\mathbf v_{i,j,H,t}
=
(v_{i,j,H,1,t},\ldots,v_{i,j,H,d,t})',
\]
denote the corresponding block of regressors, where the block dimension
\(d\) is fixed. Examples include
\[
\mathbf v_{i,j,H,t}
=
(y_{j,t},y_{j,t-1})'
\]
in dynamic network models, and
\[
\mathbf v_{i,j,H,t}
=
(y_{j,t},\mathbf x_{j,t}')'
\]
in spatial Durbin specifications.

The theoretical analysis closely parallels the scalar case. The key difference
is that signal strength is now measured by a population Wald criterion rather
than a scalar IV screening signal.

Define the population fitted block
\[
\bar{\mathbf V}^{*}_{i,j,H,t}
=
\boldsymbol{\Pi}_{V,i,j,H}'
\boldsymbol{\mathcal Z}_{i,j,H,t},
\qquad
\boldsymbol{\Pi}_{V,i,j,H}
=
(\mathbf Q^{\boldsymbol{\mathcal Z}}_{i,j,H})^{-1}
\E\!\left(
\boldsymbol{\mathcal Z}_{i,j,H,t}
\mathbf v_{i,j,H,t}'
\right),
\]
and let
\[
\boldsymbol{\Gamma}_{VC,i,j,H}
=
\left\{
\E\!\left(
\mathbf C^*_{i,j,H,t}
\mathbf C_{i,j,H,t}^{*\prime}
\right)
\right\}^{-1}
\E\!\left(
\mathbf C^*_{i,j,H,t}
\bar{\mathbf V}_{i,j,H,t}^{*\prime}
\right).
\]
The corresponding population partial fitted block is
\[
\mathbf V^*_{i,j,H,t}
=
\bar{\mathbf V}^*_{i,j,H,t}
-
\boldsymbol{\Gamma}_{VC,i,j,H}'
\mathbf C^*_{i,j,H,t},
\]
namely the fitted block after projecting out the fitted selected controls.

In this subsection, let $\boldsymbol\vartheta_{i,k}$ denote the coefficient vector on
the regressor block contributed by link $k$, and let
$\mathcal S_i^n=\{k:\|\boldsymbol\vartheta_{i,k}\|>0\}$ denote the corresponding
group support; the candidate partition and conditioning-set notation are understood
groupwise. Reinterpret the omitted component as
\[
\mu_{i,j,H,t}:=\mu^F_{i,H,t}
=
\sum_{k\in\mathcal S_i^n\setminus H}
\boldsymbol\vartheta_{i,k}'\mathbf v_{i,k,H,t}.
\]
Likewise, $\widehat\sigma_{i,j,H}^2$ and $\sigma_{i,j,H}^2$ denote, respectively,
the residual-variance estimator from the corresponding block-stage IV regression and
its population limit.

Define
\[
\mathbf G^*_{i,j,H}
=
\E\!\left(
\mathbf V^*_{i,j,H,t}
\mathbf V_{i,j,H,t}^{*\prime}
\right),
\qquad
\mathbf m^F_{i,j,H}
=
\E\!\left(
\mathbf V^*_{i,j,H,t}
\mu_{i,j,H,t}
\right),
\]
where \(\mathbf m^F_{i,j,H}\) is the population block screening score vector.
Define the normalised population block screening signal
\begin{equation}
Q^F_{i,j,H}
=
\frac{1}{\sigma_{i,j,H}^{2}}
\mathbf m_{i,j,H}^{F\prime}
(\mathbf G^*_{i,j,H})^{-1}
\mathbf m^F_{i,j,H}.
\label{eq:QF}
\end{equation}

Let
\[
\mathbf S_{i,j,H}=T^{-1/2}\widehat{\mathbf V}_{i,j,H}'\widetilde{\mathbf y}_{i,j,H},
\qquad
\mathbf G_{i,j,H}=T^{-1}\widehat{\mathbf V}_{i,j,H}'\widehat{\mathbf V}_{i,j,H}.
\]
The sample block screening quadratic is
\begin{equation}
F_{i,j,H}
=
\widehat\sigma_{i,j,H}^{-2}
\mathbf S_{i,j,H}'\mathbf G_{i,j,H}^{-1}\mathbf S_{i,j,H},
\label{eq:Fstat}
\end{equation}
which corresponds to the homoskedastic fitted-regressor covariance normalisation. Under
serial dependence it is used as a screening quadratic, not as a finite-sample Wald statistic.

For the block results, Assumptions~\ref{ass:sampling} and~\ref{ass:tails} are
imposed also on the centred scalar components of
\[
V^*_{i,j,H,t,a}\varepsilon_{i,t},
\qquad
V^*_{i,j,H,t,a}\mu_{i,j,H,t},
\qquad
V^*_{i,j,H,t,a}V^*_{i,j,H,t,b},
\]
where \(\mathbf V^*_{i,j,H,t}\) denotes the population partial fitted block.

\begin{assumption}
\label{ass:block_first_stage}
The Gram and sample-moment conditions in Assumptions~\ref{ass:iv}
and~\ref{ass:first_stage} are also imposed blockwise. In addition,
\[
\inf_{i,j,H}
\lambda_{\min}\left\{\E(\mathbf V^*_{i,j,H,t}
\mathbf V_{i,j,H,t}^{*\prime})\right\}>0,
\]
and the same uniform population variance bounds imposed in
Assumption~\ref{ass:iv} hold for the corresponding block-stage IV regressions.
Moreover, uniformly over low-dimensional conditioning sets $H\in\mathcal H_i$ and candidates
$j\notin H$,
\[
\sup_{i,j,H}
\left\|
T^{-1}\sum_{t=1}^T
\boldsymbol{\mathcal Z}_{i,j,H,t}\mathbf v_{i,j,H,t}'
-
\E(\boldsymbol{\mathcal Z}_{i,j,H,t}\mathbf v_{i,j,H,t}')
\right\|=o_p(1),
\]
\[
\sup_{i,j,H}
T^{-1}\sum_{t=1}^T
\|\widehat{\mathbf V}_{i,j,H,t}-\mathbf V^*_{i,j,H,t}\|^2
=o_p(1),
\]
and, componentwise,
\[
\sup_{i,j,H}\max_{1\le a\le d}
\left|
T^{-1/2}\sum_{t=1}^T
(\widehat V_{i,j,H,t,a}-V^*_{i,j,H,t,a})\varepsilon_{i,t}
\right|
=o_p(\sqrt{\log N}).
\]
The same componentwise bound holds with $\varepsilon_{i,t}$ replaced by
$\mu_{i,j,H,t}$.
In addition,
\[
\sup_{i,j,H}
|\widehat\sigma_{i,j,H}^2-\sigma_{i,j,H}^2|
=o_p(1).
\]
\end{assumption}

\begin{assumption}[Block true-link detectability]
\label{ass:F_signal}
For block screening, true links are detectable before completion:
\[
\inf_{\{(i,H):\,H\in\mathcal H_i^0\}}
\max_{j\in U_i(H)}
\frac{TQ^F_{i,j,H}}{\log N}
\to\infty.
\]
If \(\mathcal H_i^0=\varnothing\) for some equation \(i\), the condition is
vacuous for that equation.
In addition, irrelevant candidate blocks have negligible population screening signal:
\[
\sup_i\sup_{H\in\mathcal H_i}\max_{j\in\mathcal S_i^d}
\frac{TQ^F_{i,j,H}}{\log N}\to0.
\]
\end{assumption}

\begin{assumption}[Primitive block proxy-alignment condition]
\label{ass:F_dominance}
For block exact recovery, the following two conditions hold.

First, before completion, proxy block signals are uniformly bounded relative to
the strongest remaining true-link block signal. There exists a constant
\(\lambda_F\in[0,1)\) such that, uniformly over \(i\) and
\(H\in\mathcal H_i^0\),
\[
\max_{p\in P_i(H)} Q^F_{i,p,H}
\le
\lambda_F
\max_{s\in U_i(H)} Q^F_{i,s,H}.
\]

Second, after completion, remaining proxy links have no residual block
screening signal:
\[
\max_{p\in P_i(H)}
\frac{TQ^F_{i,p,H}}{\log N}
\to0
\]
uniformly over \(i\) and \(H\in\mathcal H_i^1\).
\end{assumption}

For the block-screening results, Assumption~\ref{ass:tails} is understood to apply also to the centred block score and Gram arrays entering the Wald statistic.

Let $\widehat{\mathcal S}_i^F$ denote the set selected by the block BOLMT procedure based on the statistic $F_{i,j,H}$ and threshold $c_N^F=C_F\log N$, where $C_F$ is a sufficiently large fixed constant.
Define the exact-recovery event
\begin{equation}
\mathcal A_1^F
=
\left\{
\widehat{\mathcal S}_i^F
=
\mathcal S_i^n
\text{ for all } i
\right\}.
\label{eq:A1F}
\end{equation}

\begin{theorem}
\label{thm:F_exact}
Suppose Assumptions~\ref{ass:sampling}--\ref{ass:tails},
\ref{ass:block_first_stage},
\ref{ass:F_signal},
and
\ref{ass:F_dominance}
hold.
Then
\[
\Pr(\mathcal A_1^F)
\to
1.
\]
\end{theorem}

Theorem~\ref{thm:F_exact} establishes exact recovery of the structural support under block screening. The result applies to dynamic network models, spatial Durbin specifications, and more general settings in which each link contributes a fixed-dimensional block of regressors.

The block results can be viewed as tests of a candidate group.  The population noncentrality \(Q^F_{i,j,H}\) aggregates information across the components of the block.  A block may be powerful when the signal is spread across several variables or transformations, even if no single component has a dominant scalar signal.  Because the statistic is quadratic, the threshold is of order \(\log N\), and the proxy-dominance condition must dominate both the \(\log N\) stochastic component and the quadratic cross term in the block expansion.

\subsection{Common Factors}

Many applications involve pervasive common shocks that generate cross-sectional dependence beyond the network structure itself. Suppose the disturbance admits the factor representation
\[
\varepsilon_{i,t}
=
\boldsymbol{\lambda}_i'
\mathbf f_t
+
u_{i,t},
\]
where the number of factors is fixed.

A factor-augmented implementation includes observed factors or estimated factor proxies among the always-included regressors and in the first-stage projections. Let $\widehat{\mathbf F}$ denote the matrix of observed or estimated factors and let $\mathbf M_{\widehat F}$ denote the corresponding residual-maker matrix. The residual factor component must be negligible on the screening-score scale. For scalar screening this requires
\[
\sup_{i,j,H}
\left|
T^{-1/2}
\sum_{t=1}^T
v^*_{i,j,H,t}
(\mathbf M_{\widehat F}\mathbf F\boldsymbol{\lambda}_i)_t
\right|
=
o_p(\sqrt{\log N}),
\]
and for block screening it requires
\[
\sup_{i,j,H}\max_{1\le a\le d}
\left|
T^{-1/2}
\sum_{t=1}^T
V^*_{i,j,H,t,a}
(\mathbf M_{\widehat F}\mathbf F\boldsymbol{\lambda}_i)_t
\right|
=
o_p(\sqrt{\log N}).
\]
If, in addition, the factor-purged idiosyncratic errors satisfy the same mixing, tail,
first-stage, and signal conditions as the baseline disturbances, and the feasible
factor-replacement bounds in Lemma~\ref{lem:app_factor} hold, then the scalar and block
selection results apply to the purged variables.

The role of factor adjustment is particularly important in the present setting because latent common factors provide a natural source of proxy links. Even when a candidate unit does not enter the structural equation directly, it may exhibit a nonzero screening signal if its outcome co-moves with the true links through common macroeconomic, sectoral, or market-wide shocks. In this sense, factor-induced dependence can create spurious network connections that resemble genuine interaction effects.

\begin{remark}
Suppose that the outcomes of different units are driven partly by a common factor. Then candidate links that do not belong to the true support of equation $i$ may nevertheless exhibit nonzero screening signals because they co-move with the true links through the common factor.
In this setting, Assumption~\ref{ass:t_dominance} requires that the direct contribution of at least one omitted true link remains larger than the indirect contribution generated by any proxy link through the common factor. This condition is plausible whenever the common factor does not induce stronger dependence between a proxy link and the outcome than the dependence generated by the underlying structural interaction itself. In such cases, factor adjustment further strengthens the separation between true and proxy links by removing a common source of cross-sectional dependence.
\end{remark}

The preceding argument establishes that the screening and recovery results continue to hold after factor adjustment, provided that the residual factor component is asymptotically negligible and the purged idiosyncratic errors satisfy the maintained regularity conditions. The extension is particularly relevant in economic and financial applications where latent network interactions coexist with pervasive common shocks.


\subsection{Scalar versus Block Screening}

Scalar screening is natural when each candidate link contributes a single interaction regressor. This includes the dual-network model developed in Sections~2--4, where each candidate contributes a contemporaneous interaction term.
Block screening is appropriate when each candidate link contributes multiple regressors simultaneously. Examples include dynamic network models involving contemporaneous and lagged interaction effects, spatial Durbin specifications involving links' covariates, and more general contextual-effect models.

The two procedures need not select identical interaction structures in finite samples. A scalar statistic may select a link because one component generates a strong signal. By contrast, a block statistic may select a link whose joint contribution is substantial even when each individual component is only moderately informative. In empirical applications it may therefore be useful to examine the robustness of the selected interaction structure across scalar and block implementations.

\subsection{Scope of the Theory}

The results are asymptotic and are stated under exact sparsity of the relevant local interaction structure. Approximate sparsity can be accommodated by allowing sufficiently small interaction effects to enter the proxy set or approximation error. In that case, however, the target becomes an approximating interaction structure rather than exact support recovery. The present formulation keeps the structural target explicit.

The theory is fundamentally local and proceeds equation by equation. It does not require global graph restrictions such as symmetry, bounded degree, connectedness, or particular topological features of the interaction network. Such restrictions may be useful in specific applications, but they are not needed for the row-wise selection arguments developed here.

Finally, the framework is agnostic regarding the economic mechanism generating the interaction structure. The same methodology applies to spatial, social, financial, production, trade, or other network environments, provided suitable instruments are available and the required signal conditions are satisfied.

\subsection{Implementation}

In empirical applications, it is useful to report the estimated degree distribution, the sequence of selected links for representative equations, and diagnostics for the associated first-stage regressions. Weak instruments affect both screening and post-selection estimation and may compromise network recovery. In particular, if the fitted-regressor variance is close to zero, the relevance conditions underlying the screening statistics may not be credible for that candidate.

The scalar critical value can be implemented using the threshold $c_p(N,\delta)$ defined in
Eq. \eqref{eq:cvf}. For block screening, the analogous critical value is based on Wald-type
quantiles with tail probability proportional to $N^{-\delta}$. These choices correspond to
thresholds of order $\sqrt{\log N}$ for scalar statistics and $\log N$ for block statistics.
The theory does not rely on exact finite-sample normal or Wald distributions. The thresholds
should instead be interpreted as high-dimensional screening thresholds; uniform false-selection
control requires the sufficiently large leading constants specified in the recovery theorems.

\section{Monte Carlo Study}\label{sec:mc}

\subsection{Design}

We examine the finite-sample performance of the proposed BOLMT
procedure through Monte Carlo experiments based on the following
dual-network interaction model:
\begin{equation}
y_{i,t}
=
\eta_i
+
\beta_1 x_{1,i,t}
+
\beta_2 x_{2,i,t}
+
\rho_1
\sum_{j\neq i}
w_{1,i,j}y_{j,t}
+
\rho_2
\sum_{j\neq i}
w_{2,i,j}y_{j,t}
+
\varepsilon_{i,t},
\label{eq:MCdgp1}
\end{equation}
for
$i=1,\ldots,N$
and
$t=1,\ldots,T$.
The disturbance term satisfies
$\varepsilon_{i,t} \sim
i.i.d. N(0,\sigma_{\varepsilon}^{2})$,
while the individual effects $\eta_i$ are generated independently from a standard normal distribution.

The covariates follow the processes
\begin{equation}
x_{\ell,i,t}
=
\eta_{\ell,i}
+
\rho_x x_{\ell,i,t-1}
+
\sqrt{1-\rho_x^2}\,
v_{\ell,i,t},
\qquad
\ell=1,2,
\label{eq:MCx}
\end{equation}
where the innovation vector
$\boldsymbol v_{\ell,t}
= (v_{\ell,1,t},\ldots,v_{\ell,N,t})'$
is cross-sectionally correlated according to
\[
\boldsymbol v_{\ell,t}
\sim
N(\boldsymbol 0,\boldsymbol\Sigma_v),
\qquad
(\Sigma_v)_{ij}
=
\sigma_v^2\rho_x^{|i-j|}.
\]
Thus, $\rho_x$ controls both the persistence of the covariates through the autoregressive component and the cross-sectional correlation of the innovations through $\boldsymbol\Sigma_v$.
The individual-specific covariate effects are correlated with the
fixed effects according to
\begin{equation}
\eta_{\ell,i}
=
\rho_{\eta}\eta_i
+
\sqrt{1-\rho_{\eta}^2}\,
\xi_{\ell,i},
\label{eq:MCeta}
\end{equation}
where
$\xi_{\ell,i}\sim i.i.d.N(0,1)$.
Accordingly, $\rho_{\eta}$ determines the degree of correlation between the covariates and the
individual effects.

We consider three alternative dual-network designs.

\medskip

\noindent
\textbf{Design 1: Circular network.}
Each unit is connected to four linking units arranged on a
circular lattice (see e.g., \citet{KapoorKelejianPrucha2007}).
Two links are assigned to the displacement network
$\boldsymbol W_1$
and two links to the reinforcing network
$\boldsymbol W_2$.
The displacement network is asymmetric, with row weights
$(0.75,0.25)$,
whereas the reinforcing network is symmetric with weights
$(0.5,0.5)$.

\medskip

\noindent
\textbf{Design 2: Block-diagonal network.}
Units are partitioned into independent groups of size five.
Within each block, each unit interacts with four other units.
Two links are assigned to
$\boldsymbol W_1$
and two links to
$\boldsymbol W_2$.
As in Design 1, the displacement network is asymmetric and the
reinforcing network is symmetric.

\medskip

\noindent
\textbf{Design 3: Random network with a dominant unit.}
For each unit, links are generated randomly subject to sparsity. Each unit is connected to two displacement links and two reinforcing links. One unit acts as a dominant hub in the reinforcing network, thereby generating heterogeneous network centrality and stronger dependence propagation.

In all designs, the supports of $\boldsymbol W_1$ and $\boldsymbol W_2$ are disjoint and each row is normalised to sum to unity within each network separately.

Let
\[
\boldsymbol A
=
\rho_1\boldsymbol W_1
+
\rho_2\boldsymbol W_2.
\]
The outcome process is generated according to
\[
\boldsymbol y_t
=
(\boldsymbol I-\boldsymbol A)^{-1}
(
\boldsymbol\eta
+
\beta_1\boldsymbol x_{1,t}
+
\beta_2\boldsymbol x_{2,t}
+
\boldsymbol\varepsilon_t
).
\]
The signal-to-noise ratio is calibrated using the spectral radius of
$\boldsymbol A$:
\begin{equation}
SNR
=
\frac{\psi_A}{\sigma_{\varepsilon}^{2}},
\qquad
\psi_A
=
\max
|\lambda(\boldsymbol A)|.
\label{eq:SNR}
\end{equation}
Throughout the experiments we fix $SNR=5$ and set $\sigma_{\varepsilon}^{2}=\psi_A/SNR$. The parameter values are fixed at $\beta_1=2, \beta_2=-1, \rho_1=0.5, \rho_2=-0.4$. We consider all combinations of $N\in\{25,50,100,200\}, T\in\{25,50,100,200\}$. Finally, for the critical value function in Eq.~(\ref{eq:cvf}), we set $\delta=1, c=1, p=0.05$. All experiments are based on $5,000$ Monte Carlo replications.

\subsection{Estimation}

All variables are first transformed using the within transformation in
order to eliminate the individual-specific fixed effects. For notational
simplicity, the transformed variables are denoted using the original
notation throughout.

In addition to the proposed BOLMT algorithm developed in this paper, we consider two alternative link-selection methods, based on the OCMT and MOCMT procedures of \citet{chu2018}. These procedures have recently been adapted to network estimation by \citet{LiBhattacharjee2024}. OCMT applies a single-stage multiple-testing selection rule, whereby all candidate links whose associated $t$-statistics exceed the critical value threshold are selected simultaneously. MOCMT extends this procedure to multiple iterations by updating the conditioning set after each selection round and then re-testing the remaining candidate links. Thus, both OCMT and MOCMT admit all statistically significant candidate links at a given iteration. By contrast, BOLMT proceeds in a more sequential and less greedy manner. At each step, if at least one remaining candidate exceeds the threshold, BOLMT selects only the single link associated with the largest absolute $t$-statistic before updating the conditioning set and continuing to the next step. If no remaining candidate exceeds the threshold, the procedure stops and the current selected set is retained. We index the three procedures by $h=1,2,3$, where $h=1$ corresponds to BOLMT, $h=2$ to OCMT, and $h=3$ to MOCMT.

For each unit $i$, the selected support set obtained from procedure $h$
is denoted by $\widehat{\mathcal S}^{(h)}_i$, where $h=1$ corresponds to BOLMT, $h=2$ to OCMT, and $h=3$ to MOCMT, with cardinality
$\widehat{k}^{(h)}_i
=
|\widehat{\mathcal S}^{(h)}_i|$.
Conditional on the selected support, we estimate the restricted
structural equation
\begin{equation}
\mathbf y_i
=
\mathbf X_i\boldsymbol\beta_i
+
\sum_{j\in\widehat{\mathcal S}^{(h)}_i}
\alpha_{i,j}\mathbf y_j
+
\boldsymbol\varepsilon_i,
\label{eq:MCpost}
\end{equation}
by instrumental variables.

Define the regressor matrix
\begin{equation}
\widehat{\mathbf C}^{(h)}_i
=
\left(
\mathbf X_i,
\{\mathbf y_j:j\in\widehat{\mathcal S}^{(h)}_i\}
\right),
\label{eq:Cihat}
\end{equation}
and the instrument matrix
\begin{equation}
\widehat{\mathbf Z}^{(h)}_i
=
\left(
\mathbf X_i,
\{\mathbf X_j:j\in\widehat{\mathcal S}^{(h)}_i\}
\right).
\label{eq:Zihat}
\end{equation}

The corresponding individual-specific IV estimator is given by
\begin{equation}
\widehat{\boldsymbol\theta}^{(h)}_{i}
=
\left(
\widehat{\mathbf G}^{(h)\prime}_{i,T}
(\widehat{\mathbf Q}^{(h)}_{i,T})^{-1}
\widehat{\mathbf G}^{(h)}_{i,T}
\right)^{-1}
\widehat{\mathbf G}^{(h)\prime}_{i,T}
(\widehat{\mathbf Q}^{(h)}_{i,T})^{-1}
\widehat{\mathbf g}^{(h)}_{i,T},
\label{eq:IVi}
\end{equation}
where
\begin{equation}
\widehat{\mathbf G}^{(h)}_{i,T}
=
\frac{1}{T}
\widehat{\mathbf Z}^{(h)\prime}_i
\widehat{\mathbf C}^{(h)}_i; \quad
\widehat{\mathbf Q}^{(h)}_{i,T}
=
\frac{1}{T}
\widehat{\mathbf Z}^{(h)\prime}_i
\widehat{\mathbf Z}^{(h)}_i; \quad
\widehat{\mathbf g}^{(h)}_{i,T}
=
\frac{1}{T}
\widehat{\mathbf Z}^{(h)\prime}_i
\mathbf y_i.
\end{equation}
The vector $\widehat{\boldsymbol\theta}^{(h)}_{IV,i}$ contains the IV estimates of the slope coefficients together with the estimated interaction coefficients associated with the selected links. Following the decomposition introduced in Section~\ref{sec:psel_estimation}, the estimated interaction coefficients are assigned to the reinforcing and displacement networks according to the sign restrictions imposed in
the dual-network representation, yielding the corresponding estimates of $\widehat\rho^{(h)}_{1,i}$ and $\widehat\rho^{(h)}_{2,i}$, together with the associated row-normalised interaction weights.

Finally, the mean-group estimator is constructed by averaging the
individual-specific estimates across cross-sectional units:
\begin{equation}
\widehat{\boldsymbol\theta}^{(h)}_{MG}
=
\frac{1}{N}
\sum_{i=1}^{N}
\widehat{\boldsymbol\theta}^{(h)}_i.
\label{eq:MG}
\end{equation}

\subsection{Results}

\subsubsection{Recovery of the network structure}
Table \ref{tab:summary_results} and Figure \ref{fig:MAD} report the finite-sample performance of BLOMT and the aforementioned alternative link-selection procedures across the various network designs. Figures \ref{fig:TPR}--\ref{fig:FDR} in Appendix~B provide additional visual evidence on the decomposition of network recovery performance into true and false link selection rates.

The recovery measures are defined in terms of the true and estimated support matrices.
Let
\[
s_{i,j}
=
\mathbf 1(\alpha_{i,j}\neq 0),
\qquad
\widehat s_{i,j}
=
\mathbf 1(\widehat\alpha_{i,j}\neq 0),
\qquad i\neq j,
\]
and let $\boldsymbol S=(s_{i,j})$ and $\widehat{\boldsymbol S}=(\widehat s_{i,j})$ denote the
corresponding true and estimated support matrices. The mean absolute deviation is
\[
\mathrm{MAD}
=
\frac{1}{N(N-1)}
\sum_{i=1}^N
\sum_{j\neq i}
\left|
\widehat s_{i,j}-s_{i,j}
\right|.
\]
Equivalently, MAD is the fraction of pairwise link decisions that are incorrectly classified.

Let
\[
\mathrm{TP}
=
\sum_{i=1}^N\sum_{j\neq i}
\mathbf 1(\widehat s_{i,j}=1,s_{i,j}=1),
\qquad
\mathrm{FP}
=
\sum_{i=1}^N\sum_{j\neq i}
\mathbf 1(\widehat s_{i,j}=1,s_{i,j}=0),
\]
\[
\mathrm{FN}
=
\sum_{i=1}^N\sum_{j\neq i}
\mathbf 1(\widehat s_{i,j}=0,s_{i,j}=1),
\qquad
\mathrm{TN}
=
\sum_{i=1}^N\sum_{j\neq i}
\mathbf 1(\widehat s_{i,j}=0,s_{i,j}=0).
\]
The true positive rate, false positive rate, true discovery rate, and false discovery rate are
then defined as
\[
\mathrm{TPR}
=
\frac{\mathrm{TP}}{\mathrm{TP}+\mathrm{FN}},
\qquad
\mathrm{FPR}
=
\frac{\mathrm{FP}}{\mathrm{FP}+\mathrm{TN}},
\]
and
\[
\mathrm{TDR}
=
\frac{\mathrm{TP}}{\mathrm{TP}+\mathrm{FP}},
\qquad
\mathrm{FDR}
=
\frac{\mathrm{FP}}{\mathrm{TP}+\mathrm{FP}}.
\]
Thus, TPR measures the fraction of true links that are recovered, FPR measures the fraction of non-links that are incorrectly selected, TDR measures the fraction of selected links that are true links, and FDR measures the fraction of selected links that are false links.

The proposed BOLMT procedure performs remarkably well across nearly all Monte Carlo designs. In particular, BOLMT achieves small misclassification rates, with median MAD values close to zero in both the full sample and the restricted sample with $T>25$. The procedure also delivers near-perfect true positive rates and true discovery rates, especially for moderate and large values of $T$. At the same time, both the false positive rate and the false discovery rate remain very small across almost all designs. These findings indicate that the proposed sequential selection mechanism is highly effective in
recovering the underlying sparse interaction structure while avoiding spurious link selection.

By contrast, both OCMT and MOCMT exhibit substantially weaker recovery performance. Although the two procedures often attain relatively high true positive rates, this comes at the cost of substantially larger false positive and false discovery rates. In particular, the FDR results indicate that a considerable proportion of the selected links are spurious across most designs. This problem persists even in larger samples and is especially pronounced in the more challenging network structures.

Figure~\ref{fig:MAD} reports MAD values across all 48 Monte Carlo designs, grouped by network structure and different $(N,T)$ combinations. As shown,  BOLMT performs consistently well across all three network structures. Recovery accuracy improves as $T$ increases, and by $T=50$ the
procedure delivers very low mis-classification rates in most cases. This pattern remains stable even under the random network with a dominant unit, where higher-order
interaction effects are expected to propagate more strongly through the network.

On the other hand, OCMT and MOCMT display noticeably higher MAD values, indicating less accurate
support recovery. OCMT performs worst in most designs, especially when $T=25$, while MOCMT usually improves on OCMT but remains well above BOLMT. The deterioration is particularly pronounced under the random network with a dominant unit, where the denser selection rules appear more vulnerable to spurious links generated by indirect dependence propagation.

The differences between the procedures appear to reflect their distinct selection mechanisms.
OCMT and MOCMT admit all candidates whose $t$-statistics exceed the multiple-testing
threshold at a given iteration. This aggressive rule allows pseudo links generated by indirect
dependence propagation to enter alongside true structural links and to retain spurious
explanatory power in later iterations. BOLMT, by contrast, selects at most one link at each
iteration and updates the conditioning set before reconsidering the remaining candidates.
This less greedy strategy reduces premature admission of pseudo links and improves
discrimination between direct and indirect dependence.

\begin{table}[p]
\centering
\caption{Summary performance measures across Monte Carlo designs}
\label{tab:summary_results}
\small
\renewcommand{\arraystretch}{1.15}
\setlength{\tabcolsep}{4pt}
\begin{tabular}{lccccc|ccccc}
\hline
& \multicolumn{5}{c|}{\textbf{Panel A: All $T$}}
& \multicolumn{5}{c}{\textbf{Panel B: $T>25$}} \\
\cline{2-11}

& \multicolumn{10}{c}{\textbf{MAD}} \\
\hline
 & Median & IQR & Min & Max & $\Pr(\text{MAD}>0.1)$
 & Median & IQR & Min & Max & $\Pr(\text{MAD}>0.1)$ \\
\hline
BOLMT & 0.00 & 0.01 & 0.00 & 0.04 & 0.0\%
       & 0.00 & 0.00 & 0.00 & 0.01 & 0.0\% \\
OCMT  & 0.04 & 0.06 & 0.04 & 0.23 & 20.8\%
       & 0.04 & 0.06 & 0.01 & 0.23 & 19.4\% \\
MOCMT & 0.04 & 0.05 & 0.04 & 0.27 & 14.6\%
       & 0.04 & 0.06 & 0.01 & 0.27 & 16.7\% \\

\hline
& \multicolumn{10}{c}{\textbf{TPR}} \\
\hline
 & Median & IQR & Min & Max & $\Pr(\text{TPR}>0.8)$
 & Median & IQR & Min & Max & $\Pr(\text{TPR}>0.8)$ \\
\hline
BOLMT & 1.00 & 0.02 & 0.55 & 1.00 & 89.6\%
       & 1.00 & 0.00 & 0.98 & 1.00 & 100.0\% \\
OCMT  & 0.71 & 0.36 & 0.71 & 0.99 & 33.3\%
       & 0.78 & 0.20 & 0.47 & 0.99 & 44.4\% \\
MOCMT & 1.00 & 0.05 & 1.00 & 1.00 & 83.3\%
       & 1.00 & 0.00 & 0.94 & 1.00 & 100.0\% \\

\hline
& \multicolumn{10}{c}{\textbf{FPR}} \\
\hline
 & Median & IQR & Min & Max & $\Pr(\text{FPR}>0.1)$
 & Median & IQR & Min & Max & $\Pr(\text{FPR}>0.1)$ \\
\hline
BOLMT & 0.00 & 0.00 & 0.00 & 0.03 & 0.0\%
       & 0.00 & 0.00 & 0.00 & 0.01 & 0.0\% \\
OCMT  & 0.03 & 0.04 & 0.03 & 0.25 & 10.4\%
       & 0.03 & 0.05 & 0.01 & 0.25 & 13.9\% \\
MOCMT & 0.03 & 0.05 & 0.03 & 0.32 & 12.5\%
       & 0.04 & 0.06 & 0.01 & 0.32 & 16.7\% \\

\hline
& \multicolumn{10}{c}{\textbf{TDR}} \\
\hline
 & Median & IQR & Min & Max & $\Pr(\text{TDR}>0.8)$
 & Median & IQR & Min & Max & $\Pr(\text{TDR}>0.8)$ \\
\hline
BOLMT & 0.98 & 0.05 & 0.63 & 0.99 & 93.8\%
       & 0.98 & 0.01 & 0.93 & 0.99 & 100.0\% \\
OCMT  & 0.59 & 0.14 & 0.59 & 0.80 & 0.0\%
       & 0.58 & 0.13 & 0.36 & 0.79 & 0.0\% \\
MOCMT & 0.62 & 0.17 & 0.62 & 0.83 & 4.2\%
       & 0.60 & 0.17 & 0.34 & 0.82 & 2.8\% \\

\hline
& \multicolumn{10}{c}{\textbf{FDR}} \\
\hline
 & Median & IQR & Min & Max & $\Pr(\text{FDR}>0.1)$
 & Median & IQR & Min & Max & $\Pr(\text{FDR}>0.1)$ \\
\hline
BOLMT & 0.02 & 0.05 & 0.01 & 0.37 & 18.8\%
       & 0.02 & 0.01 & 0.01 & 0.07 & 0.0\% \\
OCMT  & 0.41 & 0.14 & 0.41 & 0.64 & 100.0\%
       & 0.42 & 0.13 & 0.21 & 0.64 & 100.0\% \\
MOCMT & 0.38 & 0.17 & 0.38 & 0.66 & 100.0\%
       & 0.40 & 0.17 & 0.18 & 0.66 & 100.0\% \\

\hline
\end{tabular}
\vspace{0.2cm}
\begin{minipage}{0.98\textwidth}
\footnotesize
\textit{Notes:} MAD denotes the mean absolute deviation between the estimated and true
support matrices. TPR is the true positive rate, defined as the fraction of true links that
are selected. FPR is the false positive rate, defined as the fraction of true non-links that
are selected. TDR is the true discovery rate, defined as the fraction of selected links that
are true links. FDR is the false discovery rate, defined as the fraction of selected links
that are false links. Panel A summarises results across all $48$ designs, while Panel B
restricts attention to designs with $T>25$, leading to $36$ designs.
\end{minipage}
\end{table}

\clearpage
\begingroup
\pagestyle{empty}
\thispagestyle{empty}
\begin{figure}[p]
\centering
\begin{minipage}{0.8\linewidth}
\centering
\includegraphics[width=\linewidth]{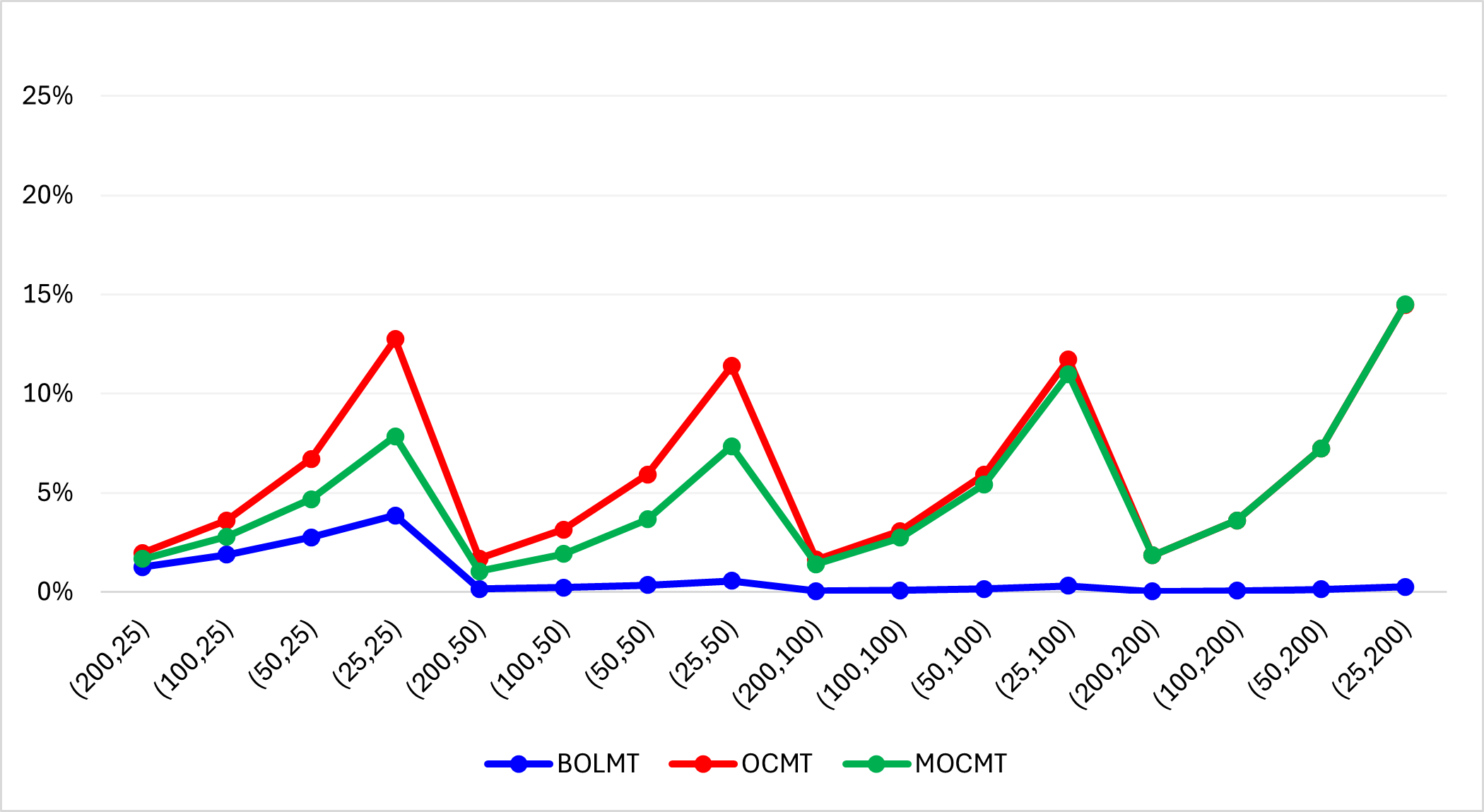}\\
\par\vspace{-0.5em}
\textbf{Circular $\boldsymbol{W}$}
\end{minipage}
\par\vspace{1.1em}
\begin{minipage}{0.8\linewidth}
\centering
\includegraphics[width=\linewidth]{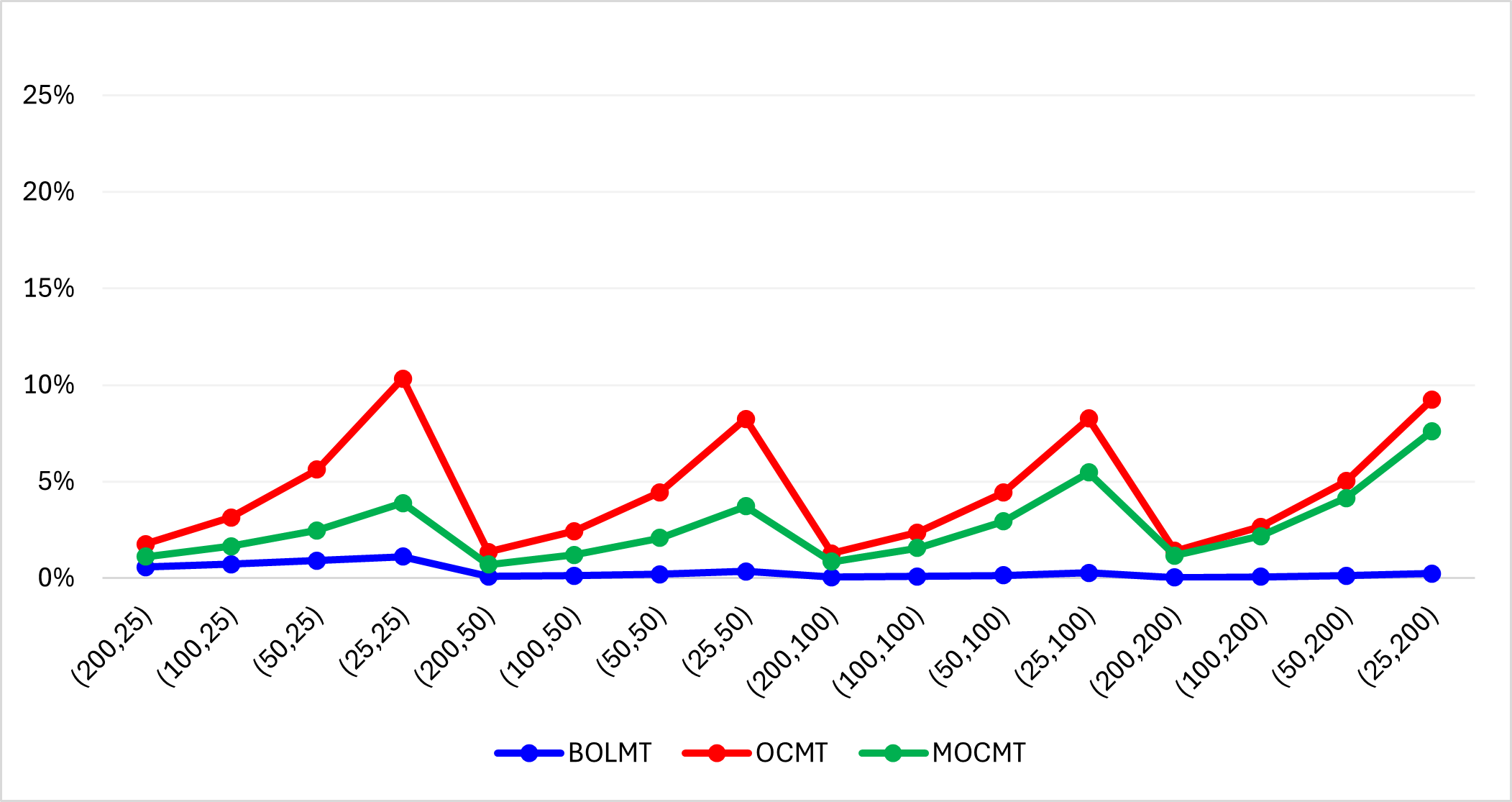}\\
\par\vspace{-0.5em}
\textbf{Block-Diagonal $\boldsymbol{W}$}
\end{minipage}
\par\vspace{1.1em}
\begin{minipage}{0.8\linewidth}
\centering
\includegraphics[width=\linewidth]{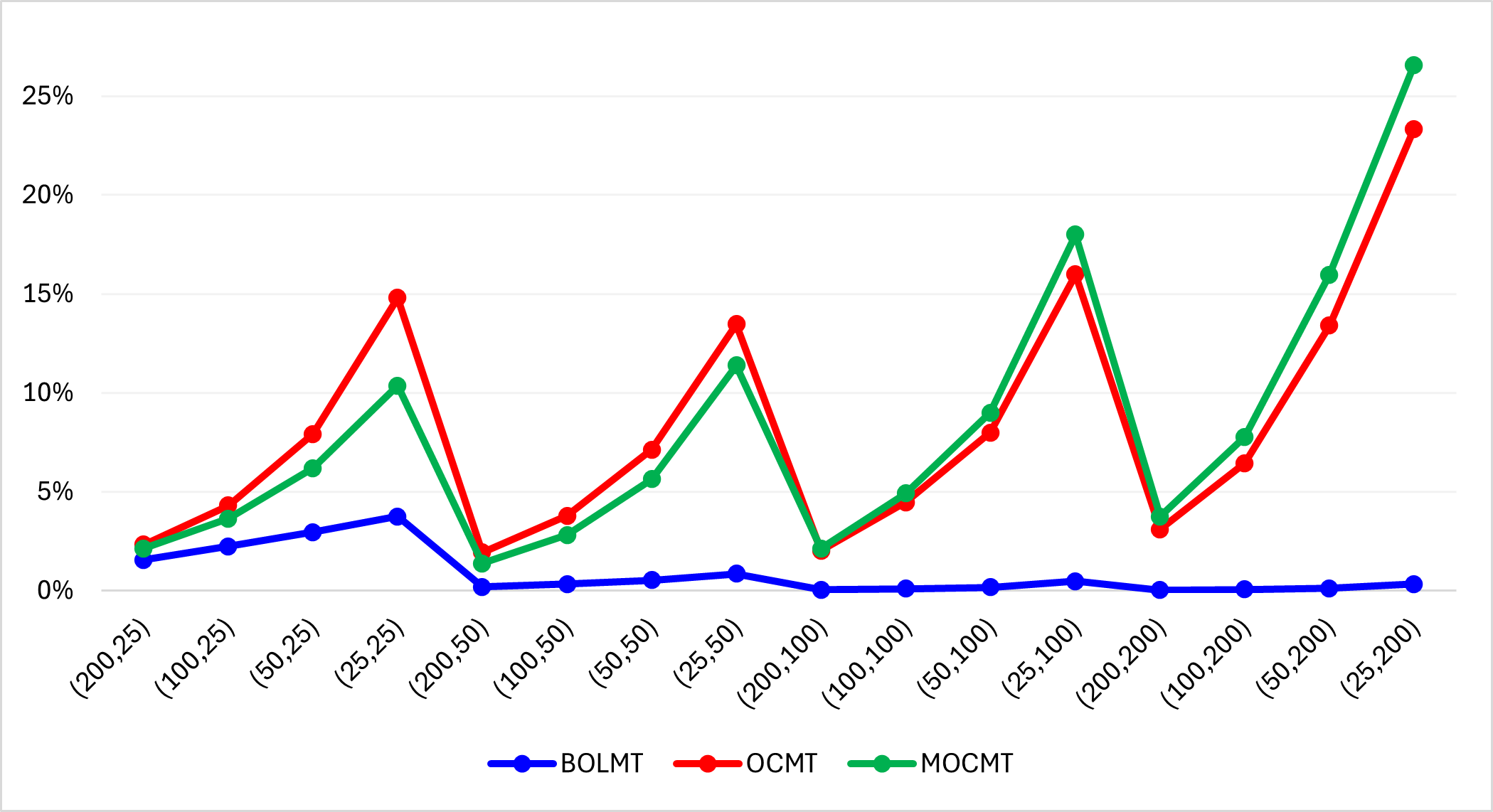}\\
\par\vspace{-0.5em}
\textbf{Random $\boldsymbol{W}$ with Dominant Unit}
\end{minipage}
\vspace{-0.4em}
\captionsetup{font=large}
\caption{MAD performance evaluation over different $(T,N)$ values}
\label{fig:MAD}
\end{figure}
\clearpage
\endgroup
\pagestyle{plain}

Furthermore, for BOLMT, we also examined the accuracy with which the weights associated with the asymmetric displacement network are recovered. In the data generating process, the asymmetric component is constructed using row weights $(0.75,0.25)$, which determine the relative contribution of the two directional components within the displacement network. The BOLMT estimates are tightly concentrated around the true values across all network designs. The largest deviations arise when $T=25$, particularly under the circular network structure, but these discrepancies decline rapidly as the time dimension increases. Overall, the results indicate that the proposed procedure not only recovers the sparse network structure accurately, but
also estimates the relative importance of the asymmetric displacement components with high precision.

\begin{table}[H]
\centering
\caption{Recovery of asymmetric component weights using BOLMT}
\label{tab:weight_summary}
\small
\renewcommand{\arraystretch}{1.1}

\begin{tabular}{lcc|cc}
\hline
& \multicolumn{2}{c|}{$\hat{\pi}_1$}
& \multicolumn{2}{c}{$\hat{\pi}_2$} \\
\cline{2-5}

Network
& Median
& Max. abs. error
& Median
& Max. abs. error \\
\hline

Circular
& 0.750
& 0.017
& 0.250
& 0.017 \\

Block
& 0.750
& 0.001
& 0.250
& 0.001 \\

Random
& 0.750
& 0.004
& 0.250
& 0.004 \\

\hline
\end{tabular}

\vspace{0.2cm}
\footnotesize
\textit{Notes:} The true asymmetric component weights are
$\pi_1=0.75$ and $\pi_2=0.25$. Max. abs. error denotes
$\max |\hat{\pi}_k-\pi_k|$ across all Monte Carlo designs.
\end{table}

Finally, Table \ref{tab:times} reports median computation times per Monte Carlo
replication for BOLMT and MOCMT across different network structures and
cross-sectional dimensions. As expected, computation times increase
substantially with $N$ for both procedures, reflecting the higher-dimensional
nature of the network recovery problem. Nevertheless, BOLMT is consistently
faster than MOCMT across all designs considered. This finding is noteworthy,
since MOCMT admits multiple links simultaneously at each stage, whereas
BOLMT proceeds sequentially by selecting only the single most significant
link at a time. The computational advantage of BOLMT is particularly
pronounced under the random network structure, where the BOLMT-to-MOCMT
time ratio falls below 0.5 for all values of $N$.

\begin{table}[H]
\centering
\caption{Median computation times per Monte Carlo replication}
\label{tab:times}
\small
\renewcommand{\arraystretch}{1.1}
\setlength{\tabcolsep}{4pt}

\begin{tabular}{c|ccc|ccc|ccc}
\hline
& \multicolumn{3}{c|}{Circular}
& \multicolumn{3}{c|}{Block}
& \multicolumn{3}{c}{Random} \\
\cline{2-10}

$N$
& BOLMT & MOCMT & Ratio
& BOLMT & MOCMT & Ratio
& BOLMT & MOCMT & Ratio \\
\hline

25
& 0.18 & 0.31 & 0.64
& 0.21 & 0.39 & 0.56
& 0.19 & 0.43 & 0.43 \\

50
& 0.84 & 1.26 & 0.65
& 0.81 & 1.40 & 0.58
& 0.86 & 1.97 & 0.42 \\

100
& 3.85 & 5.52 & 0.68
& 3.77 & 5.95 & 0.62
& 3.78 & 8.90 & 0.44 \\

200
& 16.70 & 23.46 & 0.74
& 16.80 & 25.47 & 0.66
& 16.39 & 35.45 & 0.48 \\

\hline
\end{tabular}

\vspace{0.2cm}
\footnotesize
\textit{Notes:} The table reports median computation times per Monte Carlo
replication (in minutes). The ratio column reports the ratio of BOLMT to
MOCMT computation times.
\end{table}

\subsubsection{Structural parameter estimation conditional on network recovery}

The results reported in Table \ref{tab:MC_summary} provide a concise overview of the finite sample performance of the competing estimators across the 48 Monte Carlo designs considered in the study. The table reports median absolute bias, median RMSE, and RMSE ratios relative to the Oracle IV estimator across all designs.

For each parameter, absolute bias is computed as the absolute difference between the Monte Carlo mean of the estimator and the true parameter value, multiplied by 100. In particular, for a vector of mean-group
estimates $\widehat{\boldsymbol\beta}_{MG}$ and true values $\boldsymbol\beta$, this is computed as
\[
\left|
\operatorname{mean}\left(\widehat{\boldsymbol\beta}_{MG}\right)
-
\boldsymbol\beta'
\right|
\times 100 .
\]

The RMSE ratio measures the efficiency loss relative to the infeasible benchmark that treats the true network structure as known. Detailed simulation results for all parameters and estimators are reported in Tables \ref{tab:beta1_BOLMT}--\ref{tab:rho2_MOCMT} in Appendix~B.

First, the proposed BOLMT estimator performs well across all parameters. Median biases remain small for both slope coefficients and spatial parameters, while the corresponding RMSE values are consistently close to those of the Oracle IV benchmark. In particular, the RMSE ratios are close to unity for all parameters, ranging from 1.074 to 1.784. This indicates
that the loss in overall estimation accuracy resulting from estimating the
underlying network structure is relatively modest. The results therefore suggest that BOLMT is able to recover the relevant network information sufficiently accurately to deliver estimation performance close to the infeasible Oracle estimator that treats the true network structure as known.

Second, the results reveal substantial differences between BOLMT and the alternative model selection procedures. OCMT exhibits considerably larger median biases and RMSE values across all parameters. The deterioration is particularly pronounced for the spatial autoregressive coefficients $\rho_1$ and $\rho_2$, where RMSE ratios exceed 80 and 100, respectively. These findings indicate that OCMT frequently fails to recover the relevant network interactions in finite samples, leading to substantial estimation inaccuracies relative to the Oracle benchmark. MOCMT performs considerably better than OCMT and remains competitive for the slope coefficients $\beta_1$ and $\beta_2$. In particular, the RMSE ratios for these parameters remain close to one, indicating relatively small losses in overall estimation accuracy. This improvement is not  surprising, since MOCMT applies the OCMT procedure iteratively, thereby allowing links omitted
in earlier stages to be subsequently admitted. However, this additional flexibility may also increase the likelihood of selecting spurious links.
While the iterative procedure improves performance relative to OCMT, the results indicate that MOCMT still encounters difficulties in accurately
recovering the spatial dependence structure. In particular, the RMSE ratios for $\rho_1$ and $\rho_2$ remain substantially larger than those obtained by
BOLMT, which is also reflected in the larger median biases for the spatial coefficients.


\begin{table}[H]
\centering
\caption{Summary of finite-sample estimation performance}
\label{tab:MC_summary}
\small
\renewcommand{\arraystretch}{1.15}
\setlength{\tabcolsep}{5pt}

\begin{tabular}{lcccc|cccc|cccc}
\hline
& \multicolumn{4}{c|}{\textbf{Median Bias}}
& \multicolumn{4}{c|}{\textbf{Median RMSE}}
& \multicolumn{4}{c}{\textbf{RMSE Ratio}} \\
\cline{2-13}

& $\beta_1$ & $\beta_2$ & $\rho_2$ & $\rho_1$
& $\beta_1$ & $\beta_2$ & $\rho_2$ & $\rho_1$
& $\beta_1$ & $\beta_2$ & $\rho_2$ & $\rho_1$ \\
\hline

BOLMT
& 0.012 & 0.008 & 0.067 & 0.064
& 0.003 & 0.002 & 0.002 & 0.002
& 1.141 & 1.074 & 1.784 & 1.721 \\

OCMT
& 2.013 & 1.009 & 5.577 & 9.409
& 0.033 & 0.021 & 0.057 & 0.095
& 18.426 & 11.781 & 87.395 & 103.058 \\

MOCMT
& 0.062 & 0.028 & 1.029 & 0.963
& 0.003 & 0.002 & 0.011 & 0.010
& 1.227 & 1.117 & 11.310 & 9.248 \\

\hline
\end{tabular}

\vspace{0.2cm}
\footnotesize
\textit{Notes:} The table reports median absolute bias and median RMSE across 48 Monte Carlo designs. Absolute bias is computed as the absolute difference between the Monte Carlo mean and the true parameter value, multiplied by 100. RMSE ratios are computed relative to the Oracle IV estimator, which uses the true network structure.
\end{table}

Figure \ref{fig:RMSE_ratio_rho2} reports the ratio of the RMSE of BOLMT relative to MOCMT for the estimation of the spatial parameter $\rho_2$ across all combinations of $(N,T)$ and the three network structures. Values below unity indicate that BOLMT outperforms MOCMT in terms of estimation accuracy. The solid, dotted, and dashed lines
correspond to $\mathbf W_1$, $\mathbf W_2$, and $\mathbf W_3$, respectively.

Several patterns emerge. First, BOLMT consistently dominates MOCMT across all designs, with RMSE
ratios well below one throughout. This finding complements the results reported in Table \ref{tab:MC_summary} and suggests that the step-wise, one-link-at-a-time selection rule associated with BOLMT, which limits the inclusion of
spurious interaction terms, contributes to the lower estimation error observed across the Monte Carlo designs.

Second, the gains from BOLMT are particularly pronounced in designs with relatively large cross-sectional dimensions. For a given value of $T$, the RMSE ratios generally decline as $N$ increases, indicating that BOLMT becomes especially advantageous when the candidate network is more high-dimensional. In several large-$N$ designs, the ratios fall below 0.1,
implying that the RMSE of BOLMT is less than one tenth of that obtained by MOCMT.

Third, the size of the gains depends on network topology. The largest improvements are
typically observed under the random network with a dominant unit, $\mathbf W_3$, where MOCMT appears most affected by highly asymmetric and heterogeneous interaction patterns.
BOLMT also delivers sizeable gains under the block-diagonal network, $\mathbf W_2$, while
the circular network, $\mathbf W_1$, produces smoother patterns and somewhat smaller
relative improvements. This is consistent with regular network structures being easier to
recover for both procedures.


\begin{figure}[H]
\centering
\includegraphics[width=0.9\textwidth]{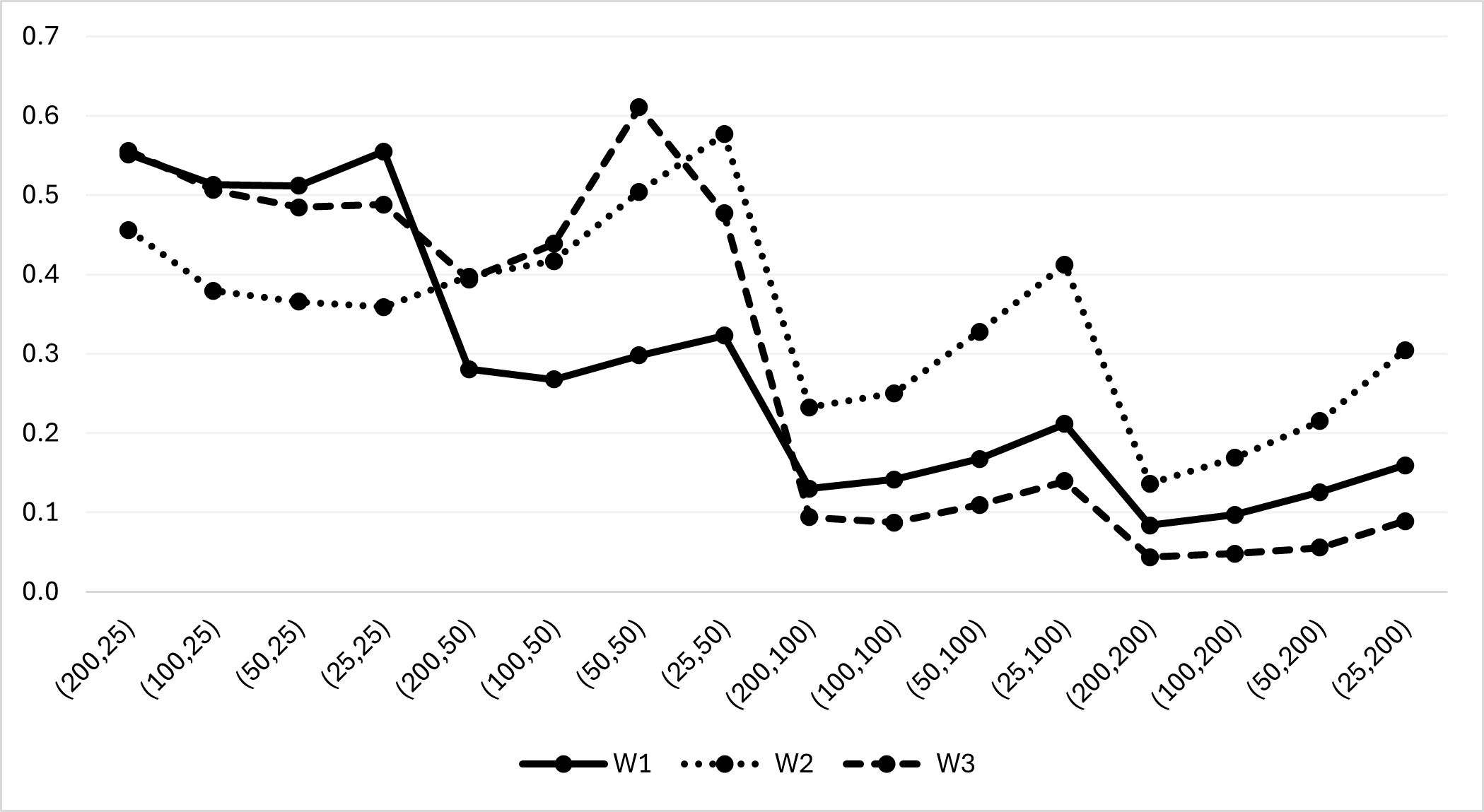}
\caption{RMSE ratios of BOLMT relative to MOCMT for the estimation of $\rho_2$ across different combinations of $(N,T)$ and network structures. W1 corresponds to the circular network, W2 to the block-diagonal network, and W3 to the random network with a dominant unit. Values below unity indicate superior performance of BOLMT relative to MOCMT.}
\label{fig:RMSE_ratio_rho2}
\end{figure}


\section{Empirical Application: Dual Network Effects in Corporate Financial Policies} \label{sec:empirical}

This section illustrates the practical usefulness of BOLMT by examining network effects in corporate financial policies. The analysis is motivated by the corporate finance literature on peer effects, which shows that firms do not make financial decisions
in isolation. Capital structure, liquidity management, investment, and other financial
policies may be influenced by economically related firms operating in similar product markets, giving rise to strategic interactions and cross-firm dependence.

A substantial body of empirical work has documented such interactions. Early contributions
include \citet{LearyRoberts2014}, who provide evidence of peer effects in capital-structure
decisions using instrumental-variable methods. Subsequent studies, including
\citet{KaustiaRantala2015}, \citet{AdhikariAgrawal2018}, \citet{Grennan2019}, and
\citet{GrieserEtAl2022}, employ richer peer-group or network structures to model strategic
interactions among firms. In particular, \citet{GrieserEtAl2022} develop a spatial
econometric framework based on product-market networks constructed from textual similarity
measures and find economically important interactions in leverage decisions.

Although these studies differ in how peer groups are defined, they share the common feature that the interaction network is externally specified. In practice, networks are typically constructed using industry classifications,  geographic proximity, social connections, or product-market similarity measures, and are then treated as fixed throughout the analysis. The empirical focus is therefore on estimating the strength and implications of interactions within a predetermined network, rather than recovering the network itself. Moreover, interactions are usually modelled through a single channel, so that all links are governed by the same qualitative mechanism. This restriction may be limiting:
the relevant network is rarely observed directly, and different links may reflect distinct economic forces. Some interactions may reinforce firms' decisions through imitation,
information transmission, learning, or common opportunities, whereas others may generate displacement effects through competition, market-share reallocation, or strategic
differentiation.

The methodology proposed in this paper addresses both limitations. BOLMT estimates the network directly from the data by identifying the economically relevant links associated with each firm, rather than imposing a pre-specified interaction structure. The recovered interaction matrix is then decomposed into reinforcing and displacement components,
allowing the structure of the network and the nature of the interactions to be determined empirically. In contrast to approaches where inference is conditional on a predetermined
network, the present framework treats network recovery as part of the econometric problem and allows positive and negative interaction effects to coexist across different links.

The dataset consists of a panel of 55 publicly listed U.S. firms observed between 1990 and 2023. Firm-level accounting and market variables are obtained from the CRSP/Compustat Merged database, following \citet{AsimakopoulosEtAl2026}. We consider the following dynamic network
specification:
\begin{equation}
y_{i,t}
=
\eta_i
+
\alpha_i y_{i,t-1}
+
\boldsymbol{\beta}_i'
\mathbf x_{i,t}
+
\rho_1
\sum_{j\neq i}
w_{1,i,j}y_{j,t}
+
\rho_2
\sum_{j\neq i}
w_{2,i,j}y_{j,t}
+
\varepsilon_{i,t}.
\label{eq:empirical_model}
\end{equation}
Equivalently, the procedure first recovers the composite interaction coefficients $\alpha_{i,j}=\rho_1w_{1,i,j}+\rho_2w_{2,i,j}$ and then assigns the estimated links to the reinforcing or displacement component according to the sign of $\widehat{\alpha}_{i,j}$, as described in Section~\ref{sec:psel_estimation}.
The dependent variable $y_{i,t}$ denotes the firm's leverage ratio, and the vector
\[
\mathbf x_{i,t}
=
\left(
\text{Cash}_{i,t},
\text{MB}_{i,t},
\text{Tang}_{i,t},
\text{Prof}_{i,t},
\text{Size}_{i,t}
\right)'
\]
contains the firm-specific characteristics commonly used in the capital-structure literature. Following \citet{GrieserEtAl2022}, these variables measure cash holdings relative to assets, the market-to-book ratio, asset tangibility, profitability, and firm size, respectively. The parameters $\eta_i$ capture firm-specific fixed effects, while the slope coefficients $\boldsymbol{\beta}_i$ are allowed to differ across firms. The matrices $\mathbf W_1$ and $\mathbf W_2$ represent reinforcing and displacement interaction networks, respectively.

Table~\ref{tab:network_results} reports Mean Group IV estimates obtained using four alternative network-selection procedures. Our discussion focuses primarily on BOLMT, while the remaining methods are considered for comparison purposes.

Under BOLMT, the estimated autoregressive coefficient is $0.257$ and highly significant, indicating persistence in firms' leverage decisions. The magnitude suggests gradual
adjustment over time, although the implied speed of adjustment remains relatively high.

The firm-specific controls generally have the expected signs. Cash holdings enter negatively and are statistically significant at the 10\% level, suggesting that firms with greater internal liquidity rely less on leverage. Market-to-book ratios enter positively and are also significant at the 10\% level, consistent with firms with stronger growth opportunities adopting more intensive financing policies. Tangibility has a positive and statistically significant coefficient, in line with the collateral role of tangible assets. Profitability enters negatively and is highly significant, consistent with the view that more profitable
firms rely less on external financing. Firm size has a positive and significant effect, suggesting that larger firms may benefit from better access to capital markets and greater
financial flexibility. These patterns are broadly stable across the alternative network-selection procedures.

The estimated network effects reveal the presence of two distinct interaction mechanisms. The reinforcing effect is positive and highly significant ($\widehat{\rho}_1=0.715$), while the displacement effect is negative and highly significant ($\widehat{\rho}_2=-0.606$). The comparable magnitudes of the two coefficients suggest that both mechanisms are quantitatively important in explaining cross-sectional dependence in firms' financial decisions.

The positive estimate of $\rho_1$ implies that firms connected through the reinforcing network exhibit positively associated leverage adjustments, after controlling for firm-specific characteristics and fixed effects. This is consistent with strategic complementarities arising from imitation, information transmission, learning, or exposure to common opportunities. By contrast, the negative estimate of $\rho_2$ points to a displacement mechanism, under which connected firms move in opposite directions. Such effects may reflect competitive pressures, market-share reallocation, or strategic
differentiation.  A conventional single-network specification would compress these opposing mechanisms into a single aggregate interaction effect, potentially masking important heterogeneity in the nature of cross-firm dependence. The dual-network framework therefore provides a richer description of strategic interactions in corporate financial policies.

The lower panel of Table~\ref{tab:network_results} reports network diagnostics. Under BOLMT, the estimated network contains 137 directed links, corresponding to a density of 4.6\%. The average degree is 2.5 and the median degree equals 2, indicating a sparse network in which most firms interact with only a small number of economically relevant counterparts. The relatively low density suggests that only a small fraction of all potential interactions are supported by the data.  The estimated spectral radius is 0.861, which is well below unity and therefore consistent with the stability condition imposed by the model.
The estimated network is also balanced in terms of positive and negative links. Approximately 55\% of connections belong to the reinforcing network and 45\% belong to the displacement network. This balance provides additional evidence that both interaction mechanisms are empirically relevant.

The remaining columns of Table~\ref{tab:network_results} provide useful comparisons. BOLMTP is a pruning refinement of BOLMT in which the selected support is re-assessed during
the boosting procedure. At each stage, previously selected links are re-tested conditional on the current selected set, and links that no longer remain significant are removed.  BOLMTP produces results that are very similar to BOLMT, implying that the stage-wise pruning step has little effect on the recovered support. OCMT identifies a considerably denser network, with 228 links and an average degree exceeding four. MOCMT produces the densest network by a substantial margin, selecting 451 links and yielding an average degree above eight.

The greater density of the OCMT and MOCMT networks is accompanied by larger estimated interaction effects. In particular, MOCMT yields estimates of $\rho_1$ and $\rho_2$ whose absolute magnitudes exceed unity. Moreover, the estimated interaction matrix associated with MOCMT has a spectral radius of 1.334, violating the stability condition imposed by the model. By contrast, the BOLMT estimates imply a substantially sparser interaction structure and a spectral radius below unity, consistent with the theoretical framework developed in this paper.

\begin{table}[!htbp]
\centering
\caption{Mean Group Estimates and Network Diagnostics}
\label{tab:network_results}
\small
\setlength{\tabcolsep}{5pt}
\renewcommand{\arraystretch}{1.0}
\begin{tabular}{lcccc}
\toprule
 & BOLMT & BOLMTP & OCMT & MOCMT \\
\midrule

$\alpha$ ($y_{i,t-1}$)
& 0.257$^{***}$ & 0.250$^{***}$ & 0.296$^{***}$ & 0.187$^{***}$ \\
& (0.039) & (0.038) & (0.036) & (0.040) \\

\text{Cash over deposits}
& -0.185$^{*}$ & -0.191$^{*}$ & -0.319$^{***}$ & -0.284$^{***}$ \\
& (0.107) & (0.101) & (0.075) & (0.099) \\

\text{Market-to-book}
& 0.003$^{*}$ & 0.004$^{***}$ & 0.002 & 0.003 \\
& (0.002) & (0.002) & (0.002) & (0.002) \\

\text{Tangibility}
& 0.187$^{*}$ & 0.204$^{*}$ & 0.134 & 0.225$^{*}$ \\
& (0.106) & (0.110) & (0.117) & (0.121) \\

\text{Profit}
& -0.339$^{***}$ & -0.295$^{***}$ & -0.306$^{***}$ & -0.361$^{***}$ \\
& (0.094) & (0.099) & (0.098) & (0.095) \\

\text{Log(assets)}
& 0.027$^{**}$ & 0.027$^{**}$ & 0.025$^{**}$ & 0.041$^{***}$ \\
& (0.011) & (0.011) & (0.010) & (0.011) \\

$\rho_1$ (reinforcing)
& 0.715$^{***}$ & 0.660$^{***}$ & 0.685$^{***}$ & 1.368$^{***}$ \\
& (0.095) & (0.094) & (0.134) & (0.204) \\

$\rho_2$ (displacement)
& -0.606$^{***}$ & -0.583$^{***}$ & -0.590$^{***}$ & -1.129$^{***}$ \\
& (0.117) & (0.110) & (0.150) & (0.178) \\

\hline
\multicolumn{5}{c}{\textbf{Network Diagnostics}}\\
\hline

Total Links
& 137 & 124 & 228 & 451 \\

Density (\%)
& 4.61 & 4.18 & 7.68 & 15.19 \\

Spectral Radius
& 0.861 & 0.794 & 0.754 & 1.334 \\

Average Degree
& 2.49 & 2.25 & 4.15 & 8.20 \\

Median Degree
& 2 & 2 & 3 & 6 \\

Maximum Degree
& 8 & 6 & 20 & 54 \\

95th Percentile Degree
& 5.0 & 4.0 & 12.8 & 19.5 \\

Positive Links (\%)
& 54.75 & 54.03 & 52.63 & 49.45 \\

Negative Links (\%)
& 45.25 & 45.97 & 47.37 & 50.55 \\

\bottomrule
\end{tabular}
\vspace{0.5em}

\begin{minipage}{\textwidth}
\footnotesize
\textit{\textbf{Notes}:}
Mean Group IV estimates are reported. Standard errors are shown in parentheses.
Significance levels:
$^{*}$ 10\%,
$^{**}$ 5\%,
$^{***}$ 1\%.
The bottom panel reports network diagnostics based on the estimated interaction matrices.
Density is computed as the percentage of nonzero links relative to the maximum number of directed links.
The spectral radius refers to the largest modulus eigenvalue of the estimated interaction matrix
$\widehat{\boldsymbol{A}}=\widehat{\rho}_1\widehat{\mathbf W}_1+\widehat{\rho}_2\widehat{\mathbf W}_2$.
\end{minipage}
\end{table}

Table~\ref{tab:effects_bmt} reports the decomposition of covariate effects into direct, indirect, and total components. Because the model contains contemporaneous interactions, a change in a firm's characteristics affects not only its own outcome but also the outcomes of other firms connected through the estimated networks. These feedback effects propagate through the interaction structure and therefore modify the overall impact of firm-level characteristics.
Since the reduced-form representation of the model can be written as $\mathbf y_t =
(\mathbf I_N-\boldsymbol{A})^{-1}
\left( \boldsymbol{\eta} +
\mathbf X_t\boldsymbol{\beta}
+ \mathbf u_t \right)$, with $\boldsymbol A = \rho_1\mathbf W_1+\rho_2\mathbf W_2$, the matrix of marginal effects associated with covariate $k$ is
\[
\mathbf S_k
=
(\mathbf I_N-\boldsymbol{A})^{-1}\beta_k.
\]
Following the spatial econometrics literature \citep{DebarsyLeSagePace2012,LeSagePace2009}, the average direct effect (DE) is computed as the average of the diagonal elements of $\mathbf S_k$, measuring the impact of a change in a firm's own characteristic on its own outcome. The average indirect effect (IE) is computed as the average row sum of the off-diagonal elements of $\mathbf S_k$, capturing the impact transmitted through the network to other firms. The total effect (TE) is the sum of the direct and indirect effects. Because the interaction matrix is decomposed into reinforcing and displacement components, the indirect effect can be further separated into positive and negative contributions associated with $\mathbf W_1$ and $\mathbf W_2$, respectively.

Several findings emerge from Table~\ref{tab:effects_bmt}. First, the direct effects closely resemble the estimated slope coefficients reported in Table~\ref{tab:network_results}, indicating that network feedback does not substantially alter the immediate impact of firm characteristics on financial policies. The indirect effects are nevertheless economically
meaningful, increasing the magnitude of the direct effects by approximately 14\% on average.

A notable feature of the results is that the indirect effects preserve the sign of the corresponding direct effects. Variables with positive direct effects, namely market-to-book
ratios, tangibility, and firm size, also generate positive indirect effects. Similarly, cash holdings and profitability have negative direct and indirect effects. Thus, network propagation reinforces rather than offsets the underlying influence of firm characteristics. Profitability exhibits the largest effect in absolute value. Its direct effect equals
$-0.348$, while the total effect reaches $-0.406$, indicating that network interactions strengthen the role of profitability in leverage decisions. More generally, the gap between direct and total effects is non-negligible across all variables, suggesting that part of the overall effect operates indirectly through the recovered network structure.

Overall, the results indicate that direct effects remain the dominant channel, but network propagation is quantitatively relevant. Approximately 14\% of the total effect of firm
characteristics operates indirectly through the estimated interaction networks.

\begin{table}[!htbp]
\centering
\caption{Direct and Indirect Effects}
\label{tab:effects_bmt}
\small
\setlength{\tabcolsep}{5pt}
\renewcommand{\arraystretch}{1.0}

\begin{tabular}{lccccc}
\toprule
& DE & IE & Positive IE & Negative IE & TE \\
\midrule

Cash over deposits
& -0.190 & -0.031 & 0.459 & -0.491 & -0.221 \\

Market-to-book
& 0.003 & 0.001 & 0.008 & -0.007 & 0.004 \\

Tangibility
& 0.192 & 0.032 & 0.496 & -0.464 & 0.223 \\

Profit
& -0.348 & -0.058 & 0.843 & -0.901 & -0.406 \\

Log(assets)
& 0.028 & 0.005 & 0.073 & -0.068 & 0.033 \\

\bottomrule
\end{tabular}

\vspace{0.5em}

\begin{minipage}{\textwidth}
\footnotesize
\textit{\textbf{Notes}:}
\textit{\textbf{Notes}:}
The table reports average direct effects (DE), indirect effects (IE), and total effects (TE) implied by the estimated dual-network model under BOLMT. Direct effects measure the average impact of a change in a firm's characteristic on its own outcome, while indirect effects capture the average impact transmitted through the estimated interaction networks. Positive and negative indirect effects report the contributions associated with reinforcing and displacement interactions, respectively. Total effects are computed as the sum of direct and indirect effects.
\end{minipage}
\end{table}

To investigate whether the estimated network merely reflects similarities in firm size, we compare the distribution of pairwise asset distances for linked and unlinked firm pairs. Let
\[
d_{i,j}
=
\left|
\log(\text{Assets}_i)
-
\log(\text{Assets}_j)
\right|
\]
denote the absolute difference in log assets between firms $i$ and $j$. If the estimated network simply connects firms of similar size, linked pairs should exhibit systematically smaller values of $d_{i,j}$ than unlinked pairs.

Following the nonparametric approach in \citet{AsimakopoulosEtAl2026} and \citet{KripfganzSarafidis2026}, we employ a Wilcoxon rank-sum test to compare the two distributions. We also report the Relative Homophily Index (RHI),
\[
\text{RHI}
=
\frac{
\operatorname{Median}(d_{i,j}\mid \widehat a_{i,j}\neq 0)
}{
\operatorname{Median}(d_{i,j}\mid \widehat a_{i,j}=0)
},
\]
which compares the median asset distance among linked firms with the corresponding median distance among unlinked firms. Values below unity indicate that linked firms tend to be more similar in size than unlinked firms, whereas values close to unity suggest little evidence of size-based homophily.

The results provide little evidence that the recovered network is driven by firm size. The median asset distance among linked pairs equals 9.055, compared with 9.061 among unlinked pairs, yielding an RHI of 1.001. Moreover, the Wilcoxon rank-sum test also fails to reject equality of the two distributions ($p=0.680$). These findings suggest that the links identified by BOLMT do not simply connect firms of similar size and that the recovered network captures information beyond firm-size similarity.

\section{Concluding Remarks}\label{sec:conclusions}

This paper proposes a step-wise IV screening procedure for panel network models with unknown interaction matrices. The procedure estimates local supports one equation at a time,
using scalar $t$-statistics or block $F$-statistics constructed from fitted regressors. This allows network links to be recovered directly from the structural equations, rather than being imposed a priori or inferred from reduced-form dependence.

A central feature of the framework is that it accommodates dual interaction structures. The recovered structural interaction matrix can be decomposed into reinforcing and
displacement components, allowing positive and negative interaction mechanisms to coexist within the same system. This distinction is important in economic applications where some
links may reflect imitation, learning, or common  opportunities, while others may arise from competition, market-share reallocation, or strategic differentiation.

The theory is centred on exact structural recovery, but it also clarifies the weaker support-containment property that obtains when the post-completion proxy-control condition is relaxed. This distinction is useful because non-link proxies may generate nonzero population screening criteria through indirect network propagation or common sources of dependence. Exact recovery requires both detectability of omitted true links and sufficient separation between true links and proxy candidates along the selection path. Once exact recovery holds, post-selection IV estimators are asymptotically equivalent to oracle IV estimators that know the true supports.

The probability theory is developed under weak temporal dependence and allows either thin-tailed or heavy-tailed score arrays. The distinction affects the admissible growth of
$N$ relative to $T$: thin tails support substantially higher-dimensional candidate sets, whereas heavy tails require polynomial growth restrictions. The framework also extends
naturally to block screening, covering dynamic network specifications and spatial Durbin-type models in which each candidate link contributes multiple regressors.

The Monte Carlo experiments show that the proposed BOLMT procedure recovers sparse network structures accurately and delivers post-selection estimates close to the oracle benchmark. Compared with more aggressive selection rules such as OCMT and MOCMT, the step-wise, one-link-at-a-time strategy substantially reduces the inclusion of spurious links and leads
to lower estimation error in the designs considered.

The empirical application to U.S. corporate leverage decisions illustrates the practical value of the approach. The results reveal the coexistence of reinforcing and displacement
effects in corporate financial policies, suggesting that firms' leverage decisions are shaped by both complementarities and competitive forces. These findings highlight the usefulness of estimating the interaction structure directly and allowing different types of network effects to operate simultaneously.

Overall, the proposed methodology provides a transparent and flexible approach for recovering latent network interactions in high-dimensional panel settings. It offers a way
to move beyond pre-specified interaction matrices and single-channel network models, while retaining support-recovery guarantees and oracle-equivalent post-selection inference.

\bibliographystyle{apalike}
\bibliography{networks}

@article{Kapetanios2026BMT,
  title={Model Selection in High-Dimensional Linear Regression using Boosting with Multiple Testing},
  author={Kapetanios, George and Sarafidis, Vasilis and Ventouri, Alexia},
  journal={arXiv preprint arXiv:2602.19705},
  year={2026}
}

@article{Yang2018,
  author    = {Zhenlin Yang},
  title     = {Unified M-estimation of Fixed-Effects Spatial Dynamic Models with Short Panels},
  journal   = {Journal of Econometrics},
  year      = {2018},
  volume    = {205},
  number    = {2},
  pages     = {423--447},
  publisher = {Elsevier},
  doi       = {10.1016/j.jeconom.2017.08.019}
}

@unpublished{Manresa2016,
  author = {Manresa, Elena},
  title  = {Estimating the Structure of Social Interactions Using Panel Data},
  year   = {2016},
  note   = {Unpublished manuscript, CEMFI, Madrid}
}

@article{Rose2017,
  author    = {Rose, Christiern D.},
  title     = {Identification of Peer Effects through Social Networks Using Variance Restrictions},
  journal   = {The Econometrics Journal},
  year      = {2017},
  volume    = {20},
  number    = {3},
  pages     = {S47--S60},
  publisher = {Oxford University Press}
}

@Article{Pesaran199579,
  author   = {M.Hashem Pesaran and Ron Smith},
  journal  = {Journal of Econometrics},
  title    = {Estimating Long-run Relationships From Dynamic Heterogeneous Panels},
  year     = {1995},
  issn     = {0304-4076},
  number   = {1},
  pages    = {79 - 113},
  volume   = {68},
  doi      = {http://dx.doi.org/10.1016/0304-4076(94)01644-F},
  url      = {http://www.sciencedirect.com/science/article/pii/030440769401644F},
}

@incollection{Baltagi2021Spatial,
  author    = {Baltagi, Badi H.},
  title     = {Spatial Panel Data Models},
  booktitle = {Econometric Analysis of Panel Data},
  edition    = {6},
  pages      = {391--424},
  publisher  = {Springer},
  address    = {Cham},
  year       = {2021},
  series     = {Springer Texts in Business and Economics}
}

@article{LamSouza2020JBES,
  author  = {Lam, Clifford and Souza, Pedro C. L.},
  title   = {Estimation and Selection of Spatial Weight Matrix in a Spatial Lag Model},
  journal = {Journal of Business \& Economic Statistics},
  year    = {2020},
  volume  = {38},
  number  = {3},
  pages   = {693--710},
  doi     = {10.1080/07350015.2019.1569526}
}

@book{LeSagePace2009,
  title     = {Introduction to Spatial Econometrics},
  author    = {LeSage, James P. and Pace, R. Kelley},
  publisher = {CRC Press},
  address   = {Boca Raton, FL},
  year      = {2009},
  series    = {Statistics: A Series of Textbooks and Monographs},
  isbn      = {9781420064254},
  doi       = {10.1201/9781420064254}
}

@article{DebarsyLeSagePace2012,
  title   = {The Interpretation and Presentation of Results in Spatial Econometrics},
  author  = {Debarsy, Nicolas and LeSage, James P. and Pace, R. Kelley},
  journal = {The Review of Regional Studies},
  year    = {2012},
  volume  = {42},
  number  = {2},
  pages   = {107--142},
  doi     = {10.52324/001c.8135}
}

@article{dePaula2020,
  author  = {de Paula, {\'A}ureo},
  title   = {Econometric Models of Network Formation},
  journal = {Annual Review of Economics},
  year    = {2020},
  volume  = {12},
  number  = {1},
  pages   = {775--799},
  doi     = {10.1146/annurev-economics-093019-113859}
}

@article{Elhorst2003,
  title={Specification and estimation of spatial panel data models},
  author={Elhorst, J. Paul},
  journal={International Regional Science Review},
  volume={26},
  number={3},
  pages={244--268},
  year={2003}
}

@article{KelejianPrucha2010,
  title={Specification and estimation of spatial autoregressive models with autoregressive and heteroskedastic disturbances},
  author={Kelejian, Harry H and Prucha, Ingmar R},
  journal={Journal of Econometrics},
  volume={157},
  number={1},
  pages={53--67},
  year={2010}
}

@book{Elhorst2014Book,
  title={Spatial Econometrics: From Cross-Sectional Data to Spatial Panels},
  author={Elhorst, J. Paul},
  year={2014},
  publisher={Springer}
}

@article{CuiEtal2023,
    author = {Cui, Guowei and Sarafidis, Vasilis and Yamagata, Takashi},
    title = "{IV Estimation of Spatial Dynamic Panels with Interactive Effects: Large Sample Theory and an Application on Bank Attitude Toward Risk}",
    journal = {The Econometrics Journal},
    volume = {26},
    number = {2},
    pages = {124-146},
    year = {2023},    
}

@article{Manski1993,
  author  = {Manski, Charles F.},
  title   = {Identification of Endogenous Social Effects: The Reflection Problem},
  journal = {Review of Economic Studies},
  year    = {1993},
  volume  = {60},
  number  = {3},
  pages   = {531--542}
}

@article{dePaulaRasulSouza2024,
  author  = {de Paula, {\'A}ureo and Rasul, Imran and Souza, Pedro C. L.},
  title   = {Identifying Network Ties from Panel Data: Theory and
an application to tax competition},
  journal = {Review of Economic Studies},
  year    = {2025},
  volume  = {92},
  number  = {4},
  pages   = {2691--2729}
}

@article{BlumeEtal2015,
   author = {Lawrence E. Blume and William A. Brock and Steven N. Durlauf and Rajshri Jayaraman},
   journal = {Journal of Political Economy},
 number = {2},
 pages = {444--496},
 publisher = {The University of Chicago Press},
 title = {Linear Social Interactions Models},
 volume = {123},
 year = {2015},
}

@article{LearyRoberts2014,
   author = {Mark T. Leary and Michael R. Roberts},
   journal = {Journal of Finance},
   number = {1},
   pages = {139--178},
   publisher = {Wiley},
   title = {Do Peer Firms Affect Corporate Financial Policy?},
   volume = {69},
   year = {2014},
}

@article{KaustiaRantala2015,
   author = {Markku Kaustia and Ville Rantala},
   journal = {Journal of Financial Economics},
   number = {3},
   pages = {653--669},
   publisher = {Elsevier},
   title = {Social Learning and Corporate Peer Effects},
   volume = {117},
   year = {2015},
}

@article{Grennan2019,
   author = {Jillian Grennan},
   journal = {Journal of Financial Economics},
   number = {3},
   pages = {549--570},
   publisher = {Elsevier},
   title = {Dividend Payments as a Response to Peer Influence},
   volume = {131},
   year = {2019},
}

@article{GrieserEtAl2022,
   author = {William Grieser and Charles J. Hadlock and James P. LeSage and Morad Zekhnini},
   journal = {Journal of Financial Economics},
   number = {1},
   pages = {247--272},
   publisher = {Elsevier},
   title = {Network Effects in Corporate Financial Policies},
   volume = {144},
   year = {2022},
}

@article{AdhikariAgrawal2018,
   author = {Binay Kumar Adhikari and Anup Agrawal},
   journal = {Journal of Corporate Finance},
   pages = {615--637},
   publisher = {Elsevier},
   title = {Peer Influence on Payout Policies},
   volume = {48},
   year = {2018},
}

@article{KrisztinPiribauer2023,
  author  = {Krisztin, Tam{\'a}s and Piribauer, Philipp},
  title   = {A Bayesian Approach for the Estimation of Weight Matrices in Spatial Autoregressive Models},
  journal = {Spatial Economic Analysis},
  year    = {2023},
  volume  = {18},
  number  = {3},
  pages   = {377--399},
  doi     = {10.1080/17421772.2022.2157456}
}

@article{AhrensBhattacharjee2015,
  author  = {Ahrens, Achim and Bhattacharjee, Arnab},
  title   = {Two-Step Lasso Estimation of the Spatial Weights Matrix},
  journal = {Econometrics},
  year    = {2015},
  volume  = {3},
  number  = {1},
  pages   = {128--155},
  doi     = {10.3390/econometrics3010128}
}

@article{Elhorst2024TwoLaws,
  author  = {Elhorst, J. Paul},
  title   = {Raising the Bar in Spatial Economic Analysis: Two Laws of Spatial Economic Modelling},
  journal = {Spatial Economic Analysis},
  year    = {2024},
  volume  = {19},
  number  = {2},
  pages   = {115--132},
  doi     = {10.1080/17421772.2024.2334845}
}

@article{LewbelQuTang2023,
   author = {A. Lewbel and X. Qu and X.Tang},
   journal = {Journal of Political Economy},
 number = {4},
 pages = {898--946},
 publisher = {The University of Chicago Press},
 title = {Social Networks with Unobserved Links},
 volume = {131},
 year = {2023},
}

@article{ShiLee2017,
  title   = {Spatial Dynamic Panel Data Models with Interactive Fixed Effects},
  author  = {Shi, Wei and Lee, Lung-fei},
  journal = {Journal of Econometrics},
  year    = {2017},
  volume  = {197},
  number  = {2},
  pages   = {323--347},
  doi     = {10.1016/j.jeconom.2016.12.001}
}

@article{BramoulleEtAl2009,
  author  = {Bramoull{\'e}, Yann and Djebbari, Habiba and Fortin, Bernard},
  title   = {Identification of Peer Effects through Social Networks},
  journal = {Journal of Econometrics},
  year    = {2009},
  volume  = {150},
  number  = {1},
  pages   = {41--55},
  doi     = {10.1016/j.jeconom.2008.12.021}
}

@article{KapoorKelejianPrucha2007,
  title   = {Panel Data Models with Spatially Correlated Error Components},
  author  = {Kapoor, Mudit and Kelejian, Harry H. and Prucha, Ingmar R.},
  journal = {Journal of Econometrics},
  year    = {2007},
  volume  = {140},
  number  = {1},
  pages   = {97--130},
  doi     = {10.1016/j.jeconom.2006.09.004}
}

@article{AsimakopoulosEtAl2026,
   author = {Stylianos Asimakopoulos and George Kapetanios and Vasilis Sarafidis and Alexia Ventouri},
   journal = {Journal of Money, Credit and Banking},
   pages = {1--49},
   title = {Network Effects in Corporate Emissions: Evidence from a Data-Dependent Spatial Panel Model},
   year = {2026},
   note = {Forthcoming},
}

@article{HuangEtAl2020,
  author  = {Wei Huang and Feng Jin and Yi Li and Xinyu Liu},
  title   = {The Two-Mode Network Autoregressive Model},
  journal = {Journal of Business \& Economic Statistics},
  volume  = {38},
  number  = {4},
  pages   = {1042--1056},
  year    = {2020}
}

@article{DendramisEtal2022, 
title={Estimation of Time-Varying Covariance Matrices for Large Datasets}, 
author={Dendramis, Yiannis and Giraitis, Liudas and Kapetanios, George},
volume={37}, 
number={6}, 
journal={Econometric Theory}, 
 year={2021}, 
pages={1100–1134}}

@article{KripfganzSarafidis2026,
   author = {Sebastian Kripfganz and Vasilis Sarafidis},
   journal = {Spatial Economic Analysis},
   number = {1},
   pages = {61--79},
   publisher = {Taylor \& Francis},
   title = {Chasing Opportunity: Spillovers and Drivers of US State Population Growth},
   volume = {21},
   year = {2026},
}

@book{ArbiaBaltagi2009,
  title     = {Spatial Econometrics: Methods and Applications},
  editor    = {Arbia, Giuseppe and Baltagi, Badi H.},
  year      = {2009},
  publisher = {Physica-Verlag},
  address   = {Heidelberg},
  series    = {Studies in Empirical Economics},
  isbn      = {9783790820690}
}

@article{BaltagiEtAl2013,
  title   = {A Generalized Spatial Panel Data Model with Random Effects},
  author  = {Baltagi, Badi H. and Egger, Peter and Pfaffermayr, Michael},
  journal = {Econometric Reviews},
  year    = {2013},
  volume  = {32},
  number  = {5--6},
  pages   = {650--685},
  doi     = {10.1080/07474938.2012.742342}
}

@book{Anselin1988,
  author    = {Anselin, Luc},
  title     = {Spatial Econometrics: Methods and Models},
  publisher = {Kluwer Academic Publishers},
  address   = {Dordrecht},
  year      = {1988}
}

@Article{chu2018,
  Title                    = {A One Covariate at a Time, Multiple Testing Approach to Variable
Selection in High-Dimensional Linear Regression Models},
  Author                   = {A. Chudik and G. Kapetanios and M. H. Pesaran},
  Journal                  = {Econometrica},
  Year                     = {2018},
  Pages                    = {1479-1512},
  Volume                   = {86}
}

@unpublished{ChenEtAl2025,
  author = {Chen, Jing and Cui, Guodong and Sarafidis, Vasilis and Yamagata, Takashi},
  title  = {I{V} Estimation of Heterogeneous Spatial Dynamic Panel Models with Interactive Effects},
  year   = {2025},
  note   = {Working paper, available at arXiv:2501.18467}
}

@unpublished{LiBhattacharjee2024,
  author       = {Li, Haoyang and Bhattacharjee, Arnab},
  title        = {Estimating Spatial Weights in High Dimensions: A CCE-OCMT Variable Selection Procedure},
  note         = {Work in Progress},
  year         = {2024},
  institution  = {Heriot-Watt University}
}

\clearpage

\appendix

\renewcommand{\thesection}{A}
\renewcommand{\thesubsection}{A.\arabic{subsection}}

\section*{Appendix A: Proofs}
\addcontentsline{toc}{section}{Appendix A: Proofs}

\setcounter{subsection}{0}
\setcounter{equation}{0}

\renewcommand{\theequation}{A.\arabic{equation}}

\setcounter{lemma}{0}
\renewcommand{\thelemma}{A.\arabic{lemma}}

\setcounter{proposition}{0}
\renewcommand{\theproposition}{A.\arabic{proposition}}

\setcounter{theorem}{0}
\renewcommand{\thetheorem}{A.\arabic{theorem}}

Throughout the appendix, \(C\) denotes a finite positive constant that may vary from line to line.  All limits are as \(N,T\to\infty\).  The proofs use the notation of Sections~\ref{sec:assumptions}--\ref{sec:block_screening}.  In particular,
\[
\mathcal K_i=\mathcal S_i^n\cup\mathcal S_i^p,
\qquad
\mathcal H_i=\{H:H\subseteq\mathcal K_i\}.
\]
Since \(\bar K=\sup_i|\mathcal K_i|<\infty\),
\begin{equation}
|\mathcal H_i|\le 2^{\bar K}.
\label{eq:app_H_count}
\end{equation}
The appendix is written for the scalar and block procedures jointly where possible.  The only difference between the two cases is the order of the screening statistic: the scalar score is linear and has stochastic order \(\sqrt{\log N}\), while the block statistic is quadratic in a fixed-dimensional score vector and has stochastic order \(\log N\).

\subsection{\textbf{Admissible conditioning sets and realised selected sets}}

The realised selected set is random.  The proof avoids taking a union bound over all subsets of the \(N-1\) candidates.  It first establishes uniform bounds over admissible conditioning sets \(H\in\mathcal H_i\), and then shows by induction that the realised path remains in this low-dimensional conditioning-set class with probability tending to one.

\begin{lemma}\label{lem:app_path_induction}
Suppose that, for a fixed equation \(i\), whenever the current selected set \(H\) belongs to \(\mathcal H_i\), the next selected candidate, if any, belongs to \(\mathcal K_i\).  Then every selected set reached by the algorithm for equation \(i\) belongs to \(\mathcal H_i\).
\end{lemma}

\begin{proof}
The initial selected set is \(H_{i,0}=\varnothing\), and \(\varnothing\subseteq\mathcal K_i\).  Suppose by induction that \(H_{i,m}\subseteq\mathcal K_i\).  If the algorithm stops, there is nothing to prove.  If it selects another candidate, the premise implies that the selected candidate \(j_{m+1}\) belongs to \(\mathcal K_i\).  Hence
\[
H_{i,m+1}=H_{i,m}\cup\{j_{m+1}\}\subseteq\mathcal K_i.
\]
The result follows by induction over stages.
\end{proof}

\begin{lemma}\label{lem:app_count}
The number of null scalar score arrays that must be controlled uniformly over equations, irrelevant candidates, and admissible conditioning sets is \(O(N^2)\).  The same order applies to each fixed component of a block score.
\end{lemma}

\begin{proof}
For each equation \(i\), there are at most \(N-1\) candidates and at most \(2^{\bar K}\) admissible conditioning sets.  Therefore
\[
\sum_{i=1}^N |\mathcal S_i^d|\,|\mathcal H_i|
\le
N(N-1)2^{\bar K}=O(N^2).
\]
The block dimension is fixed, so multiplication by the number of block components affects only the constant.
\end{proof}

\subsection{\textbf{DGK inequalities and tail closure}}

For a scalar array \(\xi_t\), recall the classes \(E(s)\) and \(H(\theta)\) from Section~\ref{sec:assumptions}.  The appendix uses the DGK inequalities for dependent unbounded variables.  These inequalities are applied to population score arrays such as \(v^*_{i,j,H,t}\varepsilon_{i,t}\), not directly to sample fitted regressors with full-sample random weights.

\begin{lemma}\label{lem:app_dgk}
Let
\[
S_T=T^{-1/2}\sum_{t=1}^T(\xi_t-\E\xi_t),
\]
and suppose that \(\xi_t-\E\xi_t\) is exponentially \(\alpha\)-mixing.  If \(\xi_t\in E(s)\), then, with \(\gamma=s/(s+1)\),
\[
\Pr(|S_T|>\zeta)\le f_T(2,\gamma,c,\zeta).
\]
If \(\xi_t\in H(\vartheta)\), then for every \(2<\vartheta_0<\vartheta\),
\[
\Pr(|S_T|>\zeta)\le g_T(2,\vartheta_0,c,\zeta).
\]
\end{lemma}

\begin{proof}
This is the Bernstein-type inequality of Dendramis, Giraitis and Kapetanios for dependent unbounded variables.  In their argument, bounded-variable inequalities are used only after truncating the variables; the final bound applies to unbounded variables satisfying the stated \(E(s)\) or \(H(\theta)\) condition.
\end{proof}

\begin{lemma}\label{lem:app_uniform_dgk}
Let \(\{\xi_{\ell,t}\}_{t=1}^T\), \(\ell=1,\ldots,L_N\), satisfy Lemma~\ref{lem:app_dgk} uniformly, with \(L_N=O(N^2)\).  Under Regime E of Assumption~\ref{ass:tails},
\[
\max_{\ell\le L_N}
\left|T^{-1/2}\sum_{t=1}^T(\xi_{\ell,t}-\E\xi_{\ell,t})\right|
=O_p(\sqrt{\log N}).
\]
Under Regime H of Assumption~\ref{ass:tails}, the same conclusion holds.
Moreover, under either regime there is a finite $C_0$ such that, for every fixed
$C>C_0$,
\[
\Pr\!\left(
\max_{\ell\le L_N}
\left|T^{-1/2}\sum_{t=1}^T(\xi_{\ell,t}-\E\xi_{\ell,t})\right|
>C\sqrt{\log N}
\right)\to0.
\]
\end{lemma}

\begin{proof}
Set \(\zeta=C\sqrt{\log N}\).  Under Regime E, Lemma~\ref{lem:app_dgk} and the union bound give
\[
\Pr\left(
\max_{\ell\le L_N}\left|T^{-1/2}\sum_t(\xi_{\ell,t}-\E\xi_{\ell,t})\right|>\zeta
\right)
\le
C N^2 f_T(2,\gamma_E,c,C\sqrt{\log N}),
\]
which tends to zero by Assumption~\ref{ass:tails} for sufficiently large \(C\).

Under Regime H,
\[
\Pr\left(
\max_{\ell\le L_N}\left|T^{-1/2}\sum_t(\xi_{\ell,t}-\E\xi_{\ell,t})\right|>\zeta
\right)
\]
\[
\le
C N^2\left\{\exp(-cC^2\log N)
+
(C\sqrt{\log N})^{-\vartheta_0}T^{-(\vartheta_0/2-1)}\right\}.
\]
The exponential component tends to zero for \(C\) sufficiently large.  The polynomial component tends to zero by Assumption~\ref{ass:tails}.
\end{proof}

\begin{lemma}\label{lem:app_tail_products}
If \(a_t\in E(s_a)\) and \(b_t\in E(s_b)\), then \(a_tb_t\in E(s_{ab})\), where
\[
s_{ab}=\frac{s_as_b}{s_a+s_b}.
\]
If \(a_t\in H(\vartheta_a)\) and \(b_t\in H(\vartheta_b)\), then \(a_tb_t\in H(\vartheta_{ab})\) for every
\[
0<\vartheta_{ab}<\frac{\vartheta_a\vartheta_b}{\vartheta_a+\vartheta_b}.
\]
If \(a_t\in E(s_a)\) and \(b_t\in H(\vartheta_b)\), then \(a_tb_t\in H(\vartheta_0)\) for every \(0<\vartheta_0<\vartheta_b\).
\end{lemma}

\begin{proof}
For the \(E(s)\) statement, let \(p=(s_a+s_b)/s_b\) and \(q=(s_a+s_b)/s_a\).  Then \(p^{-1}+q^{-1}=1\), \(s_{ab}p=s_a\), and \(s_{ab}q=s_b\).  H\"older's inequality gives
\[
\E |a_tb_t|^{ks_{ab}}
\le
\{\E |a_t|^{ks_a}\}^{1/p}
\{\E |b_t|^{ks_b}\}^{1/q}.
\]
Expanding the exponential moment as a power series yields \(\E\exp(c|a_tb_t|^{s_{ab}})<\infty\) for sufficiently small \(c>0\).

For the \(H(\theta)\) statement, choose \(r_a,r_b>1\) such that \(r_a^{-1}+r_b^{-1}=1\), \(r_a\vartheta_{ab}<\vartheta_a\), and \(r_b\vartheta_{ab}<\vartheta_b\).  Then
\[
\E |a_tb_t|^{\vartheta_{ab}}
\le
\{\E |a_t|^{r_a\vartheta_{ab}}\}^{1/r_a}
\{\E |b_t|^{r_b\vartheta_{ab}}\}^{1/r_b}<\infty.
\]
The mixed case follows because \(E(s_a)\) implies finite moments of all orders.
\end{proof}

\subsection{\textbf{Sample first stages}}

\begin{lemma}\label{lem:app_first_stage}
Under Assumption~\ref{ass:first_stage},
\[
\sup_{i,j,H\in\mathcal H_i}
T^{-1}\sum_{t=1}^T(\widehat v_{i,j,H,t}-v^*_{i,j,H,t})^2=\op(1),
\]
and
\[
\sup_{i,j,H\in\mathcal H_i}
\left|T^{-1/2}\sum_{t=1}^T(\widehat v_{i,j,H,t}-v^*_{i,j,H,t})\varepsilon_{i,t}\right|=\op(\sqrt{\log N}).
\]
The same bound holds with $\varepsilon_{i,t}$ replaced by $\mu_{i,j,H,t}$, and
\[
\sup_{i,j,H\in\mathcal H_i}
|\widehat\sigma_{i,j,H}^2-\sigma_{i,j,H}^2|=\op(1).
\]
Under Assumption~\ref{ass:block_first_stage}, the analogous block statements
hold componentwise for
\(\widehat{\mathbf V}_{i,j,H,t}-\mathbf V^*_{i,j,H,t}\).
\end{lemma}

\begin{proof}
The scalar statements are exactly the sample first-stage approximation conditions imposed in
Assumption~\ref{ass:first_stage}, and the block statement is imposed in
Assumption~\ref{ass:block_first_stage}. The result is recorded here because the subsequent
probability bounds are proved for the population partial fitted objects and then transferred
to the feasible fitted objects using these approximations.
\end{proof}

\begin{lemma}\label{lem:primitive_stagewise_exogeneity}
Under Assumption~\ref{ass:iv}, for every admissible \((i,j,H)\),
\[
\E(v^*_{i,j,H,t}\varepsilon_{i,t})=0.
\]
For block screening,
\[
\E(\mathbf V^*_{i,j,H,t}\varepsilon_{i,t})=\mathbf 0.
\]
\end{lemma}

\begin{proof}
By construction, \(\bar v^*_{i,j,H,t}\) and \(\mathbf C^*_{i,j,H,t}\) are linear combinations of
\(\boldsymbol{\mathcal Z}_{i,j,H,t}\). Therefore the partial fitted candidate
\[
v^*_{i,j,H,t}=\bar v^*_{i,j,H,t}-\boldsymbol\Gamma_{vC,i,j,H}'\mathbf C^*_{i,j,H,t}
\]
is also a linear combination of \(\boldsymbol{\mathcal Z}_{i,j,H,t}\). Hence there exists a vector
\(a_{i,j,H}\) such that \(v^*_{i,j,H,t}=a_{i,j,H}'\boldsymbol{\mathcal Z}_{i,j,H,t}\). Assumption~\ref{ass:iv}
gives \(\E(\boldsymbol{\mathcal Z}_{i,j,H,t}\varepsilon_{i,t})=0\), and therefore
\[
\E(v^*_{i,j,H,t}\varepsilon_{i,t})
= a_{i,j,H}'\E(\boldsymbol{\mathcal Z}_{i,j,H,t}\varepsilon_{i,t})=0.
\]
The block statement follows column by column because every component of \(\mathbf V^*_{i,j,H,t}\) is also a linear combination of \(\boldsymbol{\mathcal Z}_{i,j,H,t}\).
\end{proof}

\subsection{\textbf{Scalar screening lemmas}}\label{sec:scalar_screening_lemmas}

Recall that $\mu_{i,j,H,t}$ is the omitted structural component defined in
Section~\ref{sec:assumptions}. For the asymptotic analysis it is convenient to write
\begin{equation}
\Delta^t_{i,j,H}
=
\frac{
\left|
\E\!\left(
 v^*_{i,j,H,t}\mu_{i,j,H,t}
\right)
\right|
}
{
\sigma_{i,j,H}
\left\{
\E\!\left[(v^*_{i,j,H,t})^2\right]
\right\}^{1/2}
}.
\label{eq:delta_t}
\end{equation}
By the definitions in Section~\ref{sec:assumptions},
$\Delta^t_{i,j,H}=|m_{i,j}(H)|$.

\begin{lemma}\label{lem:app_t_rep}
The scalar stagewise IV/FWL statistic satisfies
\[
t_{i,j,H}
=
\frac{T^{-1/2}\widehat{\mathbf v}_{i,j,H}'\widetilde{\mathbf y}_{i,j,H}}
{\widehat\sigma_{i,j,H}(T^{-1}\widehat{\mathbf v}_{i,j,H}'\widehat{\mathbf v}_{i,j,H})^{1/2}}.
\]
\end{lemma}

\begin{proof}
At intermediate model \(H\), the full stagewise IV regression treats the selected controls
\(\mathbf C_{i,H}\) and the current candidate \(\mathbf v_j\) jointly. By the IV/FWL theorem,
the coefficient on the current candidate equals the coefficient from regressing the residualised
outcome \(\widetilde{\mathbf y}_{i,j,H}\) on the partial fitted candidate
\(\widehat{\mathbf v}_{i,j,H}\), where both residualisations are with respect to the fitted selected
controls. Therefore
\[
\widehat\theta_{i,j,H}
=
\frac{\widehat{\mathbf v}_{i,j,H}'\widetilde{\mathbf y}_{i,j,H}}
{\widehat{\mathbf v}_{i,j,H}'\widehat{\mathbf v}_{i,j,H}}.
\]
Dividing the normalised numerator by the corresponding residual standard deviation gives the displayed statistic.
Moreover, $\widehat{\mathbf v}_{i,j,H}$ lies in the column space of the
stagewise instrument matrix and is orthogonal to $\widehat{\mathbf C}_{i,j,H}$.
Since $\mathbf C_{i,H}-\widehat{\mathbf C}_{i,j,H}$ is orthogonal to the same
instrument space, $\widehat{\mathbf v}_{i,j,H}'\mathbf C_{i,H}=0$. The
structural equation therefore gives the exact score identity
\[
\widehat{\mathbf v}_{i,j,H}'\widetilde{\mathbf y}_{i,j,H}
=
\sum_{t=1}^T\widehat v_{i,j,H,t}
\{\mu_{i,j,H,t}+\varepsilon_{i,t}\}.
\]
\end{proof}

\begin{lemma}\label{lem:app_t_denom}
Under Assumptions~\ref{ass:iv}--\ref{ass:tails},
\[
\inf_{i,j,H\in\mathcal H_i}T^{-1}\widehat{\mathbf v}_{i,j,H}'\widehat{\mathbf v}_{i,j,H}>c
\]
and
\[
\inf_{i,j,H\in\mathcal H_i}\widehat\sigma_{i,j,H}^2>c
\]
with probability tending to one, for some \(c>0\).
\end{lemma}

\begin{proof}
Decompose
\[
T^{-1}\widehat{\mathbf v}'\widehat{\mathbf v}
=
T^{-1}\sum_t(v_t^*)^2
+2T^{-1}\sum_t v_t^*(\widehat v_t-v_t^*)
+T^{-1}\sum_t(\widehat v_t-v_t^*)^2.
\]
The first term converges uniformly to \(\E[(v^*)^2]\), which is bounded away from zero by Assumption~\ref{ass:iv}.  Uniform convergence follows from Lemma~\ref{lem:app_uniform_dgk} applied to \((v_t^*)^2-\E[(v_t^*)^2]\).  Tail admissibility of the square is imposed by Assumption~\ref{ass:tails} for the scalar processes listed in Section~\ref{sec:assumptions}.  The second and third terms are \(\op(1)\) uniformly by Lemma~\ref{lem:app_first_stage}.  The residual variance bound follows from the uniform convergence of the residual variance estimator to its population limit \(\sigma_{i,j,H}^2\), together with the lower bound in Assumption~\ref{ass:iv}.
\end{proof}

\begin{lemma}\label{lem:app_t_null}
Under Assumptions~\ref{ass:sampling}--\ref{ass:tails} and
Assumption~\ref{ass:t_signal}, there is a finite $C_0$ such that, for every fixed
$C>C_0$,
\[
\Pr\!\left(
\max_i\max_{H\in\mathcal H_i}\max_{j\in\mathcal S_i^d}
|t_{i,j,H}|>C\sqrt{\log N}
\right)\to0.
\]
\end{lemma}

\begin{proof}
For \(j\in\mathcal S_i^d\), Assumption~\ref{ass:t_signal}(iii) and Eq.
\eqref{eq:scalar_signal_decomp} give
$\sqrt T\,|m_{i,j}(H)|=o(\sqrt{\log N})$ uniformly. By Lemma~\ref{lem:app_t_rep},
\[
T^{-1/2}\widehat{\mathbf v}'\widetilde{\mathbf y}_{i,j,H}
=
T^{-1/2}\sum_t\widehat v_t\mu_t
+
T^{-1/2}\sum_t\widehat v_t\varepsilon_t.
\]
Replacing \(\widehat v_t\) by \(v_t^*\) is \(\op(\sqrt{\log N})\) uniformly by Lemma~\ref{lem:app_first_stage}. The deterministic component of the first sum is $o(\sqrt{\log N})$, while the centred sums involving \(v_t^*\mu_t\) and \(v_t^*\varepsilon_t\) satisfy the fixed-envelope bound in Lemma~\ref{lem:app_uniform_dgk}. The number of arrays is \(O(N^2)\) by Lemma~\ref{lem:app_count}. Division by the uniformly convergent denominator and choice of a sufficiently large $C$ give the result.
\end{proof}

\begin{lemma}\label{lem:app_t_signal}
Uniformly over \(i\), \(H\in\mathcal H_i\), and \(j\in\mathcal K_i\setminus H\),
\[
|t_{i,j,H}|=\sqrt T\,\Delta^t_{i,j,H}\{1+\op(1)\}
+\Op(\sqrt{\log N}).
\]
\end{lemma}

\begin{proof}
The numerator expansion is
\[
T^{-1/2}\widehat{\mathbf v}'\widetilde{\mathbf y}_{i,j,H}
=
T^{-1/2}\sum_t\widehat v_t\mu_t
+
T^{-1/2}\sum_t\widehat v_t\varepsilon_t.
\]
Replacing \(\widehat v_t\) by \(v_t^*\) contributes \(\op(\sqrt{\log N})\).  The first term equals
\[
\sqrt T\,\E(v_t^*\mu_t)+\Op(\sqrt{\log N})
\]
uniformly by Lemma~\ref{lem:app_uniform_dgk}. The second term is \(\Op(\sqrt{\log N})\) uniformly. Uniform convergence of the Gram and residual-scale factors in Assumption~\ref{ass:first_stage}, together with Lemma~\ref{lem:app_t_denom}, gives the multiplicative $\{1+\op(1)\}$ term in the leading signal and yields the stated expansion.
\end{proof}

\begin{lemma}
\label{lem:app_primitive_t_signal}
Let \(m_{i,j}(H)\) denote the normalised
population IV screening signal defined in Eq.
\eqref{eq:m_def}. Under
Assumptions~\ref{ass:t_signal} and~\ref{ass:t_dominance}, the following hold
uniformly over the stated sets:

\begin{enumerate}[label=(\roman*)]
\item For \(H\in\mathcal H_i^0\),
\[
\inf_{i,H\in\mathcal H_i^0}
\max_{j\in U_i(H)}
\frac{\sqrt T\,|m_{i,j}(H)|}{\sqrt{\log N}}
\to\infty .
\]

\item For \(H\in\mathcal H_i^0\),
\[
\inf_{i,H\in\mathcal H_i^0}
\frac{\sqrt T}{\sqrt{\log N}}
\left[
\max_{j\in U_i(H)}|m_{i,j}(H)|
-
\max_{p\in P_i(H)}|m_{i,p}(H)|
\right]
\to\infty .
\]

\item For \(H\in\mathcal H_i^1\),
\[
\sup_{i,H\in\mathcal H_i^1}
\max_{p\in P_i(H)}
\frac{\sqrt T\,|m_{i,p}(H)|}{\sqrt{\log N}}
\to0 .
\]
\end{enumerate}
\end{lemma}

\begin{proof}
By the definition of \(m_{i,j}(H)\) and the population IV screening
decomposition,
\[
m_{i,j}(H)
=
\frac{E\!\left(\widetilde y_{j,H,t}y_{i,t}\right)}{d_{i,j,H}}
=
\sum_{k\in U_i(H)}
\alpha_{i,k}\bar\gamma_{i,j,k}(H).
\]

For any remaining true link \(j\in U_i(H)\), the triangle inequality gives
\[
|m_{i,j}(H)|
=
\left|
\sum_{k\in U_i(H)}
\alpha_{i,k}\bar\gamma_{i,j,k}(H)
\right|
\]
\[
\ge
|\alpha_{i,j}|\,|\bar\gamma_{i,j,j}(H)|
-
\sum_{\substack{k\in U_i(H)\\ k\neq j}}
|\alpha_{i,k}|\,|\bar\gamma_{i,j,k}(H)|.
\]
Taking the maximum over \(j\in U_i(H)\), then the infimum over
\((i,H)\) with \(H\in\mathcal H_i^0\), Assumption~\ref{ass:t_signal}(ii)
implies
\[
\inf_{i,H\in\mathcal H_i^0}
\max_{j\in U_i(H)}
|m_{i,j}(H)|
\ge c_t b_T .
\]
Since \(\sqrt T b_T/\sqrt{\log N}\to\infty\), part (i) follows.

For part (ii), Assumption~\ref{ass:t_dominance} gives, for every
\(H\in\mathcal H_i^0\),
\[
\max_{p\in P_i(H)}|m_{i,p}(H)|
\le
\lambda\max_{j\in U_i(H)}|m_{i,j}(H)|.
\]
Hence the gap in part (ii) is at least
$(1-\lambda)\max_{j\in U_i(H)}|m_{i,j}(H)|$. Part (i) and
$\lambda<1$ imply the required divergence.

Finally, part (iii) is exactly the post-completion condition in
Assumption~\ref{ass:t_dominance}.
\end{proof}

\begin{lemma}
\label{lem:app_t_proxy_dominance}
Under Assumptions~\ref{ass:t_signal} and~\ref{ass:t_dominance}, the scalar
population proxy-dominance conditions hold. Specifically, uniformly over
\(i\) and \(H\in\mathcal H_i^0\),
\[
\max_{s\in U_i(H)} |m_{i,s}(H)|
-
\max_{p\in P_i(H)} |m_{i,p}(H)|
=
\rho^t_{i,H},
\]
where
\[
\inf_{i,H\in\mathcal H_i^0}
\frac{\sqrt T\,\rho^t_{i,H}}{\sqrt{\log N}}
\to\infty.
\]
Moreover, uniformly over \(i\) and \(H\in\mathcal H_i^1\),
\[
\max_{p\in P_i(H)}
|m_{i,p}(H)|
=
o\!\left(\sqrt{\frac{\log N}{T}}\right).
\]
\end{lemma}

\begin{proof}
By the normalised population IV screening decomposition,
\[
m_{i,j}(H)
=
\sum_{k\in U_i(H)}
\alpha_{i,k}\bar\gamma_{i,j,k}(H).
\]
Therefore, for \(H\in\mathcal H_i^0\),
Assumption~\ref{ass:t_dominance} implies
\[
\max_{p\in P_i(H)} |m_{i,p}(H)|
\le
\lambda
\max_{s\in U_i(H)} |m_{i,s}(H)|.
\]
Hence
\[
\max_{s\in U_i(H)} |m_{i,s}(H)|
-
\max_{p\in P_i(H)} |m_{i,p}(H)|
\ge
(1-\lambda)
\max_{s\in U_i(H)} |m_{i,s}(H)|.
\]
By Assumption~\ref{ass:t_signal}, together with Lemma~\ref{lem:app_primitive_t_signal},
\[
\inf_{i,H\in\mathcal H_i^0}
\max_{s\in U_i(H)}
\frac{\sqrt T\,|m_{i,s}(H)|}{\sqrt{\log N}}
\to\infty.
\]
Since \(1-\lambda>0\), the same divergence holds for
\[
\rho^t_{i,H}
=
\max_{s\in U_i(H)} |m_{i,s}(H)|
-
\max_{p\in P_i(H)} |m_{i,p}(H)|.
\]
The post-completion statement is imposed directly by the second part of
Assumption~\ref{ass:t_dominance}.
\end{proof}

\begin{lemma}\label{lem:app_t_proxy}
Under Assumptions~\ref{ass:sampling}--\ref{ass:tails},
\ref{ass:t_signal}, and \ref{ass:t_dominance}, with probability tending to one,
the largest scalar statistic among \(\mathcal K_i\setminus H\) at every pre-completion conditioning set \(H\in\mathcal H_i^0\) is attained by a true link. At every post-completion conditioning set \(H\in\mathcal H_i^1\), remaining proxies fail to cross $C_t\sqrt{\log N}$ for $C_t$ sufficiently large.
\end{lemma}

\begin{proof}
For \(H\in\mathcal H_i^0\), write
$M_s=\max_{s\in U_i(H)}|m_{i,s}(H)|$ and
$M_p=\max_{p\in P_i(H)}|m_{i,p}(H)|$. By
Assumption~\ref{ass:t_dominance}, $M_p\le\lambda M_s$. Lemma~\ref{lem:app_t_signal}
therefore gives
\[
\max_{s\in U_i(H)}|t_{i,s,H}|
-\max_{p\in P_i(H)}|t_{i,p,H}|
\ge
\sqrt T\,M_s\{1-\lambda-\op(1)\}-\Op(\sqrt{\log N}),
\]
which is positive uniformly with probability tending to one by
Lemma~\ref{lem:app_primitive_t_signal}(i). For \(H\in\mathcal H_i^1\), the
post-completion condition and the fixed-envelope form of
Lemma~\ref{lem:app_uniform_dgk}, together with the first-stage replacement bounds,
show that remaining proxy statistics are below $C_t\sqrt{\log N}$ with probability
tending to one when $C_t$ is sufficiently large.
\end{proof}

\subsection{\textbf{Block screening lemmas}}

\begin{lemma}
\label{lem:primitive_F_signal}
Let $Q^F_{i,j,H}$ be defined by Eq. \eqref{eq:QF}.
Under Assumptions~\ref{ass:F_signal} and~\ref{ass:F_dominance}, the following hold uniformly over the stated sets:

\begin{enumerate}[label=(\roman*)]

\item For \(H\in\mathcal H_i^0\),
\[
\inf_{i,H\in\mathcal H_i^0}
\max_{j\in U_i(H)}
\frac{T\,Q^F_{i,j,H}}{\log N}
\to\infty .
\]

\item For \(H\in\mathcal H_i^1\),
\[
\sup_{i,H\in\mathcal H_i^1}
\max_{p\in P_i(H)}
\frac{T\,Q^F_{i,p,H}}{\log N}
\to0 .
\]

\end{enumerate}
\end{lemma}

\begin{proof}
Part (i) is exactly Assumption~\ref{ass:F_signal}. Part (ii) follows
directly from the second part of Assumption~\ref{ass:F_dominance}.
\end{proof}

\begin{lemma}
\label{lem:app_F_proxy_dominance}
Under Assumptions~\ref{ass:F_signal} and~\ref{ass:F_dominance}, the block
population proxy-dominance conditions hold. Specifically, uniformly over
\(i\) and \(H\in\mathcal H_i^0\),
\[
\max_{s\in U_i(H)}Q^F_{i,s,H}
-
\max_{p\in P_i(H)}Q^F_{i,p,H}
=
\rho^F_{i,H},
\]
where
\[
\inf_{i,H\in\mathcal H_i^0}
\frac{
T\rho^F_{i,H}
}{
\sqrt{T\bar Q^F_{i,H}\log N}+\log N
}
\to\infty,
\qquad
\bar Q^F_{i,H}
=
\max_{j\in\mathcal K_i\setminus H}Q^F_{i,j,H}.
\]
Moreover, uniformly over \(i\) and \(H\in\mathcal H_i^1\),
\[
\max_{p\in P_i(H)}
\frac{TQ^F_{i,p,H}}{\log N}
\to0 .
\]
\end{lemma}

\begin{proof}
By Assumption~\ref{ass:F_dominance},
\[
\max_{p\in P_i(H)}Q^F_{i,p,H}
\le
\lambda_F\max_{s\in U_i(H)}Q^F_{i,s,H}.
\]
Hence
\[
\rho^F_{i,H}
\ge
(1-\lambda_F)
\max_{s\in U_i(H)}Q^F_{i,s,H}.
\]
Assumption~\ref{ass:F_signal} gives
\[
\inf_{i,H\in\mathcal H_i^0}
\frac{T\max_{s\in U_i(H)}Q^F_{i,s,H}}{\log N}
\to\infty.
\]
Together with the proxy attenuation condition, \(\bar Q^F_{i,H}\) is of the
same order as the strongest remaining true-link block signal before completion.
Therefore,
\[
\inf_{i,H\in\mathcal H_i^0}
\frac{
T\rho^F_{i,H}
}{
\sqrt{T\bar Q^F_{i,H}\log N}+\log N
}
\to\infty.
\]
The post-completion statement is imposed directly by the second part of
Assumption~\ref{ass:F_dominance}.
\end{proof}
\begin{lemma}\label{lem:app_F_rep}
The block statistic satisfies
\[
F_{i,j,H}
=
\frac{
\mathbf S_{i,j,H}'
\mathbf G_{i,j,H}^{-1}
\mathbf S_{i,j,H}
}
{\widehat\sigma_{i,j,H}^2}.
\]
\end{lemma}

\begin{proof}
After applying the IV/FWL residualisation, the IV regression is the regression of \(\widetilde{\mathbf y}_{i,j,H}\) on the fitted block \(\widehat{\mathbf V}_{i,j,H}\). The screening quadratic is
\[
\frac{1}{\widehat\sigma^2}
\widetilde{\mathbf y}_{i,j,H}'
\widehat{\mathbf V}_{i,j,H}
(\widehat{\mathbf V}_{i,j,H}'\widehat{\mathbf V}_{i,j,H})^{-1}
\widehat{\mathbf V}_{i,j,H}'
\widetilde{\mathbf y}_{i,j,H}.
\]
Substituting
\[
\mathbf S_{i,j,H}
=
T^{-1/2}
\widehat{\mathbf V}_{i,j,H}'
\widetilde{\mathbf y}_{i,j,H}
\]
and
\[
\mathbf G_{i,j,H}
=
T^{-1}
\widehat{\mathbf V}_{i,j,H}'
\widehat{\mathbf V}_{i,j,H}
\]
gives the displayed expression.
\end{proof}

\begin{lemma}\label{lem:app_F_denom}
Under Assumptions~\ref{ass:iv}--\ref{ass:tails} and~\ref{ass:block_first_stage},
\[
\inf_{i,j,H\in\mathcal H_i}
\lambda_{\min}(\mathbf G_{i,j,H})>c,
\qquad
\inf_{i,j,H\in\mathcal H_i}
\widehat\sigma_{i,j,H}^2>c
\]
with probability tending to one.
\end{lemma}

\begin{proof}
Each element of \(\mathbf G_{i,j,H}\) is a sample average of a product of fitted block components. Replacing sample fitted components by population fitted components is \(\op(1)\) by Lemma~\ref{lem:app_first_stage}. The population-fitted sample averages converge uniformly to \(\mathbf G^*_{i,j,H}\) by Lemma~\ref{lem:app_uniform_dgk}. Since the block dimension is fixed, elementwise convergence implies operator-norm convergence. Weyl's inequality and Assumption~\ref{ass:block_first_stage} give the eigenvalue lower bound. The residual variance bound is the same as in Lemma~\ref{lem:app_t_denom}.
\end{proof}

\begin{lemma}\label{lem:app_F_signal}
Uniformly over \(i\), \(H\in\mathcal H_i\), and \(j\in\mathcal K_i\setminus H\),
\[
F_{i,j,H}
=
TQ^F_{i,j,H}\{1+\op(1)\}
+
\Op\!\left(
\sqrt{TQ^F_{i,j,H}\log N}
\right)
+
\Op(\log N).
\]
\end{lemma}

\begin{proof}
The score vector has the expansion
\[
\mathbf S_{i,j,H}
=
T^{-1/2}
\sum_t
\widehat{\mathbf V}_{i,j,H,t}
\mu_{i,j,H,t}
+
T^{-1/2}
\sum_t
\widehat{\mathbf V}_{i,j,H,t}
\varepsilon_{i,t}.
\]
Replacing \(\widehat{\mathbf V}_{i,j,H,t}\) by
\(\mathbf V^*_{i,j,H,t}\) is negligible uniformly. The first component equals
\[
\sqrt T\,\mathbf m^F_{i,j,H}
+
\Op(\sqrt{\log N}),
\]
and the second component is \(\Op(\sqrt{\log N})\). Hence
\[
\mathbf S_{i,j,H}
=
\sqrt T\,\mathbf m^F_{i,j,H}
+
\mathbf R_{i,j,H},
\qquad
\|\mathbf R_{i,j,H}\|
=
\Op(\sqrt{\log N})
\]
uniformly.

Using Lemma~\ref{lem:app_F_denom}, expand the quadratic form:
\[
\begin{aligned}
\mathbf S_{i,j,H}'
\mathbf G_{i,j,H}^{-1}
\mathbf S_{i,j,H}
&=
T
\mathbf m_{i,j,H}^{F\prime}
(\mathbf G^*_{i,j,H})^{-1}
\mathbf m^F_{i,j,H}
+
2\sqrt T\,
\mathbf m_{i,j,H}^{F\prime}
(\mathbf G^*_{i,j,H})^{-1}
\mathbf R_{i,j,H}
\\
&+
\mathbf R_{i,j,H}'
(\mathbf G^*_{i,j,H})^{-1}
\mathbf R_{i,j,H}
+
o_p\!\left(TQ^F_{i,j,H}\right)
+o_p(\log N).
\end{aligned}
\]
By definition of \(Q^F_{i,j,H}\), the first term equals
\[
T\sigma_{i,j,H}^{2}Q^F_{i,j,H}.
\]
The cross term is
\[
\Op\!\left(
\sqrt{
T\sigma_{i,j,H}^{2}Q^F_{i,j,H}\log N
}
\right),
\]
and the final quadratic remainder is \(\Op(\log N)\). Division by
\(\widehat\sigma_{i,j,H}^{2}\), together with the uniform variance bounds, gives
\[
F_{i,j,H}
=
TQ^F_{i,j,H}\{1+\op(1)\}
+
\Op\!\left(
\sqrt{TQ^F_{i,j,H}\log N}
\right)
+
\Op(\log N).
\]
\end{proof}

\begin{lemma}\label{lem:app_F_null}
Under Assumptions~\ref{ass:sampling}--\ref{ass:tails},
\ref{ass:block_first_stage}, and \ref{ass:F_signal}, there is a finite $C_0^F$ such
that, for every fixed $C>C_0^F$,
\[
\Pr\!\left(
\max_i\max_{H\in\mathcal H_i}\max_{j\in\mathcal S_i^d}
F_{i,j,H}>C\log N
\right)\to0.
\]
\end{lemma}

\begin{proof}
For an irrelevant block, the last part of Assumption~\ref{ass:F_signal} gives
$TQ^F_{i,j,H}=o(\log N)$ uniformly. The score expansion in the proof of
Lemma~\ref{lem:app_F_signal}, the fixed-envelope form of
Lemma~\ref{lem:app_uniform_dgk}, and the block first-stage replacement bounds imply
that the centred score remainder has squared norm below $C\log N$ uniformly with
probability tending to one for a sufficiently large $C$. Lemma~\ref{lem:app_F_denom}
then gives the displayed probability bound.
\end{proof}

\begin{lemma}\label{lem:app_F_proxy}
Under Lemmas~\ref{lem:primitive_F_signal},
\ref{lem:app_F_proxy_dominance},
and~\ref{lem:app_F_signal},
with probability tending to one, the largest block statistic among
\(\mathcal K_i\setminus H\) at every pre-completion conditioning set
\(H\in\mathcal H_i^0\) is attained by a true link. At every post-completion conditioning set \(H\in\mathcal H_i^1\), remaining proxy links fail to cross $C_F\log N$ for $C_F$ sufficiently large.
\end{lemma}

\begin{proof}
For \(H\in\mathcal H_i^0\), write
$Q_s=\max_{s\in U_i(H)}Q^F_{i,s,H}$ and
$Q_p=\max_{p\in P_i(H)}Q^F_{i,p,H}$. Then $Q_p\le\lambda_FQ_s$ and the
multiplicative expansion in Lemma~\ref{lem:app_F_signal} gives
\[
\max_{s\in U_i(H)}F_{i,s,H}
-\max_{p\in P_i(H)}F_{i,p,H}
\ge
TQ_s\{1-\lambda_F-\op(1)\}
-\Op\!\left(\sqrt{TQ_s\log N}+\log N\right).
\]
Assumption~\ref{ass:F_signal} makes the right-hand side positive uniformly with
probability tending to one. Hence, before completion, the largest block statistic
is attained by a true link.

For \(H\in\mathcal H_i^1\), Lemma~\ref{lem:app_F_proxy_dominance} and
Lemma~\ref{lem:app_F_signal}, together with the fixed-envelope score bound, imply
that all remaining proxy block statistics are below $C_F\log N$ with probability
tending to one when $C_F$ is sufficiently large.
\end{proof}

\subsection{\textbf{Proofs of the Main Theorems}}

\begin{proof}[\textbf{Proof of Theorem~\ref{thm:t_exact}}]

Let
\[
c_N^t
=
C_t\sqrt{\log N},
\qquad C_t\ \text{sufficiently large}.
\]

By Lemma~\ref{lem:app_t_null},
\[
\Pr
\left(
\exists i,\,
H\in\mathcal H_i,\,
j\in\mathcal S_i^d:
|t_{i,j,H}|>c_N^t
\right)
\to0
\]
for a sufficiently large threshold constant.

At every admissible pre-completion conditioning set
\(
H\in\mathcal H_i^0,
\)
Lemma~\ref{lem:app_primitive_t_signal}
implies that at least one remaining true link has a population screening signal
that dominates the stochastic screening error. Combined with the statistic
expansion in Lemma~\ref{lem:app_t_signal}, this implies that at least one
remaining true link has a statistic exceeding the stopping threshold.
Hence, for equations with at least one unselected true link, the procedure
cannot terminate before all true links have been selected.

Moreover, Lemma~\ref{lem:app_t_proxy}
implies that, with probability tending to one, the largest statistic among the
remaining true links exceeds the largest statistic among the remaining proxy
links at every pre-completion conditioning set. Since irrelevant candidates do
not cross the threshold asymptotically, every selected candidate before
completion belongs to
\[
U_i(H)
=
\mathcal S_i^n\setminus H.
\]

By Lemma~\ref{lem:app_path_induction}, the realised selection path remains
inside the deterministic low-dimensional conditioning-set class with
probability tending to one. Consequently, no proxy or irrelevant link is
selected before all true links have been selected. If
\(\mathcal S_i^n=\varnothing\), then the initial conditioning set is already
post-completion and the preceding argument is vacuous.

After completion, that is, for
\(
H\in\mathcal H_i^1,
\)
Lemma~\ref{lem:app_t_proxy}
implies that the statistics of all remaining proxy links lie below the stopping
threshold with probability tending to one. Irrelevant candidates are also below
threshold by Lemma~\ref{lem:app_t_null}. Therefore, the algorithm stops exactly
when the true-link set has been recovered. If
\(\mathcal S_i^n=\varnothing\), it stops immediately with probability tending
to one.

Hence,
\[
\widehat{\mathcal S}_i^t
=
\mathcal S_i^n
\]
for every equation \(i\) with probability tending to one, and therefore
\[
\Pr(\mathcal A_1^t)
\to
1.
\]
\end{proof}

\begin{proof}[\textbf{Proof of Proposition~\ref{prop:containment}}]
Let
\[
\mathcal C
=
\left\{
\mathcal S_i^n\subseteq \widehat{\mathcal S}_i^t
\subseteq\mathcal S_i^n\cup\mathcal S_i^p
\text{ for all }i
\right\}.
\]
By assumption, $\Pr(\mathcal C)\to1$. On $\mathcal C$, the post-selection
equation contains all variables that enter the structural equation for unit
$i$. Hence the selected model can be written as
\[
y_{i,t}
=
\eta_i
+
\boldsymbol{\beta}_i'\boldsymbol{x}_{i,t}
+
\sum_{j\in \mathcal S_i^n}
\alpha_{i,j}y_{j,t}
+
\sum_{j\in \widehat{\mathcal S}_i^t\setminus \mathcal S_i^n}
0\cdot y_{j,t}
+
\varepsilon_{i,t}.
\]
Thus, on $\mathcal C$, the coefficients associated with selected
proxy links are zero in the population, while the coefficients associated with
true links are the structural coefficients $\{\alpha_{i,j}:j\in\mathcal S_i^n\}$.

The maintained uniform post-selection IV regularity conditions imply consistency of
the IV estimator over all admissible augmented equations. Therefore,
\[
\widehat{\alpha}_{i,j}
\stackrel{p}{\to}
\alpha_{i,j},
\qquad j\in\mathcal S_i^n,
\]
and, for any selected proxy link
$j\in\widehat{\mathcal S}_i^t\setminus\mathcal S_i^n$,
\[
\widehat{\alpha}_{i,j}
\stackrel{p}{\to}
0.
\]
Since $\Pr(\mathcal C)\to1$, the same conclusions hold unconditionally. This
proves the result.
\end{proof}

\begin{proof}[\textbf{Proof of Theorem~\ref{thm:dualrecover}}]

By Theorem~\ref{thm:t_exact},
\[
\Pr
\left(
\widehat{\mathcal S}_i^t
=
\mathcal S_i^n
\text{ for all }i
\right)
\to1.
\]
On the exact-recovery event, condition Eq. \eqref{eq:post_sign_rate} and the beta-min
condition in Assumption~\ref{ass:t_signal}(i) imply that
$\operatorname{sign}(\widehat\alpha_{i,j})=\operatorname{sign}(\alpha_{i,j})$
uniformly over all true links with probability tending to one. This proves exact
recovery of both signed supports.

Since each row contains at most $\bar k$ nonzero entries,
\[
\max_i|\widehat\rho_{\ell,i}-\rho_{\ell,i}|
\le
\bar k\max_i\max_{j\in\mathcal S_i^n}
|\widehat\alpha_{i,j}-\alpha_{i,j}|=o_p(1).
\]
On active rows the denominator $|\rho_{\ell,i}|=|\rho_\ell|$ is bounded away
from zero under Assumption~\ref{ass:dual}; on inactive rows both the population
and estimated signed rows are zero with probability tending to one. The stated
uniform row-$\ell_1$ convergence of the normalised weights follows by the
continuous-mapping theorem.

\end{proof}

\begin{proof}[\textbf{Proof of Theorem~\ref{thm:eqoracle}}]

Let $\mathcal B_1^t$ be the event that all estimated coefficients on true links
have the correct signs. Condition Eq. \eqref{eq:post_sign_rate} and the beta-min
condition imply $\Pr(\mathcal B_1^t)\to1$. On
\(\mathcal A_1^t\cap\mathcal B_1^t\), the feasible post-selection IV estimator
and the oracle IV estimator use the same regressors, instruments, and signed
parameter map, up to column ordering. Hence
\[
\widehat{\boldsymbol{\theta}}_i
=
\widehat{\boldsymbol{\theta}}_i^o
\]
on \(\mathcal A_1^t\cap\mathcal B_1^t\). Since this event has probability tending to one,
the stated $o_p(T^{-1/2})$
equivalence follows.
\end{proof}

\begin{proof}[\textbf{Proof of Theorem~\ref{thm:t_panel}}]

Let $\mathcal B_1^t$ be the sign-recovery event defined in the preceding proof.
On \(\mathcal A_1^t\cap\mathcal B_1^t\),
$\widehat{\boldsymbol{\theta}}_{MG}
= \widehat{\boldsymbol{\theta}}_{MG}^o$.
Since this event has probability tending to one,
\[
\widehat{\boldsymbol{\theta}}_{MG}
-
\widehat{\boldsymbol{\theta}}_{MG}^o
=
o_p((NT)^{-1/2}).
\]

Assumption~\ref{ass:panel} gives
\[
\sqrt{NT}
\left(
\widehat{\boldsymbol{\theta}}_{MG}^o
-
\boldsymbol{\theta}_N
\right)
\Rightarrow
\mathcal N(\mathbf 0,\mathbf V).
\]

Slutsky's theorem yields the result.

\end{proof}

\begin{proof}[\textbf{Proof of Theorem~\ref{thm:F_exact}}]

Let
\[
c_N^F=C_F\log N,
\qquad C_F\ \text{sufficiently large}.
\]

By Lemma~\ref{lem:app_F_null},
\[
\Pr
\left(
\exists i,\,
H\in\mathcal H_i,\,
j\in\mathcal S_i^d:
F_{i,j,H}>c_N^F
\right)
\to0
\]
for a sufficiently large threshold constant.

At every admissible pre-completion model
\(H\in\mathcal H_i^0\),
Assumption~\ref{ass:F_signal}
and Lemma~\ref{lem:app_F_signal}
imply that at least one remaining true link exceeds the threshold. Hence, for equations with at least one unselected true link, the procedure does not stop before all true links have been selected.

Moreover, Lemma~\ref{lem:app_F_proxy}
implies that, with probability tending to one, the largest block statistic among remaining true links exceeds the largest block statistic among remaining proxy links at any pre-completion conditioning set. Since irrelevant candidates do not cross the threshold asymptotically, any selected candidate before completion belongs to
\[
\mathcal S_i^n\setminus H.
\]

By Lemma~\ref{lem:app_path_induction}, the realised selection path remains inside the low-dimensional conditioning-set class. Consequently, no proxy or irrelevant link is selected before all true links have been selected. If \(\mathcal S_i^n=\varnothing\), the initial model is already post-completion and the preceding pre-completion argument is vacuous.

After completion, that is, for
\(H\in\mathcal H_i^1\),
the post-completion part of Assumption~\ref{ass:F_dominance}
and Lemma~\ref{lem:app_F_proxy}
imply that the statistics of all remaining proxy links lie below the stopping threshold with probability tending to one. Irrelevant candidates are also below threshold by the null bound. Therefore the algorithm stops once the true link set has been recovered; if the true link set is empty, it stops at the initial model with probability tending to one.

Hence
\[
\widehat{\mathcal S}_i^F
=
\mathcal S_i^n
\]
for every equation \(i\) with probability tending to one, and therefore
\[
\Pr(\mathcal A_1^F)\to1.
\]
\end{proof}

\subsection{Factor-purged variables}

Suppose
\[
\varepsilon_{i,t}
=
\boldsymbol{\lambda}_i'
\mathbf F_t
+
u_{i,t}.
\]
Let \(\widehat{\mathbf F}\) denote observed or estimated factor proxies included among the always-included regressors, and set
\[
\mathbf M_{\widehat{\mathbf F}}
=
\mathbf I_T
-
\widehat{\mathbf F}
(
\widehat{\mathbf F}'
\widehat{\mathbf F}
)^{-1}
\widehat{\mathbf F}'.
\]

\begin{lemma}\label{lem:app_factor}
Assume the residual factor component is negligible on the screening-score scale:
\[
\sup_{i,j,H}
\left|
T^{-1/2}
\sum_{t=1}^T
v^*_{i,j,H,t}
(\mathbf M_{\widehat{\mathbf F}}\mathbf F\boldsymbol{\lambda}_i)_t
\right|
=
o_p(\sqrt{\log N}),
\]
and, for block screening,
\[
\sup_{i,j,H}\max_{1\le a\le d}
\left|
T^{-1/2}
\sum_{t=1}^T
V^*_{i,j,H,t,a}
(\mathbf M_{\widehat{\mathbf F}}\mathbf F\boldsymbol{\lambda}_i)_t
\right|
=
o_p(\sqrt{\log N}).
\]
The same two bounds hold with $v^*_{i,j,H,t}$ and $V^*_{i,j,H,t,a}$ replaced,
respectively, by $\widehat v_{i,j,H,t}-v^*_{i,j,H,t}$ and
$\widehat V_{i,j,H,t,a}-V^*_{i,j,H,t,a}$.
Assume also that, after applying
\(\mathbf M_{\widehat{\mathbf F}}\),
the idiosyncratic component \(u_{i,t}\) satisfies the same mixing, tail, first-stage, relevance, and signal conditions used above. Then the scalar and block selection results continue to hold for the factor-purged variables.
\end{lemma}

\begin{proof}
Every score term in the preceding arguments is an inner product between a fitted regressor or fitted block and the regression disturbance. After factor partialling-out,
\[
\mathbf M_{\widehat{\mathbf F}}
\boldsymbol{\varepsilon}_i
=
\mathbf M_{\widehat{\mathbf F}}
\mathbf F
\boldsymbol{\lambda}_i
+
\mathbf M_{\widehat{\mathbf F}}
\mathbf u_i.
\]
The assumptions in Lemma~\ref{lem:app_factor} state precisely that the first component is negligible after multiplication by every feasible partial fitted scalar candidate and by every component of every feasible partial fitted block. Thus it is negligible on the \(\sqrt{\log N}\) score scale used in the null and signal bounds. The second component satisfies the same DGK admissibility conditions as the baseline disturbance. Therefore the probability bounds and selection arguments apply with \(\mathbf u_i\) replacing \(\boldsymbol{\varepsilon}_i\).
\end{proof}

\clearpage

\renewcommand{\thetable}{B.\arabic{table}}
\setcounter{table}{0}
\renewcommand{\thefigure}{B.\arabic{figure}}
\setcounter{figure}{0}

\section*{Appendix B: Futher Simulation Results}
\addcontentsline{toc}{section}{Appendix B: Simulation Results}
\phantomsection
\label{sec:AppendixB}

\clearpage
\begingroup
\pagestyle{empty}
\thispagestyle{empty}
\begin{figure}[p]
\centering
\begin{minipage}{0.8\linewidth}
\centering
\includegraphics[width=\linewidth]{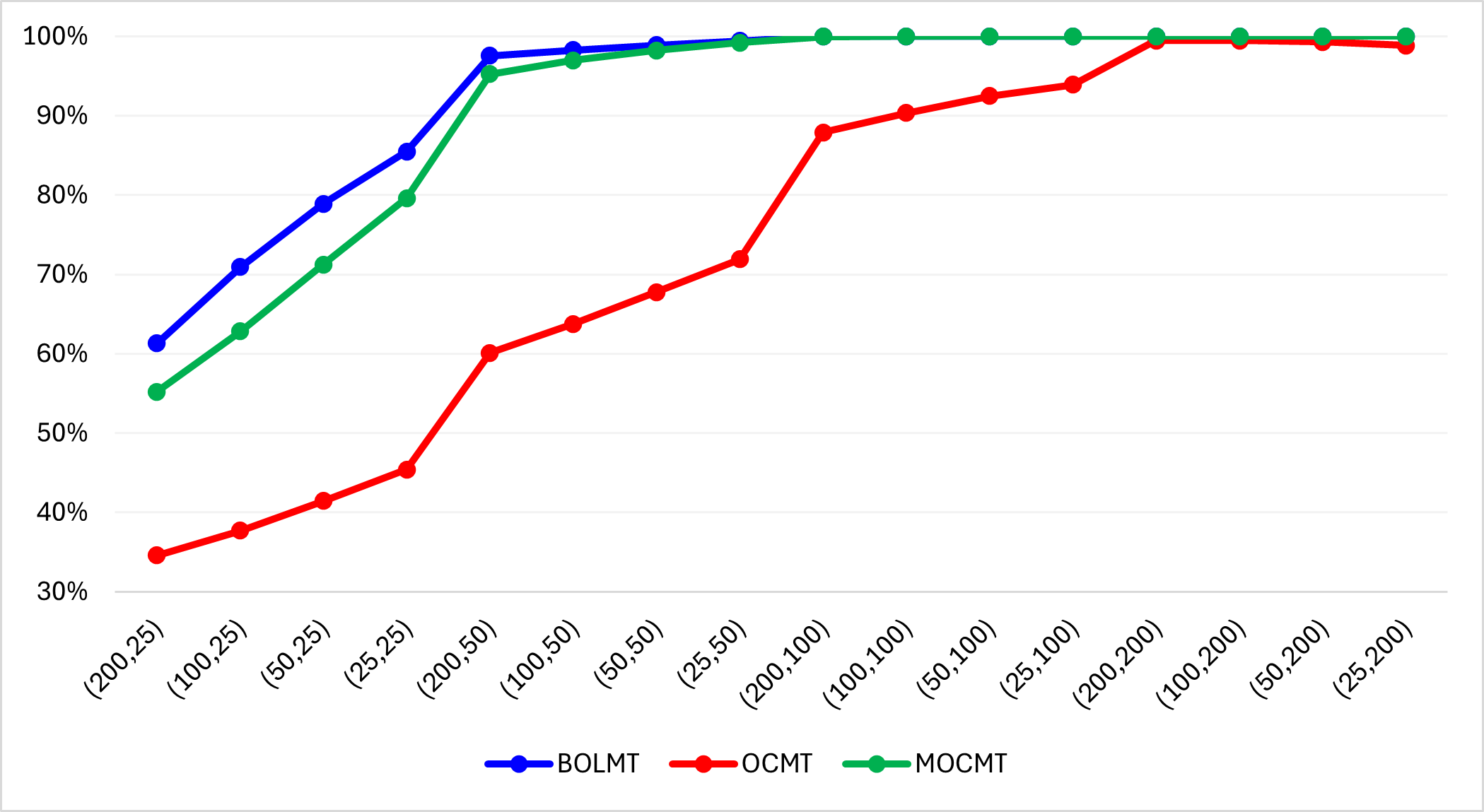}\\
\par\vspace{-0.5em}
\textbf{Circular $\boldsymbol{W}$}
\end{minipage}
\par\vspace{1.1em}
\begin{minipage}{0.8\linewidth}
\centering
\includegraphics[width=\linewidth]{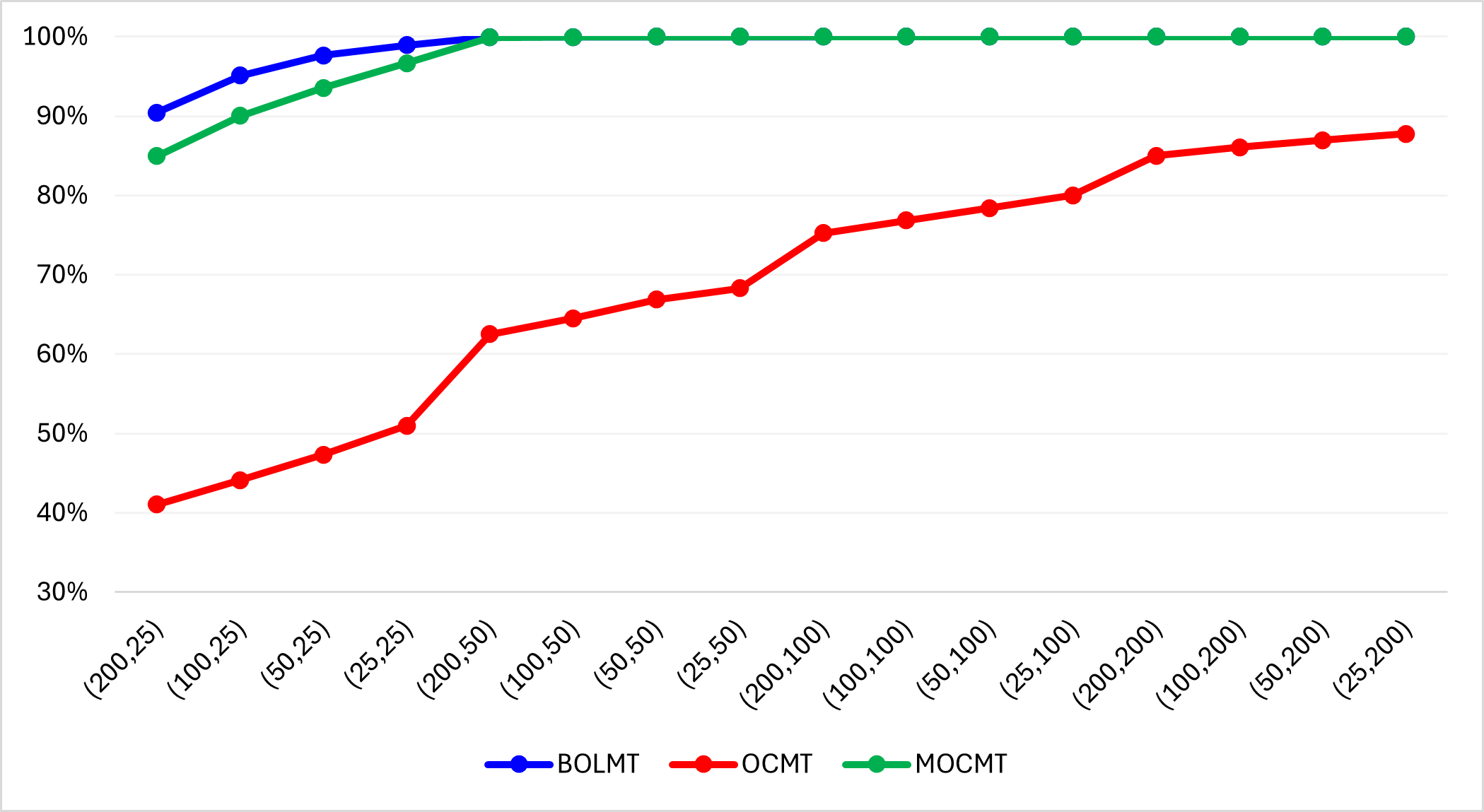}\\
\par\vspace{-0.5em}
\textbf{Block-Diagonal $\boldsymbol{W}$}
\end{minipage}
\par\vspace{1.1em}
\begin{minipage}{0.8\linewidth}
\centering
\includegraphics[width=\linewidth]{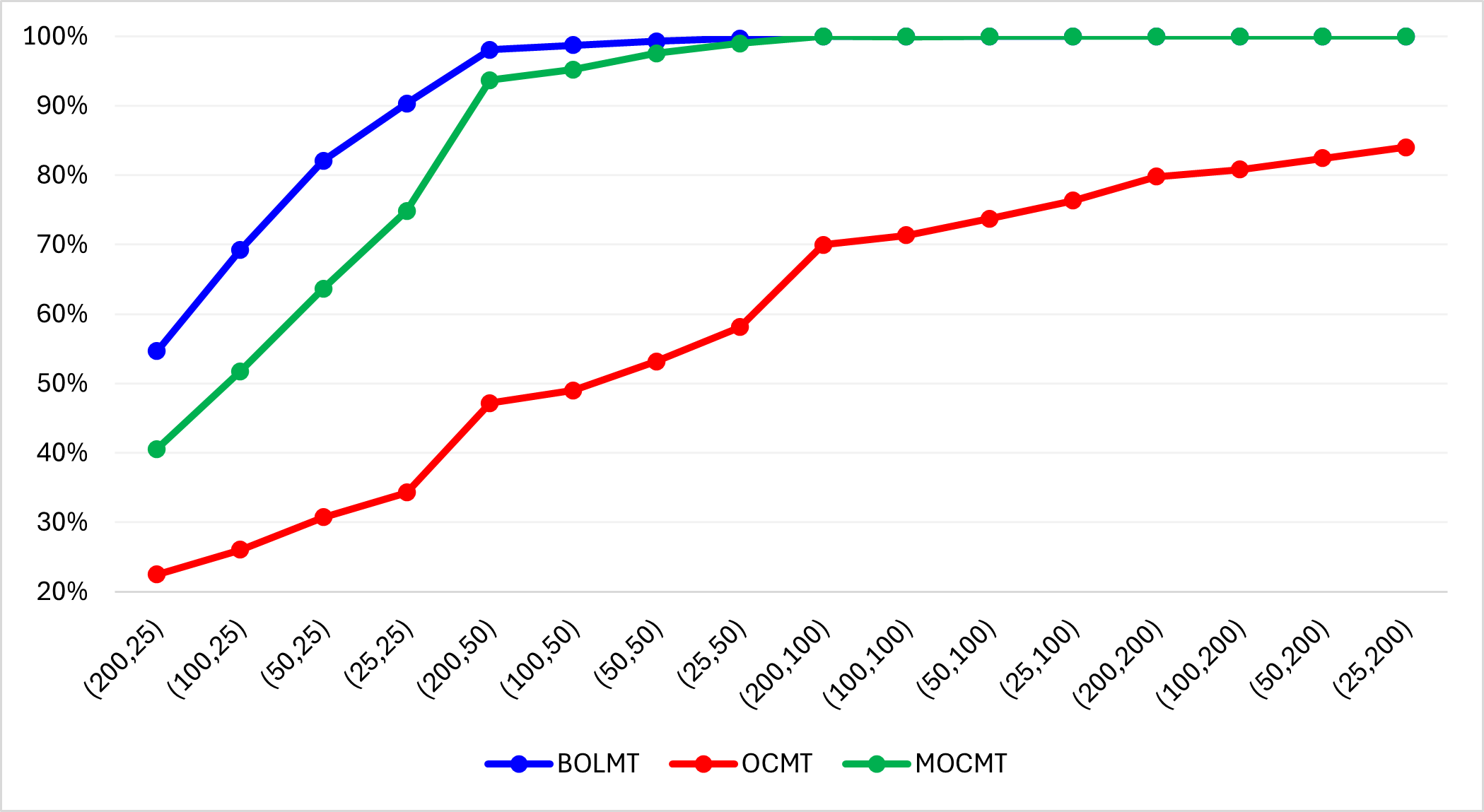}\\
\par\vspace{-0.5em}
\textbf{Random $\boldsymbol{W}$ with Dominant Unit}
\end{minipage}
\vspace{-0.4em}
\captionsetup{font=large}
\caption{TPR performance evaluation over different $(T,N)$ values}
\label{fig:TPR}
\end{figure}
\clearpage
\endgroup
\pagestyle{plain}

\clearpage
\begingroup
\pagestyle{empty}
\thispagestyle{empty}
\begin{figure}[p]
\centering
\begin{minipage}{0.8\linewidth}
\centering
\includegraphics[width=\linewidth]{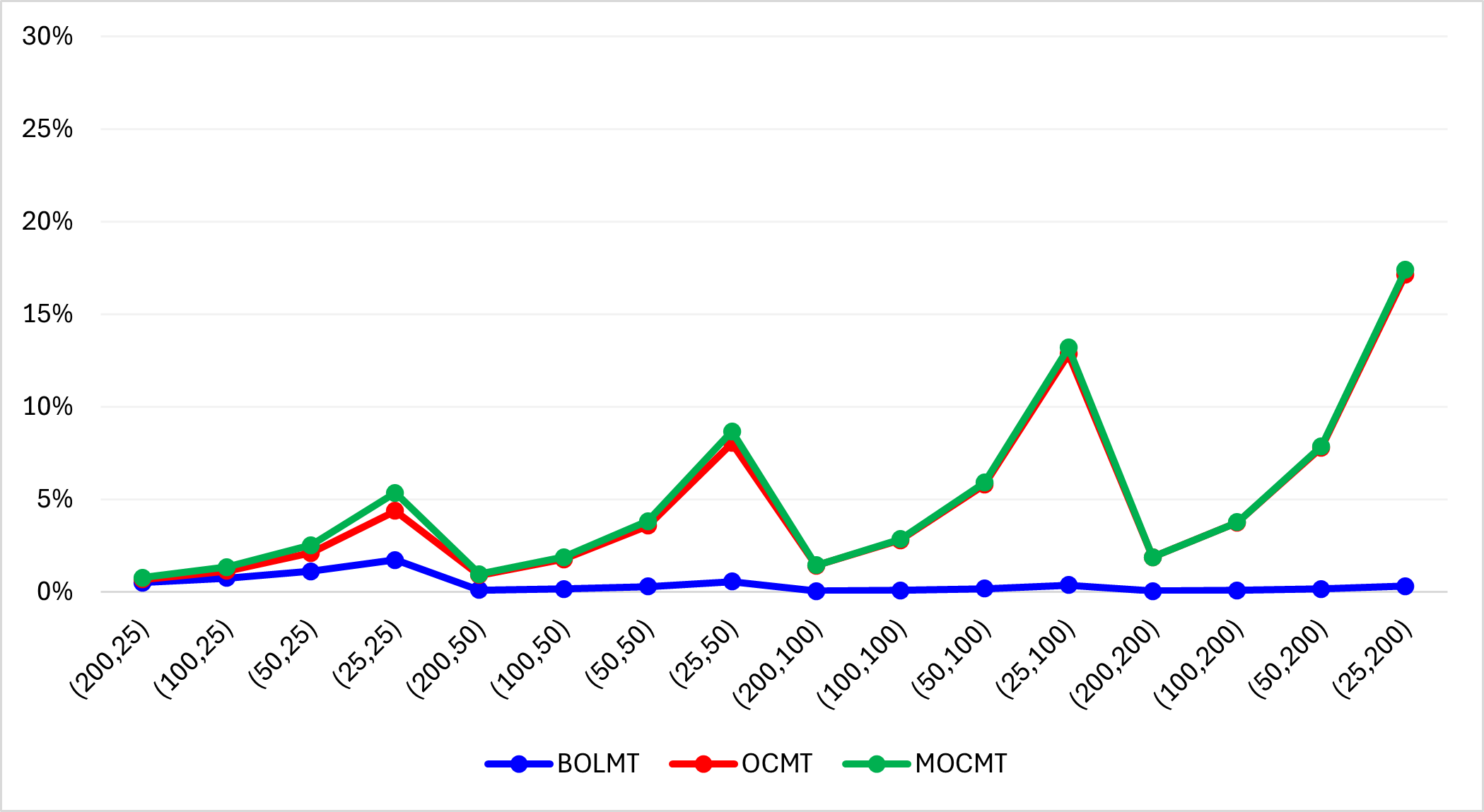}\\
\par\vspace{-0.5em}
\textbf{Circular $\boldsymbol{W}$}
\end{minipage}
\par\vspace{1.1em}
\begin{minipage}{0.8\linewidth}
\centering
\includegraphics[width=\linewidth]{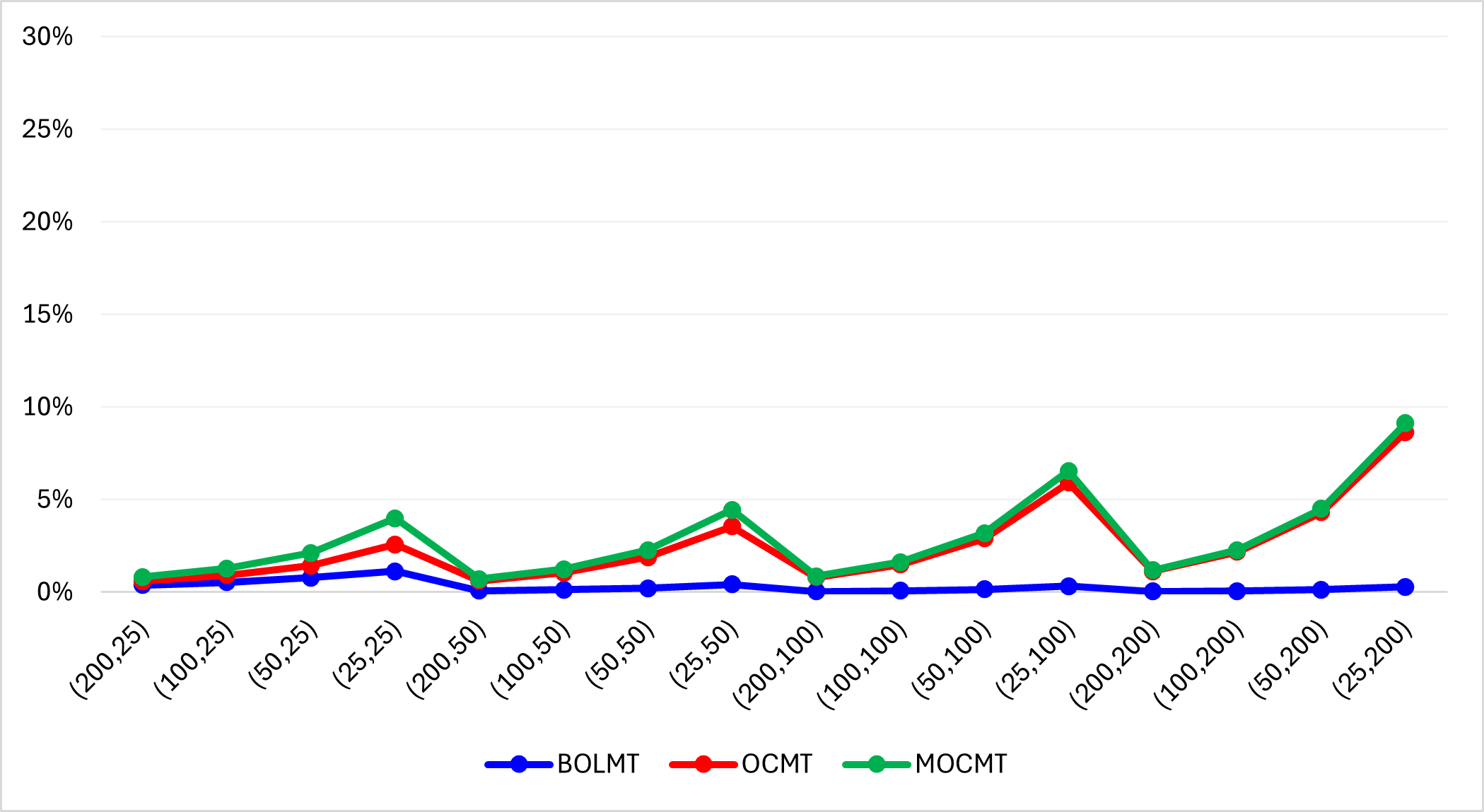}\\
\par\vspace{-0.5em}
\textbf{Block-Diagonal $\boldsymbol{W}$}
\end{minipage}
\par\vspace{1.1em}
\begin{minipage}{0.8\linewidth}
\centering
\includegraphics[width=\linewidth]{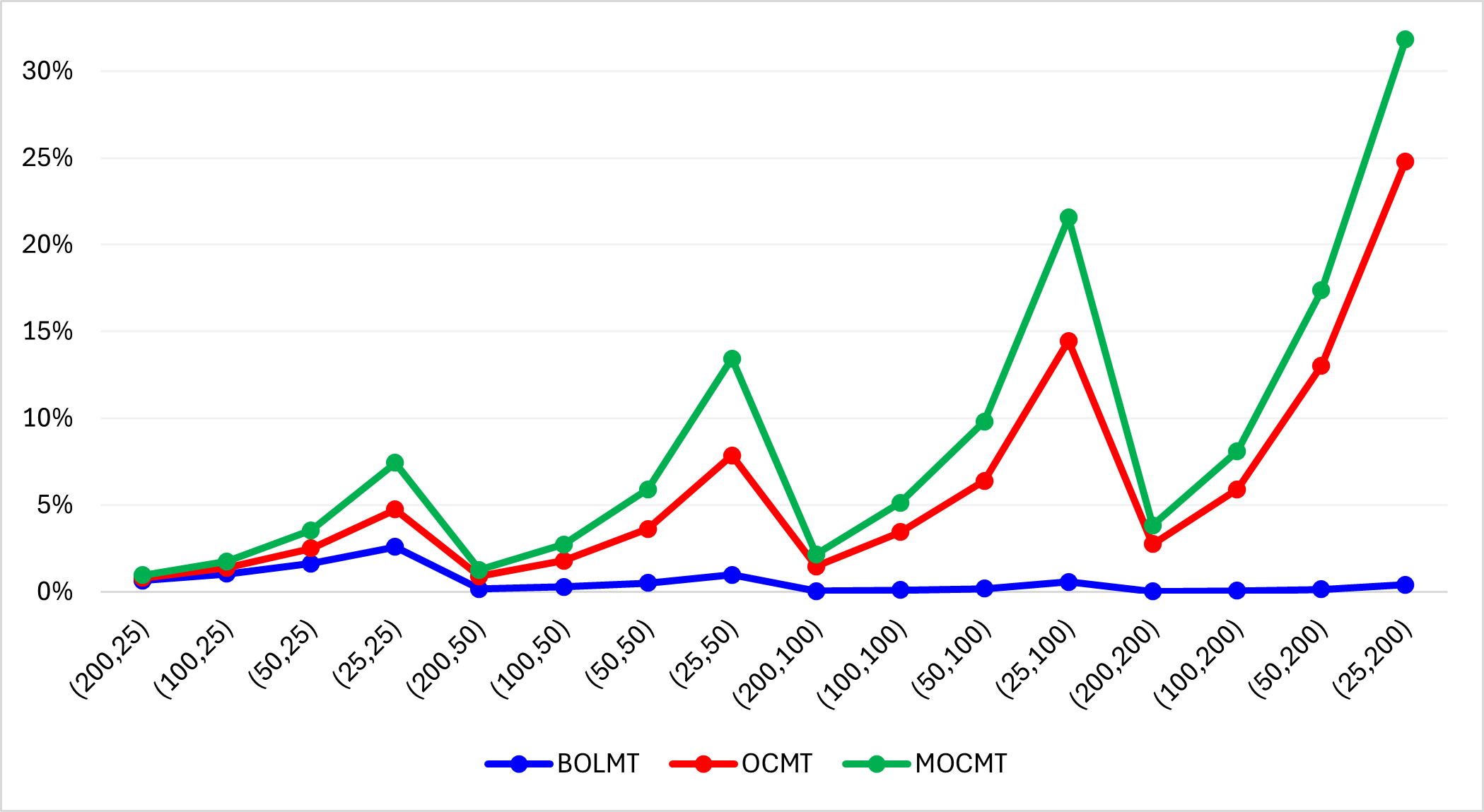}\\
\par\vspace{-0.5em}
\textbf{Random $\boldsymbol{W}$ with Dominant Unit}
\end{minipage}
\vspace{-0.4em}
\captionsetup{font=large}
\caption{FPR performance evaluation over different $(T,N)$ values}
\label{fig:FPR}
\end{figure}
\clearpage
\endgroup
\pagestyle{plain}

\clearpage
\begingroup
\pagestyle{empty}
\thispagestyle{empty}
\begin{figure}[p]
\centering
\begin{minipage}{0.8\linewidth}
\centering
\includegraphics[width=\linewidth]{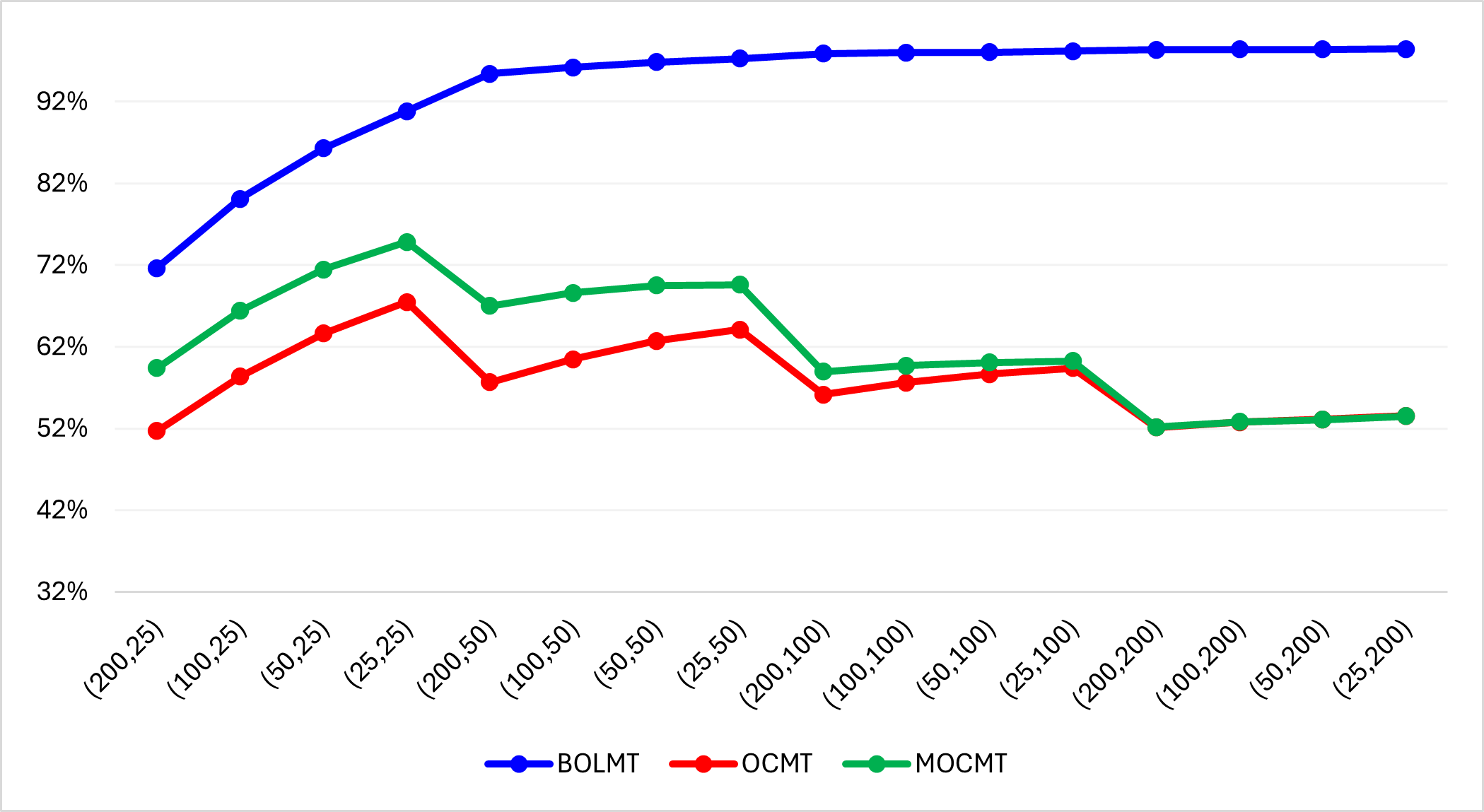}\\
\par\vspace{-0.5em}
\textbf{Circular $\boldsymbol{W}$}
\end{minipage}
\par\vspace{1.1em}
\begin{minipage}{0.8\linewidth}
\centering
\includegraphics[width=\linewidth]{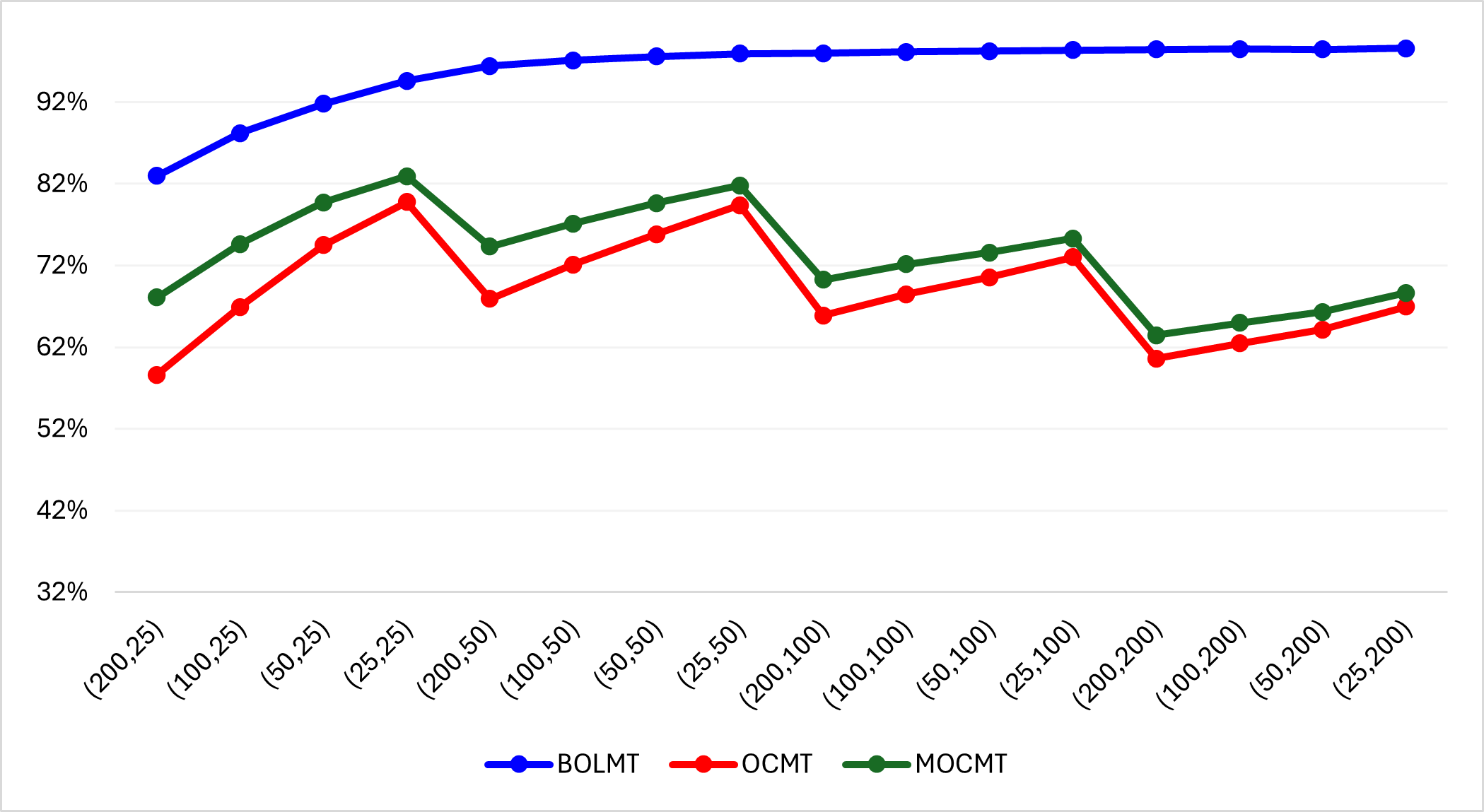}\\
\par\vspace{-0.5em}
\textbf{Block-Diagonal $\boldsymbol{W}$}
\end{minipage}
\par\vspace{1.1em}
\begin{minipage}{0.8\linewidth}
\centering
\includegraphics[width=\linewidth]{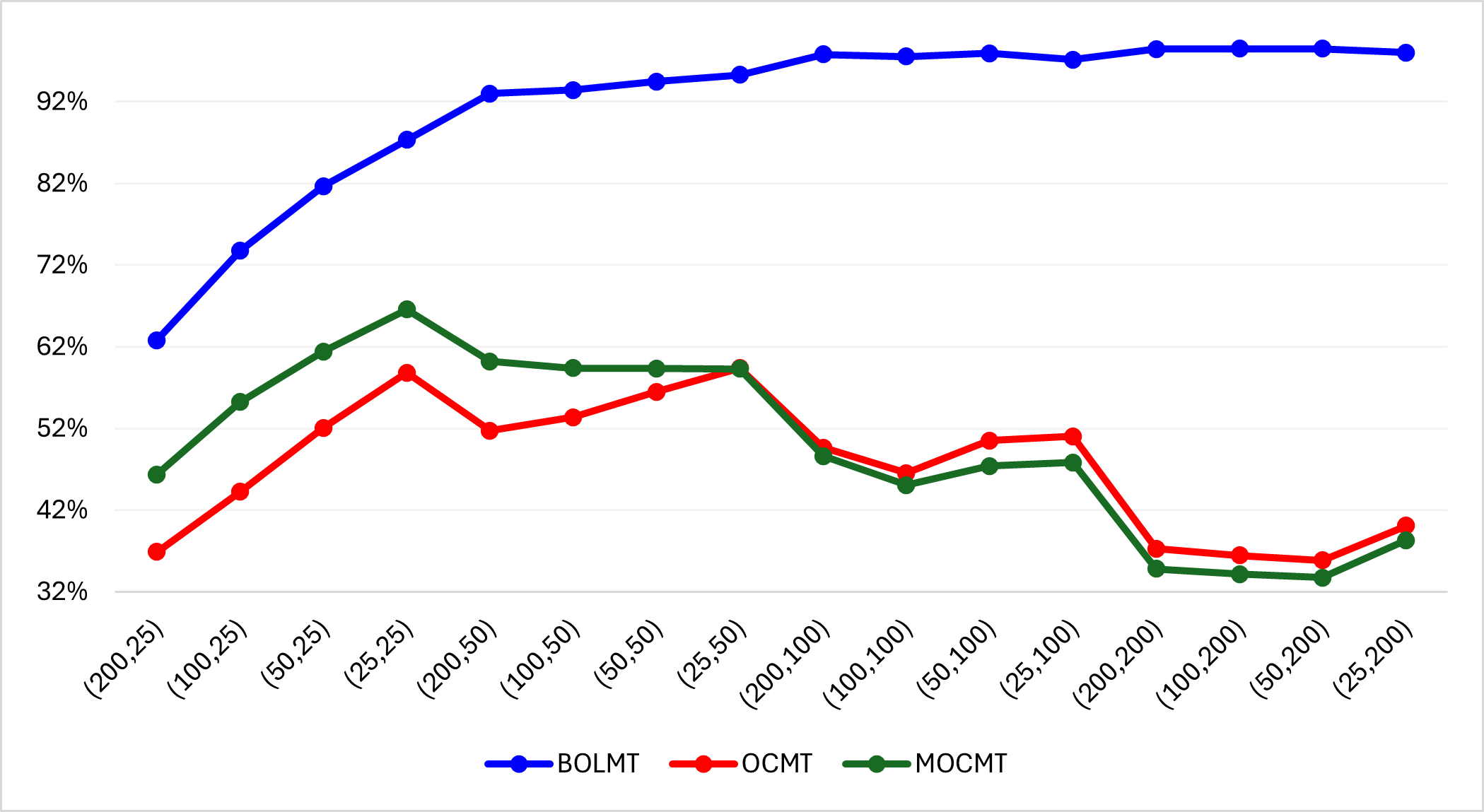}\\
\par\vspace{-0.5em}
\textbf{Random $\boldsymbol{W}$ with Dominant Unit}
\end{minipage}
\vspace{-0.4em}
\captionsetup{font=large}
\caption{TDR performance evaluation over different $(T,N)$ values}
\label{fig:TDR}
\end{figure}
\clearpage
\endgroup
\pagestyle{plain}

\clearpage
\begingroup
\pagestyle{empty}
\thispagestyle{empty}
\begin{figure}[p]
\centering
\begin{minipage}{0.8\linewidth}
\centering
\includegraphics[width=\linewidth]{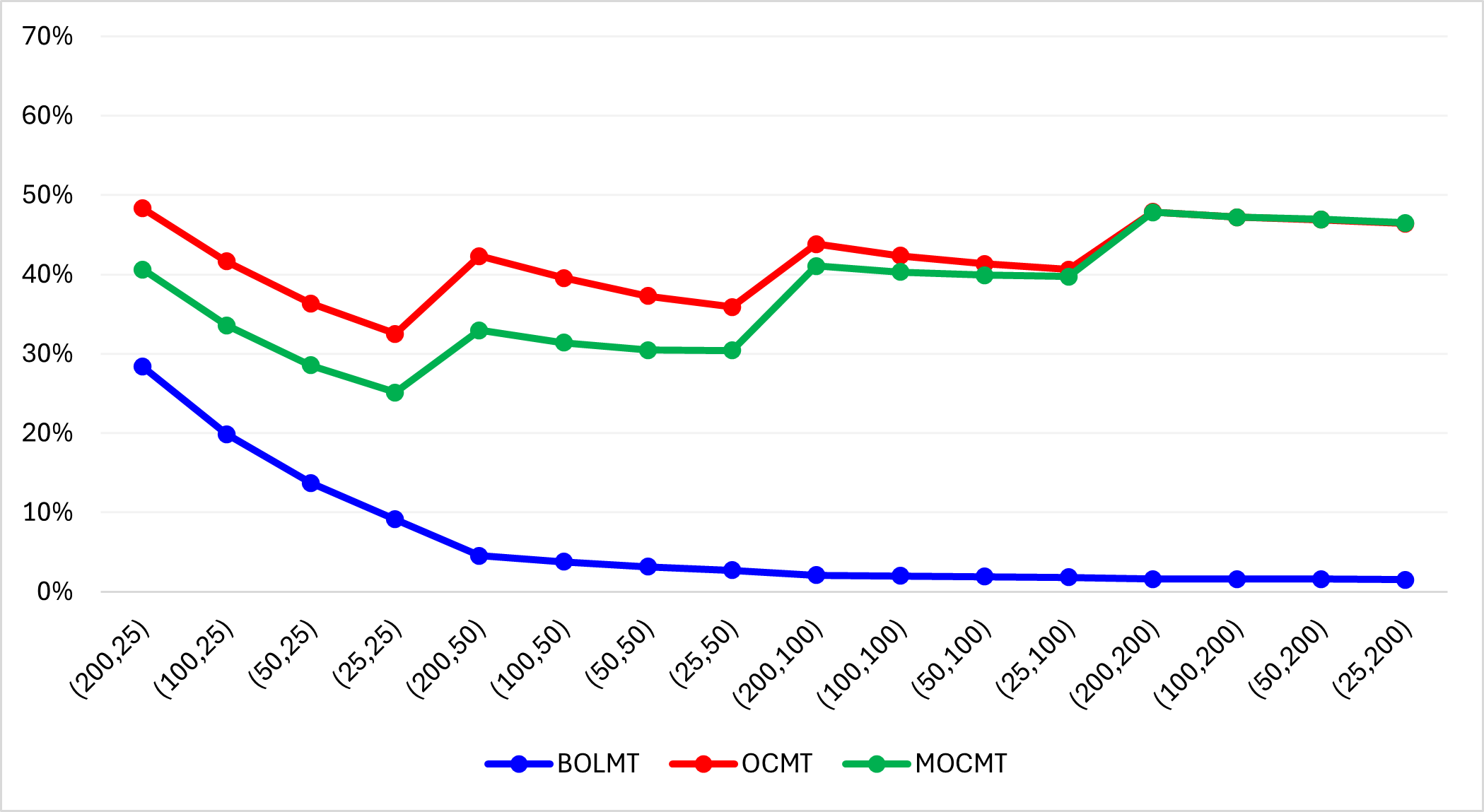}\\
\par\vspace{-0.5em}
\textbf{Circular $\boldsymbol{W}$}
\end{minipage}
\par\vspace{1.1em}
\begin{minipage}{0.8\linewidth}
\centering
\includegraphics[width=\linewidth]{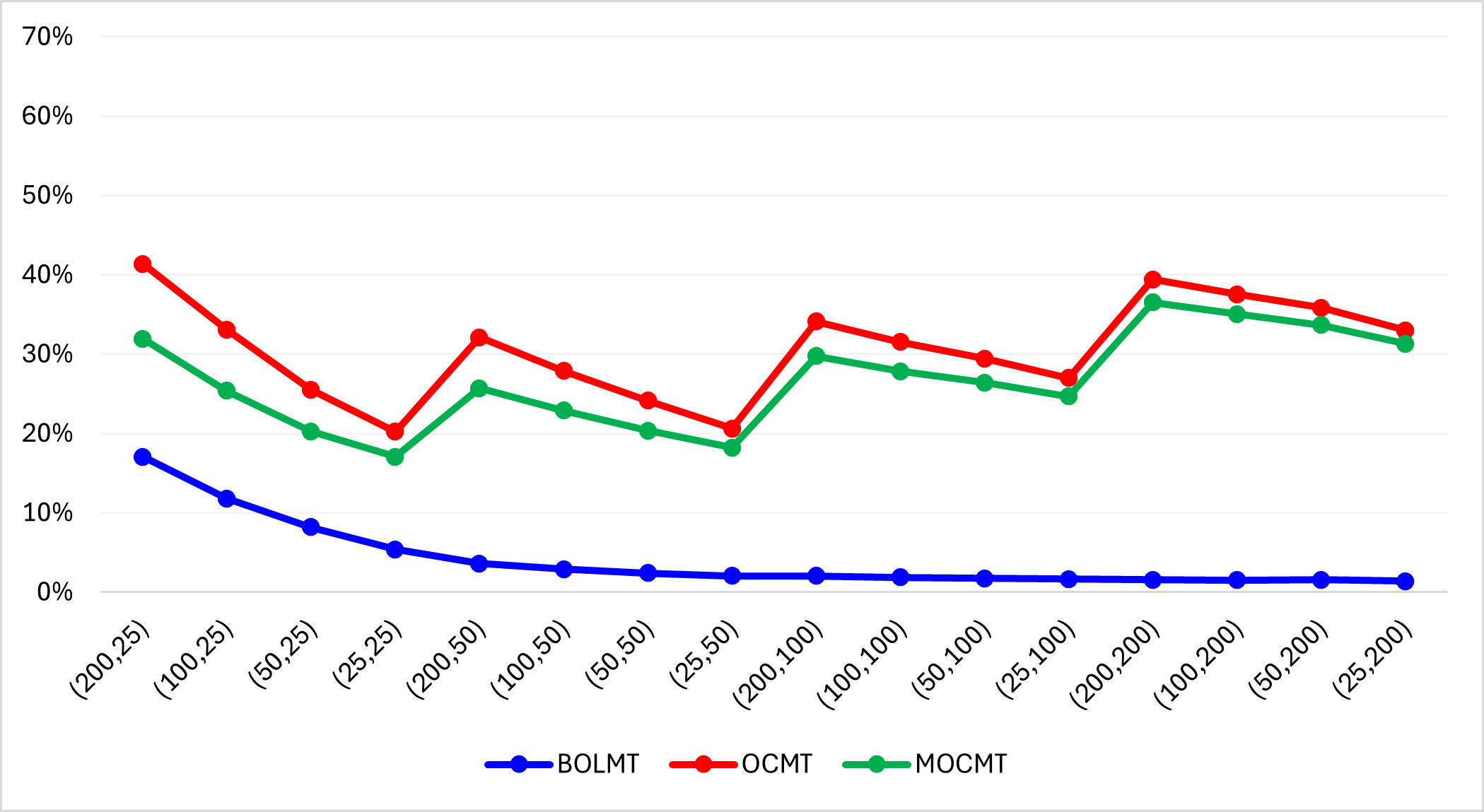}\\
\par\vspace{-0.5em}
\textbf{Block-Diagonal $\boldsymbol{W}$}
\end{minipage}
\par\vspace{1.1em}
\begin{minipage}{0.8\linewidth}
\centering
\includegraphics[width=\linewidth]{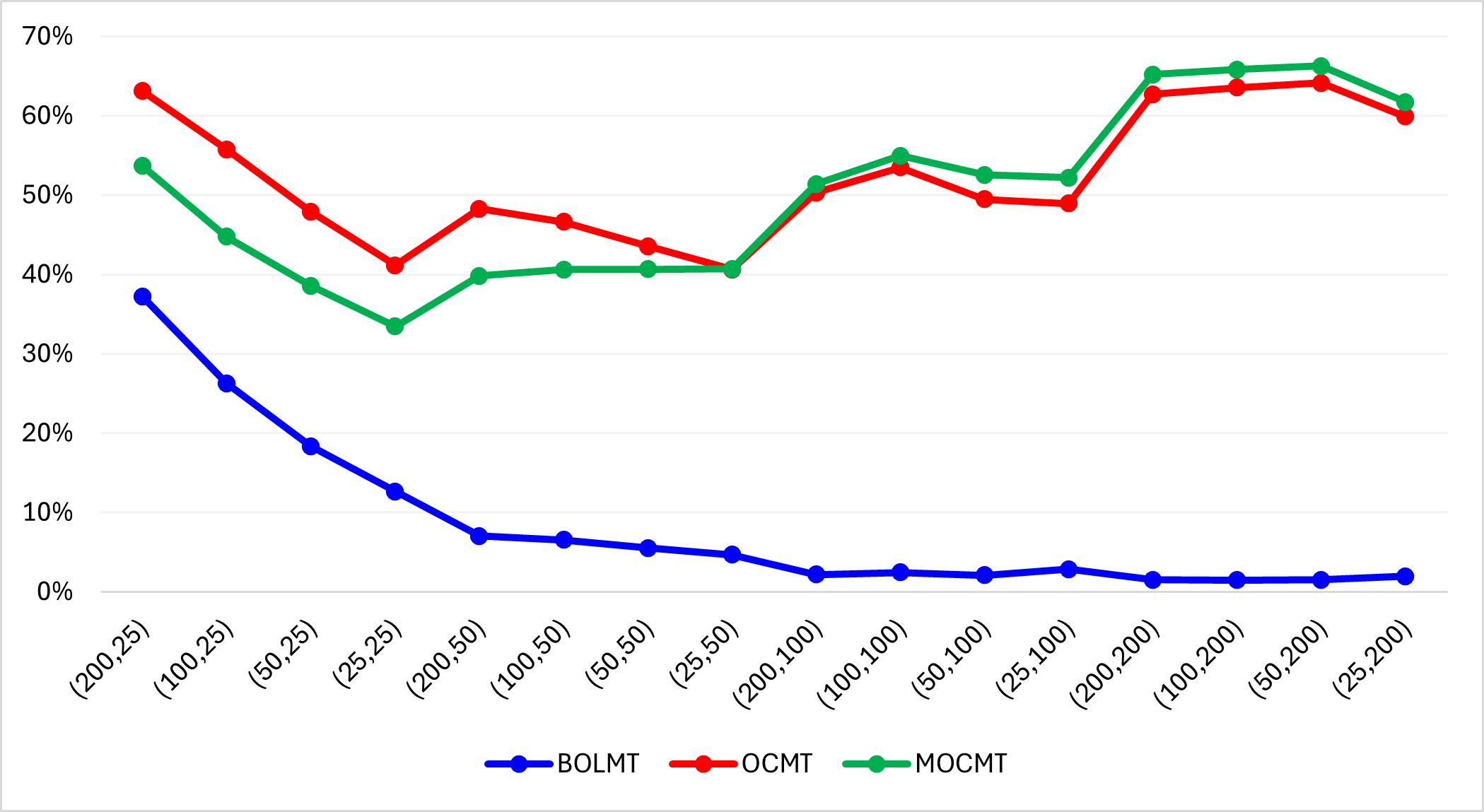}\\
\par\vspace{-0.5em}
\textbf{Random $\boldsymbol{W}$ with Dominant Unit}
\end{minipage}
\vspace{-0.4em}
\captionsetup{font=large}
\caption{FDR performance evaluation over different $(T,N)$ values}
\label{fig:FDR}
\end{figure}
\clearpage
\endgroup
\pagestyle{plain}

\begin{landscape}
\begin{table}[H]
\centering
\captionsetup{font=large}
\caption{Finite-sample performance of BOLMT for $\beta_1=2$}
\label{tab:beta1_BOLMT}

\renewcommand{\arraystretch}{1.1}
\setlength{\tabcolsep}{3.5pt}

\begin{tabular}{l ccccc|cccc|cccc|cccc}
\hline
& & \multicolumn{4}{c|}{\textbf{Mean}}
& \multicolumn{4}{c|}{\textbf{Bias ($\times$ 100)}}
& \multicolumn{4}{c|}{\textbf{Std. Dev.}}
& \multicolumn{4}{c}{\textbf{RMSE}} \\
\cline{3-18}

\textbf{Network} & \textbf{$N$} / \textbf{$T$}
& 25 & 50 & 100 & 200
& 25 & 50 & 100 & 200
& 25 & 50 & 100 & 200
& 25 & 50 & 100 & 200 \\
\hline

\multirow{4}{*}{Circular}
& 25
& 1.948 & 1.998 & 2.000 & 2.000
& 5.182 & 0.170 & 0.019 & 0.011
& 0.027 & 0.008 & 0.005 & 0.003
& 0.059 & 0.008 & 0.005 & 0.003 \\

& 50
& 1.923 & 1.996 & 2.000 & 2.000
& 7.711 & 0.369 & 0.035 & 0.013
& 0.023 & 0.007 & 0.003 & 0.002
& 0.080 & 0.008 & 0.003 & 0.002 \\

& 100
& 1.895 & 1.995 & 2.000 & 2.000
& 10.480 & 0.543 & 0.006 & 0.005
& 0.019 & 0.005 & 0.003 & 0.002
& 0.107 & 0.007 & 0.003 & 0.002 \\

& 200
& 1.858 & 1.992 & 2.000 & 2.000
& 14.189 & 0.805 & 0.014 & 0.000
& 0.015 & 0.004 & 0.002 & 0.001
& 0.143 & 0.009 & 0.002 & 0.001 \\

\hline

\multirow{4}{*}{Block}
& 25
& 2.001 & 2.000 & 2.000 & 2.000
& 0.119 & 0.007 & 0.008 & 0.001
& 0.012 & 0.003 & 0.002 & 0.001
& 0.012 & 0.003 & 0.002 & 0.001 \\

& 50
& 2.002 & 2.000 & 2.000 & 2.000
& 0.163 & 0.006 & 0.006 & 0.006
& 0.012 & 0.002 & 0.001 & 0.001
& 0.012 & 0.002 & 0.001 & 0.001 \\

& 100
& 2.005 & 2.000 & 2.000 & 2.000
& 0.478 & 0.004 & 0.001 & 0.003
& 0.011 & 0.001 & 0.001 & 0.001
& 0.012 & 0.001 & 0.001 & 0.001 \\

& 200
& 2.011 & 2.000 & 2.000 & 2.000
& 1.085 & 0.001 & 0.002 & 0.003
& 0.010 & 0.001 & 0.001 & 0.000
& 0.015 & 0.001 & 0.001 & 0.000 \\

\hline

\multirow{4}{*}{Random}
& 25
& 1.995 & 2.001 & 2.000 & 2.000
& 0.537 & 0.074 & 0.002 & 0.000
& 0.023 & 0.005 & 0.002 & 0.001
& 0.024 & 0.005 & 0.002 & 0.001 \\

& 50
& 2.001 & 1.999 & 2.000 & 2.000
& 0.056 & 0.063 & 0.011 & 0.002
& 0.020 & 0.004 & 0.001 & 0.001
& 0.020 & 0.004 & 0.001 & 0.001 \\

& 100
& 1.996 & 2.000 & 2.000 & 2.000
& 0.371 & 0.042 & 0.001 & 0.003
& 0.020 & 0.003 & 0.001 & 0.001
& 0.020 & 0.003 & 0.001 & 0.001 \\

& 200
& 1.995 & 2.000 & 2.000 & 2.000
& 0.548 & 0.028 & 0.007 & 0.001
& 0.016 & 0.003 & 0.001 & 0.000
& 0.017 & 0.003 & 0.001 & 0.000 \\

\hline
\end{tabular}

\vspace{0.2cm}
\footnotesize
\textit{Notes:} ``Block'' refers to a block-diagonal network, while ``Random''
denotes a random network with a dominant unit.

\end{table}
\end{landscape}

\begin{landscape}
\begin{table}[H]
\centering
\captionsetup{font=large}
\caption{Finite-sample performance of BOLMT for $\beta_2=-1$}
\label{tab:beta2_BOLMT}

\renewcommand{\arraystretch}{1.1}
\setlength{\tabcolsep}{3.5pt}

\begin{tabular}{l ccccc|cccc|cccc|cccc}
\hline
& & \multicolumn{4}{c|}{\textbf{Mean}}
& \multicolumn{4}{c|}{\textbf{Bias ($\times$ 100)}}
& \multicolumn{4}{c|}{\textbf{Std. Dev.}}
& \multicolumn{4}{c}{\textbf{RMSE}} \\
\cline{3-18}

\textbf{Network} & \textbf{$N$} / \textbf{$T$}
& 25 & 50 & 100 & 200
& 25 & 50 & 100 & 200
& 25 & 50 & 100 & 200
& 25 & 50 & 100 & 200 \\
\hline

\multirow{4}{*}{Circular}
& 25
& -0.975 & -0.999 & -1.000 & -1.000
& 2.514 & 0.064 & 0.004 & 0.004
& 0.021 & 0.006 & 0.004 & 0.003
& 0.033 & 0.006 & 0.004 & 0.003 \\

& 50
& -0.963 & -0.998 & -1.000 & -1.000
& 3.748 & 0.163 & 0.028 & 0.004
& 0.017 & 0.005 & 0.003 & 0.002
& 0.041 & 0.005 & 0.003 & 0.002 \\

& 100
& -0.947 & -0.997 & -1.000 & -1.000
& 5.259 & 0.285 & 0.001 & 0.003
& 0.015 & 0.003 & 0.002 & 0.001
& 0.055 & 0.004 & 0.002 & 0.001 \\

& 200
& -0.929 & -0.996 & -1.000 & -1.000
& 7.057 & 0.397 & 0.007 & 0.001
& 0.012 & 0.003 & 0.001 & 0.001
& 0.072 & 0.005 & 0.001 & 0.001 \\

\hline

\multirow{4}{*}{Block}
& 25
& -1.001 & -1.000 & -1.000 & -1.000
& 0.077 & 0.030 & 0.007 & 0.003
& 0.009 & 0.002 & 0.001 & 0.001
& 0.009 & 0.002 & 0.001 & 0.001 \\

& 50
& -1.001 & -1.000 & -1.000 & -1.000
& 0.110 & 0.002 & 0.005 & 0.001
& 0.009 & 0.002 & 0.001 & 0.001
& 0.009 & 0.002 & 0.001 & 0.001 \\

& 100
& -1.002 & -1.000 & -1.000 & -1.000
& 0.210 & 0.005 & 0.001 & 0.002
& 0.008 & 0.001 & 0.001 & 0.000
& 0.009 & 0.001 & 0.001 & 0.000 \\

& 200
& -1.005 & -1.000 & -1.000 & -1.000
& 0.515 & 0.011 & 0.001 & 0.001
& 0.008 & 0.001 & 0.000 & 0.000
& 0.010 & 0.001 & 0.000 & 0.000 \\

\hline

\multirow{4}{*}{Random}
& 25
& -0.997 & -1.000 & -1.000 & -1.000
& 0.288 & 0.039 & 0.019 & 0.004
& 0.022 & 0.004 & 0.002 & 0.001
& 0.022 & 0.004 & 0.002 & 0.001 \\

& 50
& -0.998 & -1.000 & -1.000 & -1.000
& 0.237 & 0.039 & 0.009 & 0.002
& 0.020 & 0.004 & 0.001 & 0.001
& 0.020 & 0.004 & 0.001 & 0.001 \\

& 100
& -0.998 & -1.000 & -1.000 & -1.000
& 0.205 & 0.018 & 0.005 & 0.004
& 0.019 & 0.003 & 0.001 & 0.001
& 0.019 & 0.003 & 0.001 & 0.001 \\

& 200
& -0.997 & -1.000 & -1.000 & -1.000
& 0.299 & 0.005 & 0.000 & 0.000
& 0.016 & 0.003 & 0.001 & 0.000
& 0.016 & 0.003 & 0.001 & 0.000 \\

\hline
\end{tabular}

\vspace{0.2cm}
\footnotesize
\textit{Notes:} ``Block'' refers to a block-diagonal network, while ``Random''
denotes a random network with a dominant unit.
\end{table}
\end{landscape}

\begin{landscape}
\begin{table}[H]
\centering
\captionsetup{font=large}
\caption{Finite-sample performance of BOLMT for $\rho_1=0.5$}
\label{tab:rho1_BOLMT}

\renewcommand{\arraystretch}{1.1}
\setlength{\tabcolsep}{3.5pt}

\begin{tabular}{l ccccc|cccc|cccc|cccc}
\hline
& & \multicolumn{4}{c|}{\textbf{Mean}}
& \multicolumn{4}{c|}{\textbf{Bias ($\times$ 100)}}
& \multicolumn{4}{c|}{\textbf{Std. Dev.}}
& \multicolumn{4}{c}{\textbf{RMSE}} \\
\cline{3-18}

\textbf{Network} & \textbf{$N$} / \textbf{$T$}
& 25 & 50 & 100 & 200
& 25 & 50 & 100 & 200
& 25 & 50 & 100 & 200
& 25 & 50 & 100 & 200 \\
\hline

\multirow{4}{*}{Circular}
& 25
& 0.482 & 0.501 & 0.501 & 0.501
& 1.791 & 0.084 & 0.101 & 0.061
& 0.014 & 0.004 & 0.002 & 0.002
& 0.022 & 0.004 & 0.003 & 0.002 \\

& 50
& 0.474 & 0.501 & 0.501 & 0.501
& 2.630 & 0.080 & 0.095 & 0.056
& 0.013 & 0.003 & 0.002 & 0.001
& 0.029 & 0.003 & 0.002 & 0.001 \\

& 100
& 0.462 & 0.501 & 0.501 & 0.501
& 3.811 & 0.096 & 0.113 & 0.065
& 0.012 & 0.002 & 0.001 & 0.001
& 0.040 & 0.002 & 0.002 & 0.001 \\

& 200
& 0.447 & 0.501 & 0.501 & 0.501
& 5.293 & 0.058 & 0.123 & 0.070
& 0.010 & 0.002 & 0.001 & 0.001
& 0.054 & 0.002 & 0.002 & 0.001 \\

\hline

\multirow{4}{*}{Block}
& 25
& 0.497 & 0.500 & 0.500 & 0.500
& 0.328 & 0.050 & 0.030 & 0.018
& 0.007 & 0.001 & 0.001 & 0.001
& 0.008 & 0.001 & 0.001 & 0.001 \\

& 50
& 0.493 & 0.501 & 0.500 & 0.500
& 0.709 & 0.052 & 0.030 & 0.015
& 0.008 & 0.001 & 0.001 & 0.000
& 0.011 & 0.001 & 0.001 & 0.000 \\

& 100
& 0.487 & 0.501 & 0.500 & 0.500
& 1.331 & 0.059 & 0.033 & 0.019
& 0.007 & 0.001 & 0.000 & 0.000
& 0.015 & 0.001 & 0.001 & 0.000 \\

& 200
& 0.476 & 0.501 & 0.500 & 0.500
& 2.403 & 0.064 & 0.039 & 0.020
& 0.007 & 0.001 & 0.000 & 0.000
& 0.025 & 0.001 & 0.000 & 0.000 \\

\hline

\multirow{4}{*}{Random}
& 25
& 0.458 & 0.499 & 0.501 & 0.500
& 4.192 & 0.073 & 0.054 & 0.037
& 0.021 & 0.005 & 0.001 & 0.001
& 0.047 & 0.005 & 0.001 & 0.001 \\

& 50
& 0.441 & 0.497 & 0.501 & 0.500
& 5.922 & 0.259 & 0.051 & 0.033
& 0.022 & 0.004 & 0.001 & 0.001
& 0.063 & 0.005 & 0.001 & 0.001 \\

& 100
& 0.422 & 0.496 & 0.501 & 0.500
& 7.849 & 0.443 & 0.055 & 0.033
& 0.022 & 0.004 & 0.001 & 0.000
& 0.082 & 0.006 & 0.001 & 0.001 \\

& 200
& 0.395 & 0.494 & 0.501 & 0.500
& 10.497 & 0.600 & 0.055 & 0.032
& 0.019 & 0.004 & 0.000 & 0.000
& 0.107 & 0.007 & 0.001 & 0.000 \\

\hline
\end{tabular}

\vspace{0.2cm}
\footnotesize
\textit{Notes:} ``Block'' refers to a block-diagonal network, while ``Random''
denotes a random network with a dominant unit.
\end{table}
\end{landscape}

\begin{landscape}
\begin{table}[H]
\centering
\captionsetup{font=large}
\caption{Finite-sample performance of BOLMT for $\rho_2=-0.4$}
\label{tab:rho2_BOLMT}

\renewcommand{\arraystretch}{1.1}
\setlength{\tabcolsep}{3.5pt}

\begin{tabular}{l ccccc|cccc|cccc|cccc}
\hline
& & \multicolumn{4}{c|}{\textbf{Mean}}
& \multicolumn{4}{c|}{\textbf{Bias ($\times$ 100)}}
& \multicolumn{4}{c|}{\textbf{Std. Dev.}}
& \multicolumn{4}{c}{\textbf{RMSE}} \\
\cline{3-18}

\textbf{Network} & \textbf{$N$} / \textbf{$T$}
& 25 & 50 & 100 & 200
& 25 & 50 & 100 & 200
& 25 & 50 & 100 & 200
& 25 & 50 & 100 & 200 \\
\hline

\multirow{4}{*}{Circular}
& 25
& -0.355 & -0.401 & -0.401 & -0.401
& 4.508 & 0.058 & 0.083 & 0.056
& 0.020 & 0.005 & 0.003 & 0.002
& 0.049 & 0.005 & 0.003 & 0.002 \\

& 50
& -0.336 & -0.399 & -0.401 & -0.401
& 6.407 & 0.100 & 0.089 & 0.052
& 0.017 & 0.004 & 0.002 & 0.002
& 0.066 & 0.004 & 0.002 & 0.002 \\

& 100
& -0.316 & -0.398 & -0.401 & -0.401
& 8.422 & 0.217 & 0.104 & 0.065
& 0.014 & 0.003 & 0.002 & 0.001
& 0.085 & 0.004 & 0.002 & 0.001 \\

& 200
& -0.293 & -0.396 & -0.401 & -0.401
& 10.654 & 0.365 & 0.118 & 0.070
& 0.011 & 0.002 & 0.001 & 0.001
& 0.107 & 0.004 & 0.002 & 0.001 \\

\hline

\multirow{4}{*}{Block}
& 25
& -0.397 & -0.400 & -0.400 & -0.400
& 0.314 & 0.046 & 0.023 & 0.018
& 0.008 & 0.002 & 0.001 & 0.001
& 0.009 & 0.002 & 0.001 & 0.001 \\

& 50
& -0.393 & -0.401 & -0.400 & -0.400
& 0.680 & 0.064 & 0.030 & 0.019
& 0.009 & 0.001 & 0.001 & 0.001
& 0.011 & 0.001 & 0.001 & 0.001 \\

& 100
& -0.387 & -0.401 & -0.400 & -0.400
& 1.309 & 0.058 & 0.033 & 0.019
& 0.008 & 0.001 & 0.001 & 0.000
& 0.015 & 0.001 & 0.001 & 0.000 \\

& 200
& -0.378 & -0.401 & -0.400 & -0.400
& 2.228 & 0.064 & 0.038 & 0.021
& 0.008 & 0.001 & 0.000 & 0.000
& 0.024 & 0.001 & 0.000 & 0.000 \\

\hline

\multirow{4}{*}{Random}
& 25
& -0.356 & -0.400 & -0.400 & -0.400
& 4.372 & 0.032 & 0.048 & 0.028
& 0.024 & 0.004 & 0.001 & 0.001
& 0.050 & 0.004 & 0.001 & 0.001 \\

& 50
& -0.339 & -0.397 & -0.401 & -0.400
& 6.057 & 0.277 & 0.056 & 0.029
& 0.023 & 0.005 & 0.001 & 0.001
& 0.065 & 0.006 & 0.001 & 0.001 \\

& 100
& -0.298 & -0.395 & -0.401 & -0.400
& 10.152 & 0.546 & 0.057 & 0.033
& 0.026 & 0.005 & 0.001 & 0.000
& 0.105 & 0.007 & 0.001 & 0.001 \\

& 200
& -0.271 & -0.393 & -0.401 & -0.400
& 12.851 & 0.724 & 0.053 & 0.032
& 0.028 & 0.004 & 0.001 & 0.000
& 0.132 & 0.008 & 0.001 & 0.000 \\

\hline
\end{tabular}

\vspace{0.2cm}
\footnotesize
\textit{Notes:} ``Block'' refers to a block-diagonal network, while ``Random''
denotes a random network with a dominant unit.
\end{table}
\end{landscape}

\begin{landscape}
\begin{table}[H]
\centering
\captionsetup{font=large}
\caption{Finite-sample performance of OCMT for $\beta_1=2$}
\label{tab:beta1_OCMT}

\renewcommand{\arraystretch}{1.1}
\setlength{\tabcolsep}{3.5pt}

\begin{tabular}{l ccccc|cccc|cccc|cccc}
\hline
& & \multicolumn{4}{c|}{\textbf{Mean}}
& \multicolumn{4}{c|}{\textbf{Bias ($\times$ 100)}}
& \multicolumn{4}{c|}{\textbf{Std. Dev.}}
& \multicolumn{4}{c}{\textbf{RMSE}} \\
\cline{3-18}

\textbf{Network} & \textbf{$N$} / \textbf{$T$}
& 25 & 50 & 100 & 200
& 25 & 50 & 100 & 200
& 25 & 50 & 100 & 200
& 25 & 50 & 100 & 200 \\
\hline

\multirow{4}{*}{Circular}
& 25
& 1.697 & 1.878 & 1.981 & 1.999
& 30.333 & 12.234 & 1.947 & 0.148
& 0.076 & 0.035 & 0.013 & 0.004
& 0.313 & 0.127 & 0.024 & 0.005 \\

& 50
& 1.649 & 1.848 & 1.970 & 1.999
& 35.149 & 15.225 & 2.963 & 0.148
& 0.062 & 0.028 & 0.011 & 0.003
& 0.357 & 0.155 & 0.032 & 0.003 \\

& 100
& 1.601 & 1.820 & 1.959 & 1.999
& 39.870 & 17.984 & 4.097 & 0.142
& 0.050 & 0.021 & 0.009 & 0.002
& 0.402 & 0.181 & 0.042 & 0.003 \\

& 200
& 1.566 & 1.797 & 1.946 & 1.998
& 43.373 & 20.311 & 5.449 & 0.202
& 0.038 & 0.015 & 0.008 & 0.002
& 0.435 & 0.204 & 0.055 & 0.003 \\

\hline

\multirow{4}{*}{Block}
& 25
& 2.009 & 2.029 & 2.004 & 1.999
& 0.942 & 2.949 & 0.387 & 0.112
& 0.085 & 0.034 & 0.020 & 0.008
& 0.086 & 0.045 & 0.020 & 0.008 \\

& 50
& 1.988 & 2.030 & 2.005 & 1.999
& 1.184 & 2.957 & 0.518 & 0.128
& 0.072 & 0.023 & 0.016 & 0.006
& 0.073 & 0.038 & 0.017 & 0.006 \\

& 100
& 1.968 & 2.035 & 2.005 & 1.998
& 3.234 & 3.533 & 0.462 & 0.175
& 0.055 & 0.019 & 0.012 & 0.005
& 0.064 & 0.040 & 0.012 & 0.005 \\

& 200
& 1.947 & 2.040 & 2.005 & 1.998
& 5.308 & 3.988 & 0.542 & 0.163
& 0.046 & 0.014 & 0.008 & 0.004
& 0.070 & 0.042 & 0.010 & 0.004 \\

\hline

\multirow{4}{*}{Random}
& 25
& 1.804 & 2.003 & 2.005 & 1.987
& 19.603 & 0.311 & 0.535 & 1.259
& 0.118 & 0.034 & 0.012 & 0.006
& 0.229 & 0.034 & 0.013 & 0.014 \\

& 50
& 1.792 & 1.978 & 1.997 & 1.996
& 20.766 & 2.157 & 0.331 & 0.425
& 0.099 & 0.027 & 0.007 & 0.004
& 0.230 & 0.034 & 0.008 & 0.006 \\

& 100
& 1.726 & 1.979 & 1.998 & 1.998
& 27.429 & 2.079 & 0.200 & 0.183
& 0.081 & 0.025 & 0.006 & 0.002
& 0.286 & 0.032 & 0.006 & 0.003 \\

& 200
& 1.690 & 1.973 & 1.998 & 1.998
& 30.994 & 2.655 & 0.205 & 0.163
& 0.072 & 0.019 & 0.004 & 0.002
& 0.318 & 0.032 & 0.005 & 0.002 \\

\hline
\end{tabular}

\vspace{0.2cm}
\footnotesize
\textit{Notes:} ``Block'' refers to a block-diagonal network, while ``Random''
denotes a random network with a dominant unit.
\end{table}
\end{landscape}

\begin{landscape}
\begin{table}[H]
\centering
\captionsetup{font=large}
\caption{Finite-sample performance of OCMT for $\beta_2=-1$}
\label{tab:beta2_OCMT}

\renewcommand{\arraystretch}{1.1}
\setlength{\tabcolsep}{3.5pt}

\begin{tabular}{l ccccc|cccc|cccc|cccc}
\hline
& & \multicolumn{4}{c|}{\textbf{Mean}}
& \multicolumn{4}{c|}{\textbf{Bias ($\times$ 100)}}
& \multicolumn{4}{c|}{\textbf{Std. Dev.}}
& \multicolumn{4}{c}{\textbf{RMSE}} \\
\cline{3-18}

\textbf{Network} & \textbf{$N$} / \textbf{$T$}
& 25 & 50 & 100 & 200
& 25 & 50 & 100 & 200
& 25 & 50 & 100 & 200
& 25 & 50 & 100 & 200 \\
\hline

\multirow{4}{*}{Circular}
& 25
& -0.848 & -0.940 & -0.990 & -0.999
& 15.204 & 6.028 & 0.990 & 0.080
& 0.047 & 0.022 & 0.009 & 0.003
& 0.159 & 0.064 & 0.013 & 0.003 \\

& 50
& -0.825 & -0.924 & -0.985 & -0.999
& 17.497 & 7.567 & 1.484 & 0.070
& 0.038 & 0.018 & 0.007 & 0.002
& 0.179 & 0.078 & 0.016 & 0.002 \\

& 100
& -0.799 & -0.910 & -0.980 & -0.999
& 20.076 & 9.031 & 2.046 & 0.068
& 0.029 & 0.013 & 0.005 & 0.002
& 0.203 & 0.091 & 0.021 & 0.002 \\

& 200
& -0.783 & -0.898 & -0.973 & -0.999
& 21.672 & 10.165 & 2.715 & 0.101
& 0.022 & 0.009 & 0.004 & 0.001
& 0.218 & 0.102 & 0.028 & 0.001 \\

\hline

\multirow{4}{*}{Block}
& 25
& -1.006 & -1.015 & -1.002 & -0.999
& 0.621 & 1.510 & 0.229 & 0.050
& 0.051 & 0.022 & 0.011 & 0.005
& 0.052 & 0.027 & 0.012 & 0.005 \\

& 50
& -0.995 & -1.015 & -1.002 & -0.999
& 0.543 & 1.517 & 0.248 & 0.060
& 0.042 & 0.016 & 0.009 & 0.004
& 0.042 & 0.022 & 0.010 & 0.004 \\

& 100
& -0.984 & -1.017 & -1.002 & -0.999
& 1.618 & 1.728 & 0.216 & 0.089
& 0.033 & 0.012 & 0.006 & 0.003
& 0.037 & 0.021 & 0.007 & 0.003 \\

& 200
& -0.974 & -1.020 & -1.003 & -0.999
& 2.600 & 2.028 & 0.272 & 0.082
& 0.026 & 0.009 & 0.005 & 0.002
& 0.037 & 0.022 & 0.006 & 0.002 \\

\hline

\multirow{4}{*}{Random}
& 25
& -0.901 & -1.001 & -1.004 & -0.994
& 9.870 & 0.077 & 0.364 & 0.628
& 0.073 & 0.028 & 0.011 & 0.004
& 0.123 & 0.028 & 0.011 & 0.008 \\

& 50
& -0.894 & -0.990 & -0.999 & -0.998
& 10.553 & 1.028 & 0.130 & 0.234
& 0.056 & 0.021 & 0.007 & 0.003
& 0.120 & 0.023 & 0.007 & 0.004 \\

& 100
& -0.863 & -0.990 & -0.999 & -0.999
& 13.718 & 1.039 & 0.110 & 0.101
& 0.045 & 0.018 & 0.005 & 0.002
& 0.144 & 0.020 & 0.006 & 0.003 \\

& 200
& -0.846 & -0.987 & -0.999 & -0.999
& 15.446 & 1.277 & 0.089 & 0.078
& 0.039 & 0.013 & 0.004 & 0.002
& 0.159 & 0.018 & 0.004 & 0.002 \\

\hline
\end{tabular}

\vspace{0.2cm}
\footnotesize
\textit{Notes:} ``Block'' refers to a block-diagonal network, while ``Random''
denotes a random network with a dominant unit.
\end{table}
\end{landscape}

\begin{landscape}
\begin{table}[H]
\centering
\captionsetup{font=large}
\caption{Finite-sample performance of OCMT for $\rho_1=0.5$}
\label{tab:rho1_OCMT}

\renewcommand{\arraystretch}{1.1}
\setlength{\tabcolsep}{3.5pt}

\begin{tabular}{l ccccc|cccc|cccc|cccc}
\hline
& & \multicolumn{4}{c|}{\textbf{Mean}}
& \multicolumn{4}{c|}{\textbf{Bias ($\times$ 100)}}
& \multicolumn{4}{c|}{\textbf{Std. Dev.}}
& \multicolumn{4}{c}{\textbf{RMSE}} \\
\cline{3-18}

\textbf{Network} & \textbf{$N$} / \textbf{$T$}
& 25 & 50 & 100 & 200
& 25 & 50 & 100 & 200
& 25 & 50 & 100 & 200
& 25 & 50 & 100 & 200 \\
\hline

\multirow{4}{*}{Circular}
& 25
& 0.378 & 0.458 & 0.502 & 0.509
& 12.162 & 4.188 & 0.163 & 0.921
& 0.032 & 0.014 & 0.006 & 0.003
& 0.126 & 0.044 & 0.006 & 0.010 \\

& 50
& 0.367 & 0.451 & 0.500 & 0.510
& 13.293 & 4.851 & 0.033 & 0.978
& 0.023 & 0.011 & 0.004 & 0.002
& 0.135 & 0.050 & 0.004 & 0.010 \\

& 100
& 0.354 & 0.445 & 0.497 & 0.510
& 14.561 & 5.503 & 0.322 & 0.979
& 0.017 & 0.008 & 0.003 & 0.001
& 0.147 & 0.056 & 0.005 & 0.010 \\

& 200
& 0.344 & 0.438 & 0.493 & 0.510
& 15.603 & 6.223 & 0.665 & 0.987
& 0.014 & 0.006 & 0.003 & 0.001
& 0.157 & 0.063 & 0.007 & 0.010 \\

\hline

\multirow{4}{*}{Block}
& 25
& 0.318 & 0.398 & 0.454 & 0.478
& 18.223 & 10.240 & 4.644 & 2.228
& 0.031 & 0.020 & 0.012 & 0.006
& 0.185 & 0.104 & 0.048 & 0.023 \\

& 50
& 0.305 & 0.393 & 0.448 & 0.476
& 19.451 & 10.741 & 5.189 & 2.355
& 0.025 & 0.015 & 0.009 & 0.004
& 0.196 & 0.108 & 0.053 & 0.024 \\

& 100
& 0.297 & 0.382 & 0.443 & 0.475
& 20.282 & 11.802 & 5.650 & 2.516
& 0.018 & 0.011 & 0.007 & 0.003
& 0.204 & 0.119 & 0.057 & 0.025 \\

& 200
& 0.292 & 0.374 & 0.438 & 0.473
& 20.774 & 12.633 & 6.210 & 2.734
& 0.013 & 0.009 & 0.005 & 0.002
& 0.208 & 0.127 & 0.062 & 0.027 \\

\hline

\multirow{4}{*}{Random}
& 25
& 0.272 & 0.356 & 0.447 & 0.490
& 22.828 & 14.362 & 5.252 & 0.988
& 0.050 & 0.036 & 0.020 & 0.010
& 0.234 & 0.148 & 0.056 & 0.014 \\

& 50
& 0.252 & 0.346 & 0.459 & 0.497
& 24.779 & 15.356 & 4.059 & 0.265
& 0.045 & 0.035 & 0.016 & 0.006
& 0.252 & 0.157 & 0.044 & 0.006 \\

& 100
& 0.251 & 0.333 & 0.453 & 0.497
& 24.903 & 16.692 & 4.711 & 0.318
& 0.047 & 0.034 & 0.015 & 0.005
& 0.253 & 0.170 & 0.050 & 0.006 \\

& 200
& 0.239 & 0.323 & 0.442 & 0.495
& 26.051 & 17.748 & 5.758 & 0.539
& 0.041 & 0.035 & 0.019 & 0.004
& 0.264 & 0.181 & 0.061 & 0.007 \\

\hline
\end{tabular}

\vspace{0.2cm}
\footnotesize
\textit{Notes:} ``Block'' refers to a block-diagonal network, while ``Random''
denotes a random network with a dominant unit.
\end{table}
\end{landscape}

\begin{landscape}
\begin{table}[H]
\centering
\captionsetup{font=large}
\caption{Finite-sample performance of OCMT for $\rho_2=-0.4$}
\label{tab:rho2_OCMT}

\renewcommand{\arraystretch}{1.1}
\setlength{\tabcolsep}{3.5pt}

\begin{tabular}{l ccccc|cccc|cccc|cccc}
\hline
& & \multicolumn{4}{c|}{\textbf{Mean}}
& \multicolumn{4}{c|}{\textbf{Bias ($\times$ 100)}}
& \multicolumn{4}{c|}{\textbf{Std. Dev.}}
& \multicolumn{4}{c}{\textbf{RMSE}} \\
\cline{3-18}

\textbf{Network} & \textbf{$N$} / \textbf{$T$}
& 25 & 50 & 100 & 200
& 25 & 50 & 100 & 200
& 25 & 50 & 100 & 200
& 25 & 50 & 100 & 200 \\
\hline

\multirow{4}{*}{Circular}
& 25
& -0.142 & -0.279 & -0.389 & -0.406
& 25.805 & 12.053 & 1.102 & 0.648
& 0.035 & 0.030 & 0.011 & 0.003
& 0.260 & 0.124 & 0.016 & 0.007 \\

& 50
& -0.126 & -0.252 & -0.381 & -0.408
& 27.362 & 14.811 & 1.869 & 0.790
& 0.025 & 0.023 & 0.010 & 0.002
& 0.275 & 0.150 & 0.021 & 0.008 \\

& 100
& -0.119 & -0.228 & -0.371 & -0.409
& 28.110 & 17.191 & 2.854 & 0.855
& 0.018 & 0.017 & 0.008 & 0.002
& 0.282 & 0.173 & 0.030 & 0.009 \\

& 200
& -0.120 & -0.209 & -0.360 & -0.408
& 27.975 & 19.116 & 4.011 & 0.847
& 0.012 & 0.012 & 0.007 & 0.001
& 0.280 & 0.192 & 0.041 & 0.009 \\

\hline

\multirow{4}{*}{Block}
& 25
& -0.215 & -0.299 & -0.344 & -0.368
& 18.495 & 10.051 & 5.633 & 3.215
& 0.039 & 0.021 & 0.013 & 0.007
& 0.189 & 0.103 & 0.058 & 0.033 \\

& 50
& -0.204 & -0.293 & -0.339 & -0.366
& 19.636 & 10.659 & 6.146 & 3.422
& 0.029 & 0.015 & 0.011 & 0.005
& 0.198 & 0.108 & 0.062 & 0.035 \\

& 100
& -0.196 & -0.285 & -0.333 & -0.363
& 20.366 & 11.475 & 6.724 & 3.667
& 0.021 & 0.011 & 0.007 & 0.004
& 0.205 & 0.115 & 0.068 & 0.037 \\

& 200
& -0.192 & -0.280 & -0.327 & -0.361
& 20.761 & 12.034 & 7.297 & 3.916
& 0.016 & 0.009 & 0.006 & 0.003
& 0.208 & 0.121 & 0.073 & 0.039 \\

\hline

\multirow{4}{*}{Random}
& 25
& -0.150 & -0.270 & -0.342 & -0.355
& 24.953 & 13.003 & 5.798 & 4.468
& 0.048 & 0.037 & 0.015 & 0.012
& 0.254 & 0.135 & 0.060 & 0.046 \\

& 50
& -0.174 & -0.244 & -0.319 & -0.361
& 22.623 & 15.575 & 8.060 & 3.911
& 0.043 & 0.032 & 0.013 & 0.007
& 0.230 & 0.159 & 0.082 & 0.040 \\

& 100
& -0.149 & -0.227 & -0.316 & -0.348
& 25.102 & 17.251 & 8.398 & 5.240
& 0.037 & 0.028 & 0.009 & 0.005
& 0.254 & 0.175 & 0.084 & 0.053 \\

& 200
& -0.150 & -0.231 & -0.312 & -0.340
& 25.046 & 16.932 & 8.767 & 5.951
& 0.031 & 0.027 & 0.007 & 0.004
& 0.252 & 0.171 & 0.088 & 0.060 \\

\hline
\end{tabular}

\vspace{0.2cm}
\footnotesize
\textit{Notes:} ``Block'' refers to a block-diagonal network, while ``Random''
denotes a random network with a dominant unit.
\end{table}
\end{landscape}

\begin{landscape}
\begin{table}[H]
\centering
\captionsetup{font=large}
\caption{Finite-sample performance of MOCMT for $\beta_1=2$}
\label{tab:beta1_MOCMT}

\renewcommand{\arraystretch}{1.1}
\setlength{\tabcolsep}{3.5pt}

\begin{tabular}{l ccccc|cccc|cccc|cccc}
\hline
& & \multicolumn{4}{c|}{\textbf{Mean}}
& \multicolumn{4}{c|}{\textbf{Bias ($\times$ 100)}}
& \multicolumn{4}{c|}{\textbf{Std. Dev.}}
& \multicolumn{4}{c}{\textbf{RMSE}} \\
\cline{3-18}

\textbf{Network} & \textbf{$N$} / \textbf{$T$}
& 25 & 50 & 100 & 200
& 25 & 50 & 100 & 200
& 25 & 50 & 100 & 200
& 25 & 50 & 100 & 200 \\
\hline

\multirow{4}{*}{Circular}
& 25
& 1.878 & 1.996 & 1.999 & 2.000
& 12.245 & 0.365 & 0.055 & 0.009
& 0.072 & 0.009 & 0.005 & 0.003
& 0.142 & 0.010 & 0.005 & 0.003 \\

& 50
& 1.818 & 1.992 & 1.999 & 2.000
& 18.236 & 0.787 & 0.075 & 0.041
& 0.061 & 0.009 & 0.003 & 0.003
& 0.192 & 0.012 & 0.004 & 0.003 \\

& 100
& 1.753 & 1.987 & 1.999 & 2.000
& 24.654 & 1.255 & 0.072 & 0.033
& 0.050 & 0.007 & 0.003 & 0.002
& 0.252 & 0.014 & 0.003 & 0.002 \\

& 200
& 1.699 & 1.980 & 1.999 & 2.000
& 30.114 & 1.985 & 0.104 & 0.038
& 0.041 & 0.006 & 0.002 & 0.001
& 0.304 & 0.021 & 0.002 & 0.001 \\

\hline

\multirow{4}{*}{Block}
& 25
& 1.965 & 2.000 & 2.000 & 2.000
& 3.527 & 0.009 & 0.007 & 0.003
& 0.059 & 0.003 & 0.002 & 0.001
& 0.068 & 0.003 & 0.002 & 0.001 \\

& 50
& 1.940 & 2.000 & 2.000 & 2.000
& 6.017 & 0.010 & 0.007 & 0.007
& 0.055 & 0.002 & 0.001 & 0.001
& 0.081 & 0.002 & 0.001 & 0.001 \\

& 100
& 1.921 & 2.000 & 2.000 & 2.000
& 7.906 & 0.020 & 0.002 & 0.004
& 0.045 & 0.003 & 0.001 & 0.001
& 0.091 & 0.003 & 0.001 & 0.001 \\

& 200
& 1.905 & 1.999 & 2.000 & 2.000
& 9.470 & 0.069 & 0.005 & 0.002
& 0.039 & 0.003 & 0.001 & 0.000
& 0.102 & 0.003 & 0.001 & 0.000 \\

\hline

\multirow{4}{*}{Random}
& 25
& 1.827 & 1.996 & 2.000 & 2.000
& 17.299 & 0.365 & 0.006 & 0.001
& 0.110 & 0.020 & 0.002 & 0.002
& 0.205 & 0.021 & 0.002 & 0.002 \\

& 50
& 1.793 & 1.990 & 2.000 & 2.000
& 20.677 & 1.036 & 0.012 & 0.008
& 0.095 & 0.020 & 0.002 & 0.001
& 0.227 & 0.022 & 0.002 & 0.001 \\

& 100
& 1.730 & 1.979 & 2.000 & 2.000
& 26.967 & 2.052 & 0.004 & 0.005
& 0.080 & 0.021 & 0.001 & 0.001
& 0.281 & 0.029 & 0.001 & 0.001 \\

& 200
& 1.692 & 1.975 & 2.000 & 2.000
& 30.770 & 2.530 & 0.008 & 0.002
& 0.070 & 0.016 & 0.001 & 0.000
& 0.316 & 0.030 & 0.001 & 0.000 \\

\hline
\end{tabular}

\vspace{0.2cm}
\footnotesize
\textit{Notes:} ``Block'' refers to a block-diagonal network, while ``Random''
denotes a random network with a dominant unit.
\end{table}
\end{landscape}

\begin{landscape}
\begin{table}[H]
\centering
\captionsetup{font=large}
\caption{Finite-sample performance of MOCMT for $\beta_2=-1$}
\label{tab:beta2_MOCMT}

\renewcommand{\arraystretch}{1.1}
\setlength{\tabcolsep}{3.5pt}

\begin{tabular}{l ccccc|cccc|cccc|cccc}
\hline
& & \multicolumn{4}{c|}{\textbf{Mean}}
& \multicolumn{4}{c|}{\textbf{Bias ($\times$ 100)}}
& \multicolumn{4}{c|}{\textbf{Std. Dev.}}
& \multicolumn{4}{c}{\textbf{RMSE}} \\
\cline{3-18}

\textbf{Network} & \textbf{$N$} / \textbf{$T$}
& 25 & 50 & 100 & 200
& 25 & 50 & 100 & 200
& 25 & 50 & 100 & 200
& 25 & 50 & 100 & 200 \\
\hline

\multirow{4}{*}{Circular}
& 25
& -0.938 & -0.998 & -1.000 & -1.000
& 6.185 & 0.183 & 0.024 & 0.013
& 0.039 & 0.007 & 0.004 & 0.003
& 0.073 & 0.007 & 0.004 & 0.003 \\

& 50
& -0.911 & -0.996 & -1.000 & -1.000
& 8.933 & 0.381 & 0.047 & 0.017
& 0.032 & 0.005 & 0.003 & 0.002
& 0.095 & 0.007 & 0.003 & 0.002 \\

& 100
& -0.877 & -0.994 & -1.000 & -1.000
& 12.336 & 0.647 & 0.033 & 0.014
& 0.027 & 0.004 & 0.002 & 0.001
& 0.126 & 0.008 & 0.002 & 0.001 \\

& 200
& -0.850 & -0.990 & -0.999 & -1.000
& 15.018 & 0.990 & 0.050 & 0.018
& 0.022 & 0.004 & 0.001 & 0.001
& 0.152 & 0.011 & 0.001 & 0.001 \\

\hline

\multirow{4}{*}{Block}
& 25
& -0.982 & -1.000 & -1.000 & -1.000
& 1.756 & 0.029 & 0.008 & 0.004
& 0.030 & 0.002 & 0.001 & 0.001
& 0.034 & 0.002 & 0.001 & 0.001 \\

& 50
& -0.970 & -1.000 & -1.000 & -1.000
& 2.974 & 0.001 & 0.004 & 0.001
& 0.028 & 0.002 & 0.001 & 0.001
& 0.041 & 0.002 & 0.001 & 0.001 \\

& 100
& -0.961 & -1.000 & -1.000 & -1.000
& 3.946 & 0.013 & 0.000 & 0.002
& 0.023 & 0.002 & 0.001 & 0.000
& 0.046 & 0.002 & 0.001 & 0.000 \\

& 200
& -0.952 & -1.000 & -1.000 & -1.000
& 4.772 & 0.026 & 0.001 & 0.001
& 0.020 & 0.002 & 0.000 & 0.000
& 0.052 & 0.002 & 0.000 & 0.000 \\

\hline

\multirow{4}{*}{Random}
& 25
& -0.914 & -0.998 & -1.000 & -1.000
& 8.573 & 0.192 & 0.017 & 0.004
& 0.057 & 0.011 & 0.002 & 0.001
& 0.103 & 0.011 & 0.002 & 0.001 \\

& 50
& -0.895 & -0.995 & -1.000 & -1.000
& 10.455 & 0.515 & 0.010 & 0.005
& 0.050 & 0.010 & 0.002 & 0.001
& 0.116 & 0.012 & 0.002 & 0.001 \\

& 100
& -0.865 & -0.990 & -1.000 & -1.000
& 13.538 & 0.998 & 0.006 & 0.007
& 0.043 & 0.011 & 0.001 & 0.001
& 0.142 & 0.015 & 0.001 & 0.001 \\

& 200
& -0.847 & -0.987 & -1.000 & -1.000
& 15.321 & 1.264 & 0.001 & 0.001
& 0.037 & 0.008 & 0.001 & 0.000
& 0.158 & 0.015 & 0.001 & 0.000 \\

\hline
\end{tabular}

\vspace{0.2cm}
\footnotesize
\textit{Notes:} ``Block'' refers to a block-diagonal network, while ``Random''
denotes a random network with a dominant unit.
\end{table}
\end{landscape}

\begin{landscape}
\begin{table}[H]
\centering
\captionsetup{font=large}
\caption{Finite-sample performance of MOCMT for $\rho_1=0.5$}
\label{tab:rho1_MOCMT}

\renewcommand{\arraystretch}{1.1}
\setlength{\tabcolsep}{3.5pt}

\begin{tabular}{l ccccc|cccc|cccc|cccc}
\hline
& & \multicolumn{4}{c|}{\textbf{Mean}}
& \multicolumn{4}{c|}{\textbf{Bias ($\times$ 100)}}
& \multicolumn{4}{c|}{\textbf{Std. Dev.}}
& \multicolumn{4}{c}{\textbf{RMSE}} \\
\cline{3-18}

\textbf{Network} & \textbf{$N$} / \textbf{$T$}
& 25 & 50 & 100 & 200
& 25 & 50 & 100 & 200
& 25 & 50 & 100 & 200
& 25 & 50 & 100 & 200 \\
\hline

\multirow{4}{*}{Circular}
& 25
& 0.468 & 0.511 & 0.512 & 0.510
& 3.210 & 1.075 & 1.173 & 1.039
& 0.025 & 0.005 & 0.003 & 0.002
& 0.041 & 0.012 & 0.012 & 0.011 \\

& 50
& 0.447 & 0.509 & 0.512 & 0.511
& 5.330 & 0.939 & 1.164 & 1.062
& 0.020 & 0.004 & 0.002 & 0.002
& 0.057 & 0.010 & 0.012 & 0.011 \\

& 100
& 0.424 & 0.508 & 0.512 & 0.511
& 7.593 & 0.821 & 1.156 & 1.057
& 0.016 & 0.003 & 0.002 & 0.001
& 0.078 & 0.009 & 0.012 & 0.011 \\

& 200
& 0.403 & 0.506 & 0.512 & 0.511
& 9.673 & 0.601 & 1.159 & 1.065
& 0.014 & 0.003 & 0.001 & 0.001
& 0.098 & 0.007 & 0.012 & 0.011 \\

\hline

\multirow{4}{*}{Block}
& 25
& 0.487 & 0.502 & 0.502 & 0.502
& 1.254 & 0.184 & 0.186 & 0.178
& 0.017 & 0.002 & 0.001 & 0.001
& 0.022 & 0.002 & 0.002 & 0.002 \\

& 50
& 0.475 & 0.502 & 0.502 & 0.502
& 2.462 & 0.195 & 0.185 & 0.185
& 0.016 & 0.001 & 0.001 & 0.001
& 0.029 & 0.002 & 0.002 & 0.002 \\

& 100
& 0.462 & 0.502 & 0.502 & 0.502
& 3.751 & 0.202 & 0.196 & 0.195
& 0.014 & 0.001 & 0.001 & 0.000
& 0.040 & 0.002 & 0.002 & 0.002 \\

& 200
& 0.447 & 0.502 & 0.502 & 0.502
& 5.344 & 0.208 & 0.205 & 0.201
& 0.013 & 0.001 & 0.000 & 0.000
& 0.055 & 0.002 & 0.002 & 0.002 \\

\hline

\multirow{4}{*}{Random}
& 25
& 0.411 & 0.504 & 0.509 & 0.510
& 8.914 & 0.432 & 0.879 & 0.984
& 0.036 & 0.009 & 0.002 & 0.001
& 0.096 & 0.010 & 0.009 & 0.010 \\

& 50
& 0.374 & 0.499 & 0.509 & 0.511
& 12.602 & 0.076 & 0.933 & 1.135
& 0.033 & 0.008 & 0.002 & 0.001
& 0.130 & 0.008 & 0.010 & 0.011 \\

& 100
& 0.342 & 0.490 & 0.509 & 0.510
& 15.766 & 1.018 & 0.942 & 1.046
& 0.033 & 0.009 & 0.002 & 0.001
& 0.161 & 0.014 & 0.010 & 0.011 \\

& 200
& 0.310 & 0.484 & 0.508 & 0.509
& 18.956 & 1.598 & 0.769 & 0.949
& 0.031 & 0.008 & 0.001 & 0.001
& 0.192 & 0.018 & 0.008 & 0.010 \\

\hline
\end{tabular}

\vspace{0.2cm}
\footnotesize
\textit{Notes:} ``Block'' refers to a block-diagonal network, while ``Random''
denotes a random network with a dominant unit.
\end{table}
\end{landscape}

\begin{landscape}
\begin{table}[H]
\centering
\captionsetup{font=large}
\caption{Finite-sample performance of MOCMT for $\rho_2=-0.4$}
\label{tab:rho2_MOCMT}

\renewcommand{\arraystretch}{1.1}
\setlength{\tabcolsep}{3.5pt}

\begin{tabular}{l ccccc|cccc|cccc|cccc}
\hline
& & \multicolumn{4}{c|}{\textbf{Mean}}
& \multicolumn{4}{c|}{\textbf{Bias ($\times$ 100)}}
& \multicolumn{4}{c|}{\textbf{Std. Dev.}}
& \multicolumn{4}{c}{\textbf{RMSE}} \\
\cline{3-18}

\textbf{Network} & \textbf{$N$} / \textbf{$T$}
& 25 & 50 & 100 & 200
& 25 & 50 & 100 & 200
& 25 & 50 & 100 & 200
& 25 & 50 & 100 & 200 \\
\hline

\multirow{4}{*}{Circular}
& 25
& -0.341 & -0.409 & -0.411 & -0.410
& 5.899 & 0.876 & 1.085 & 0.988
& 0.026 & 0.007 & 0.004 & 0.003
& 0.064 & 0.011 & 0.011 & 0.010 \\

& 50
& -0.310 & -0.405 & -0.411 & -0.410
& 9.002 & 0.498 & 1.089 & 1.004
& 0.021 & 0.005 & 0.003 & 0.002
& 0.092 & 0.007 & 0.011 & 0.010 \\

& 100
& -0.281 & -0.402 & -0.411 & -0.410
& 11.888 & 0.172 & 1.090 & 1.023
& 0.017 & 0.004 & 0.002 & 0.001
& 0.120 & 0.005 & 0.011 & 0.010 \\

& 200
& -0.258 & -0.397 & -0.411 & -0.410
& 14.170 & 0.346 & 1.097 & 1.032
& 0.013 & 0.004 & 0.001 & 0.001
& 0.142 & 0.005 & 0.011 & 0.010 \\

\hline

\multirow{4}{*}{Block}
& 25
& -0.391 & -0.402 & -0.402 & -0.402
& 0.859 & 0.174 & 0.176 & 0.178
& 0.015 & 0.002 & 0.001 & 0.001
& 0.017 & 0.003 & 0.002 & 0.002 \\

& 50
& -0.381 & -0.402 & -0.402 & -0.402
& 1.852 & 0.200 & 0.182 & 0.189
& 0.014 & 0.001 & 0.001 & 0.001
& 0.023 & 0.002 & 0.002 & 0.002 \\

& 100
& -0.372 & -0.402 & -0.402 & -0.402
& 2.825 & 0.199 & 0.190 & 0.191
& 0.012 & 0.001 & 0.001 & 0.000
& 0.031 & 0.002 & 0.002 & 0.002 \\

& 200
& -0.360 & -0.402 & -0.402 & -0.402
& 4.007 & 0.209 & 0.202 & 0.200
& 0.012 & 0.001 & 0.000 & 0.000
& 0.042 & 0.002 & 0.002 & 0.002 \\

\hline

\multirow{4}{*}{Random}
& 25
& -0.315 & -0.405 & -0.409 & -0.410
& 8.538 & 0.459 & 0.881 & 0.978
& 0.034 & 0.008 & 0.002 & 0.002
& 0.092 & 0.009 & 0.009 & 0.010 \\

& 50
& -0.287 & -0.399 & -0.410 & -0.411
& 11.269 & 0.122 & 0.953 & 1.132
& 0.029 & 0.008 & 0.002 & 0.001
& 0.116 & 0.008 & 0.010 & 0.011 \\

& 100
& -0.242 & -0.390 & -0.409 & -0.410
& 15.786 & 0.974 & 0.944 & 1.045
& 0.028 & 0.008 & 0.002 & 0.001
& 0.160 & 0.012 & 0.010 & 0.010 \\

& 200
& -0.215 & -0.385 & -0.408 & -0.409
& 18.503 & 1.487 & 0.754 & 0.938
& 0.027 & 0.007 & 0.001 & 0.001
& 0.187 & 0.016 & 0.008 & 0.009 \\

\hline
\end{tabular}

\vspace{0.2cm}
\footnotesize
\textit{Notes:} ``Block'' refers to a block-diagonal network, while ``Random''
denotes a random network with a dominant unit.
\end{table}
\end{landscape}

\end{document}